\documentclass[letterpaper,11pt,onecolumn,oneside]{article}

\usepackage{amsmath}
\usepackage{amsthm}
\usepackage{dsfont}
\usepackage{graphicx}
\usepackage{nicefrac}
\usepackage{tabularx}
\usepackage[boxed,noline,noend,ruled,linesnumbered]{algorithm2e}
\usepackage{caption}
\usepackage{paralist}
\usepackage{thm-restate}
\usepackage{multicol}
\usepackage{etoolbox}
\usepackage{comment}

\allowdisplaybreaks

\usepackage[UKenglish]{babel}
\usepackage[top=1in,bottom=1in,left=1in,right=1in]{geometry}
\usepackage[colorlinks=true, allcolors=blue]{hyperref}
\usepackage{tikz}
\usepackage[capitalize,noabbrev]{cleveref}
\newtheorem{theorem}{Theorem}

\newtheorem{lemma}{Lemma}
\newtheorem{corollary}{Corollary}

\newtheorem{definition}{Definition}

\author{Matthias Bentert \and Fedor V. Fomin \and Petr A. Golovach}
\date{}

\newcommand{\peb}{\ensuremath{\mathcal{M}}}

\newcommand{\pname}{\textsc}
\newcommand{\ProblemFormat}[1]{\pname{#1}}
\newcommand{\ProblemIndex}[1]{\index{problem!\ProblemFormat{#1}}}
\newcommand{\ProblemName}[1]{\ProblemFormat{#1}\ProblemIndex{#1}{}}

\usetikzlibrary{graphs, shapes, arrows.meta}

\newcommand{\probVC}{\ProblemName{Vertex Cover}}

\newcommand{\probFVS}{\ProblemName{Feedback Vertex Set}}

\newcommand{\probDS}{\ProblemName{Dominating Set}}

\newcommand{\probKPath}{\ProblemName{Longest Path}}

\newcommand{\Q}{\ProblemName{$\pi$-Request}}

\newcommand{\wpvc}{\ProblemName{Weighted Partial Vertex Cover}}

\newcommand{\im}{\ProblemName{Maximum-Weight Rooted Parallel Induced Minor}}

\newcommand{\minor}{\ProblemName{Minimum-Weight Rooted Simple Minor}}
\newcommand{\minorr}{\ProblemName{Minimum-Weight Rooted Simple Minor Request}}
\newcommand{\wsp}{\ProblemName{Steiner Partition}}
\newcommand{\wis}{\ProblemName{Maximum-Weight Independent Set}}
\newcommand{\wim}{\ProblemName{Maximum-Weight Induced Matching}}
\newcommand{\wif}{\ProblemName{Maximum-Weight Induced Forest}}
\newcommand{\wifr}{\ProblemName{Maximum-Weight Induced Forest Request}}
\newcommand{\wisi}{\ProblemName{Maximum-Weight Induced Subgraph Isomorphism}}
\newcommand{\wsi}{\ProblemName{Minimum-Weight Subgraph Isomorphism}}
\newcommand{\fragment}{\ProblemName{$k$-Subgraph Scoring}}
\newcommand{\fragmentrequest}{\ProblemName{$k$-Subgraph Scoring Request}}
\newcommand{\den}{\ProblemName{Densest $k$-Subgraph}}
\newcommand{\knkcut}{\ProblemName{Maximum~$(k,n-k)$-Cut}}

\newcommand{\NP}{$\mathsf{NP}$}

\newcommand{\Oh}{\ensuremath{\mathcal{O}}}

\DeclareMathOperator{\inte}{in}
\DeclareMathOperator{\ex}{ex}
\DeclareMathOperator{\he}{he}
\DeclareMathOperator{\li}{li}
\DeclareMathOperator{\di}{di}

\DeclareMathOperator{\res}{res}

\DeclareMathOperator{\SQ}{\texttt{solve-$\pi$-r}}
\DeclareMathOperator{\clean}{\texttt{red-portals}}
\DeclareMathOperator{\base}{\texttt{base-case}}

\crefname{algocf}{Algorithm}{Algorithms}
\crefname{algocfline}{Line}{Lines}

\newcommand{\probdef}[3]{
\medskip
\begin{minipage}{\textwidth}
    \medskip
      \normalsize\textsc{#1} \smallskip \\
      \begin{tabularx}{.94\textwidth}{@{ }l@{\hspace{5pt}}X}
        \normalsize\textbf{Input:}    & \normalsize#2 \\
        \normalsize\textbf{Question:} & \normalsize#3
      \end{tabularx}
\end{minipage}
}

\newcommand{\problemdef}[3]{
\medskip
\begin{minipage}{\textwidth}
    \medskip
      \normalsize\textsc{#1} \smallskip \\
      \begin{tabularx}{0.95\textwidth}{@{}l@{\hspace{3pt}}X}
        \normalsize\textbf{Input:}    & \normalsize#2 \\
        \normalsize\textbf{Task:} & \normalsize#3
      \end{tabularx}
\end{minipage}
}

\title{A Framework for Parameterized Subexponential-Subcubic-Time Algorithms for Weighted Problems in Planar Graphs}

\begin{document}
\maketitle
\begin{abstract}
    Many problems are known to be solvable in subexponential parameterized time when the input graph is planar.
    The bidimensionality framework of Demaine, Fomin, Hajiaghay, and Thilikos~[JACM~'05] and the treewidth-pattern-covering approach by Fomin, Lokshtanov, Marx, Pilipczuk, Pilipczuk, and Saurabh~[SICOMP~'22] give robust tools for designing such algorithms.
    However, there are still many problems for which we do not know whether subexponential parameterized algorithms exist.
    The bidimensionality framework is not able to handle weights or directed graphs and the treewidth-pattern-covering approach only works for finding connected solutions.
    Building on a result by Nederlof~[STOC~'20], we provide a framework that is able to solve a variety of problems in planar graphs in subexponential parameterized time for which this was previously not known (where the polynomial part of the running time is usually~$\Oh(n^{2.49})$).
    Our framework can handle weights, does not require solutions to contain only few connected components, and applies to cases where the number of potential patterns of a solution is exponential in the parameter.

    We then use the framework to show that various weighted problems like \wpvc, \wif, \minor, and \im{} allow for subexponential parameterized algorithms.
    This was previously not known for any of them.
    Moreover, we present a very easy-to-use fragment of our framework that is still powerful enough to be applicable to problems like \den, \wis, \knkcut, and \wpvc.
    This fragment allows for significantly simpler proofs in the case of \wis{} and \knkcut{} and is able to show a subexponential parameterized algorithm for weighted versions of \den.
    Even the unweighted version was not known before and is stated as an open problem in the existing literature.
\end{abstract}

\newpage
\tableofcontents
\newpage

\section{Introduction}
\label{sec:intro}

Over the past two decades, extensive research in parameterized complexity has focused on subexponential algorithms for planar graphs and, more broadly, on various classes of sparse graphs. While for most parameterized \NP-hard problems the best we can hope for under the Exponential Time Hypothesis (ETH) of Impagliazzo and Paturi~\cite{ImpagliazzoP01} is an exponential dependence on the parameter, there are nevertheless many problems that admit subexponential parameterized algorithms when restricted to planar graphs. This phenomenon was first demonstrated in 2000~\cite{AlberBFKN02}, where the authors gave a subexponential parameterized algorithm for deciding whether a given~$n$-vertex planar graph contains a dominating set of size $k$ in $2^{\mathcal{O}(\sqrt{k})} n^{\mathcal{O}(1)}$ time.

In 2005, Demaine, Fomin, Hajiaghayi, and Thilikos~\cite{DemaineFHT05jacm} introduced the \emph{bidimensionality} framework, providing elegant tools for designing subexponential parameterized algorithms for problems such as \probVC, \probFVS, \probDS, or \probKPath. Building on the grid minor theorem of Robertson, Seymour, and Thomas~\cite{RobSeymT94}, this framework shows that for a planar graph $G$ and an integer $k$, either $G$ is a trivial yes- or no-instance of the problem at hand, or the treewidth of $G$ is in $\mathcal{O}(\sqrt{k})$. In the latter case, applying standard dynamic programming techniques on graphs of bounded treewidth yields a running time of $2^{\mathcal{O}(\sqrt{k})} n^{\mathcal{O}(1)}$, where $n$ is the number of vertices. In many situations, the bidimensionality framework can be extended to subexponential parameterized algorithms on $H$-minor-free graphs. Moreover, for a large class of problems, the running time provided by bidimensionality is essentially optimal under the ETH: one cannot expect a $2^{o(\sqrt{k})} n^{\mathcal{O}(1)}$-time algorithm.
For an overview of bidimensionality and its applications, we refer to the standard text book on parameterized complexity~\cite{cygan2015parameterized}.

However, the bidimensionality framework does not apply to directed graphs or to weighted problems, because the existence of large grid minors alone does not address these settings. A step toward overcoming this limitation is the \emph{treewidth-pattern-covering} approach~\cite{FominLMPPS22}. Central to this method is a procedure that randomly samples a vertex subset~$S$ in polynomial time with the following properties. The induced subgraph~$G[S]$ has treewidth~$\mathcal{O}(\sqrt{k}\log k)$ and for every \emph{connected} subgraph~$P$ of $G$ on at most $k$ vertices, the probability that $S$ contains all vertices of~$P$ is at least~$\nicefrac{1}{2^{\mathcal{O}(\sqrt{k}\log^2 k)} n^{\mathcal{O}(1)}}$.
By applying dynamic programming techniques on $G[S]$, this yields a randomized algorithms running in~$2^{\mathcal{O}(\sqrt{k}\log^2 k)} n^{\mathcal{O}(1)}$~time for problems such as \textsc{Directed $k$-Path}, \textsc{Weighted $k$-Path}, or \textsc{Subgraph Isomorphism} (when the sought subgraph is connected) for planar input graphs.

A crucial requirement of the treewidth-pattern-covering approach is that the pattern $P$ must be connected. To handle cases where~$P$ is disconnected, Nederlof~\cite{Ned20} extended the treewidth-pattern-covering approach to find or counts (induced) copies of a $k$-vertex pattern~$P$ in a planar graph in $2^{\tilde{\mathcal{O}}(\sqrt{k})} \mathrm{poly}(\sigma(P)n)$ time, where $\sigma(P)$ denotes the number of non-isomorphic separations of~$P$ of size~$\tilde{\mathcal{O}}(\sqrt{k})$.\footnote{A separation of a graph is a pair~$(X,Y)$ of sets of vertices such that~$X \cup Y$ contains all vertices and there are no edges between~$X \setminus Y$ and~$Y \setminus X$. The size of separation~$(X,Y)$ is~$|X \cap Y|$ and two separations~$(X,Y)$ and~$(X',Y')$ are isomorphic if~$G[X]$ is isomorphic to~$G[X']$ and~$X \cap Y = X' \cap Y'$. See \cref{sec:prelim} for formal definitions.} In particular, Nederlof’s results yield an algorithm with running time~$2^{\tilde{\mathcal{O}}(\sqrt{k})} n^{\mathcal{O}(1)}$ when~$P$ is a matching, an independent set, or a connected bounded-degree graph. For any pattern~$P$, he obtained algorithms running in $2^{\mathcal{O}(k/\log k)} n^{\mathcal{O}(1)}$ time, which are essentially tight under the ETH.

None of the above frameworks provides subexponential parameterized algorithms for problems on \emph{weighted} graphs when the solution might be disconnected.
The bidimensionality framework fails as the existence of a large grid minor does not reveal any information about a weighted solution.
In particular, a solution for e.g. \wis{} may completely avoid the grid minor if the vertices in the grid minor are too light.
The treewidth-pattern-covering approach fails because the solution cannot be covered by a small number of connected patterns of size~$\mathcal{O}(k)$, and Nederlof's approach focuses on counting problems and not on weighted problems.
The additional use of efficient inclusion-exclusion prohibits the approach to generalize to weighted problems.
In fact, subexponential parameterized problems for weighted problems that might have solutions with many connected components seem to be very rare.
We are only aware of such algorithms for \wis~\cite{masterthesis} and \textsc{Weighted Vertex Cover}~\cite{LPSXZ25}.
This naturally raises the following question:
\begin{quote}
Is there a general framework for subexponential parameterized algorithms that find a disconnected $k$-vertex pattern $P$ of maximum/minimum weight?
\end{quote}
In this paper, we develop generic algorithmic tools that allow us not only to address this question, but also to tackle related problems involving (rooted/induced) minors.
The approach is deterministic and can also handle directed planar graphs.

\subsection{Our Results}
Our contribution is threefold.
First, we introduce a framework to generate subexponential parameterized algorithms (usually running in~$2^{\tilde{\Oh}(\sqrt{k})}n^{2.49}$~time) for problems on weighted planar graphs that search for a set of~$k$ vertices with a specified property.
The framework also applies to cases where the number of potential patterns of a solution is exponential in the parameter.
The framework is introduced in \cref{sec:framework} and requires some definitions which we set up in \cref{sec:prelim}.
Second, we show how to apply this framework to prove that the following problem \fragment{} can be solved in subexponential parameterized time when parameterized by~$k$.
This problem generalizes problems like \wis, \wpvc, \den, \knkcut, and more.

\probdef{\fragment}
{A directed graph~$G=(V,A)$, a vertex-cost function~$c \colon V \rightarrow \mathds{N}$, a weight function~$w \colon V \cup A \rightarrow \mathds{R}$, an integer~$k$, and values~$\alpha_1,\alpha_2,\alpha_3,\beta,W \in \mathds{R}$.}
{Is there a set~$S$ of cost exactly~$k$, that is, $\sum_{v\in S}c(v) = k$ and \[\hspace*{-2cm}\alpha_1 \big(\sum_{\substack{(u,v) \in A\\u,v \in S}} w((u,v))\big) + \alpha_2 \big(\sum_{\substack{(u,v) \in A\\u \in S, v \notin S}} w((u,v))\big) + \alpha_3 \big(\sum_{\substack{(u,v) \in A\\u \notin S,v \in S}} w((u,v))\big) + \beta \big(\sum_{v \in S} w(S)\big) \geq W?\]}

Note that when~$c(v) = 1$ for each vertex~$v$, then~$k$ is the number of vertices in the solution.
Third, we show more involved applications of the framework in the form of \wif{} and different notions of minors.
All applications of our framework are shown in \cref{sec:applications} and summarized in the following (see \cref{sec:prelim} for problem definitions).

\begin{theorem}
    The following problems can be solved in $2^{\tilde{\Oh}(\sqrt{k})}n^{2.49}$ time when parameterized by the size of the solution:
    \begin{itemize}
        \item \fragment
        \item \den
        \item \wsp
        \item \wpvc
        \item \wim
        \item \wif
    \end{itemize}
    The following problems can be solved in~$2^{\Oh(\nicefrac{k}{\log(k)})}n^{2.49}$ time when parameterized by the size of the model (or the size of the pattern graph in the case of \textsc{Subgraph Isomorphism}):
    \begin{itemize}
        \item \minor
        \item \im
        \item \wsi
        \item \wisi
    \end{itemize}
\end{theorem}

Before we give an overview of how the framework works, let us be explicit about our perspective here.
The main goal of our work is to improve the theoretical understanding of subexponential phenomena, and we focus exclusively on this aspect.
We do not attempt to optimize our results from the perspective of algorithm engineering, and reducing the hidden constant in the exponent is an interesting challenge beyond the scope of this work.
Moreover, the framework is technically involved and considerably less accessible than, for example, bidimensionality.
This is perhaps unsurprising, since problems seem to become substantially more complicated as soon as one moves to weighted or directed graphs.
Nonetheless, this raises the natural question of whether a simpler approach---perhaps based on a logic fragment such as MSO---could achieve similar results.
However, pursuing such a direction is hindered by the delicate boundary between problems that admit subexponential parameterized algorithms on planar graphs and those that do not.
In particular, while the above theorem shows that \wsp{} can be solved in subexponential parameterized time, even slight generalizations may break this property. For instance, the closely related \textsc{Steiner Forest} problem---a variant where different terminal sets are not necessarily contained in distinct connected components---does not admit a subexponential parameterized algorithm even on unweighted planar graphs, as shown by Pilipczuk, Pilipczuk, Sankowski, and van Leeuwen~\cite{PilipczukPSL14}.

\subsection{Our Methods}
We next give an overview of how the main results in this work are achieved.
We start with a high-level overview of the general strategy using \textsc{Independent Set} as an illustrative example and show some potential pitfalls.
This general strategy (avoiding the pitfalls) follows the algorithm by Nederlof~\cite{Ned20}.
In a second step, we give a few more technical details regarding the algorithm and highlight the key distinctions between our algorithm and the one by Nederlof.
Then, we show how to apply this general framework to different problems.
The application for \fragment{} will be quite simple and for \wif{} and different variants of graph minors, we will set up increasingly complicated notions.

We now give a high-level overview of the general strategy.
As a starting point, consider the classic divide-and-conquer strategy pioneered by Lipton and Tarjan in the context of planar separators~\cite{LiptonT80}.
Given a planar graph~$G=(V,E)$, an integer~$k$, and a balanced (cycle) separator~$C$ in~$G$, we \emph{guess}\footnote{Whenever we say we ``guess'' something, we actually iterate over all possibilities, but for presentation and proof, we focus on the iteration that leads to an optimal solution.} which vertices of $C$ appear in the solution.
We then remove $C$ and all neighbors of those chosen vertices from the graph, guess how many vertices of the solution lie in the interior and exterior of~$C$, and solve the resulting subproblems recursively.
Once each subproblem becomes sufficiently small, we solve them by applying a brute-force procedure.
Unfortunately, the general idea sketched above does not achieve the desired running time for three main reasons.

First, if we guess which vertices of $C$ belong to a solution and then make separate recursive calls for each guess, the running time takes the form $f(k)^{\log n}$ for some (super-polynomial) function~$f$.
Even setting $f(k) = k$ yields a running time of $k^{\log n}$, which is not contained in $2^{\tilde{\mathcal{O}}(\sqrt{k})} \mathrm{poly}(n)$.
To overcome this, we only branch on the separator~$C$ and solve a more general problem on the interior and exterior, where we output for \emph{any} set~$X$ of vertices in~$C$ the largest independent set containing exactly the vertices in~$X$ of~$C$.

The second issue is that the separator~$C$ could intersect the optimal solution in too many vertices. 
This leads to too many possible choices resulting in an exponential running time in~$k$.
We circumvent this by using a result due to Nederlof \cite{Ned20} (see \cref{lem:unknownbalance}) that roughly allows us for any (even unknown) weight function to replace~$C$ by a set~$\mathcal{S}$ of up to~$2^{\Oh(\log^2(k))}$ balanced cycle separators such that at least one of them is balanced for the weight function and does not intersect the solution in too many vertices.
However, since~$2^{\Oh(\log^2(k))}>k$, we again run into the problem that we cannot recursively solve all of them as this leads to a running time of~$\Oh(k^{\log n})$, which is too slow.
This time, we circumvent the issue by observing that if~$C$ contains many vertices from the solution (more than~$20 \sqrt{k} \log(k)$), then it is sufficiently balanced for the weight function that assigns a constant weight to each vertex in the (unknown) solution and zero to all other vertices.
Hence, at least one cycle separator~$C' \in \mathcal{S}$ is also balanced for this weight function and does not intersect the solution in too many vertices.
This allows us to reduce the parameter~$k$ by~$\Oh(\sqrt{k})$ in these cases, which can be shown to yield a sufficiently small number of recursive calls.

Finally, the third issue relates to both of the above.
Consider a recursive call for some separator~$C$ and for the sake of argument let us assume we are considering the interior of~$C$.
We then find a separator~$C'$ in the interior of~$C$ and branch into the interior and exterior of~$C'$.
Note that the exterior of~$C'$ now has two boundaries: both the vertices in~$C$ and in~$C'$.
Even if each of these two cycle separators are small enough and contains few enough vertices of an optimal solution, together they can be too large and/or contain too many vertices of an optimal solutiuon.
Technically, this does not happen at recursion depth one, but \cref{fig:branch} shows an example where one subgraph is incident to many cycle separators.
To overcome this issue, we again use the lemma by Nederlof to find/guess a cycle separator that contains few vertices of an optimal solution and is balanced for the weight function that assigns a constant weight to all vertices in previous cycle separators that the current subgraph is incident to and zero to all other vertices.
The blue line in \cref{fig:branch} shows an example of such a cycle separator.
We then branch on this new cycle separator and show that the number of possible guesses is again small enough.
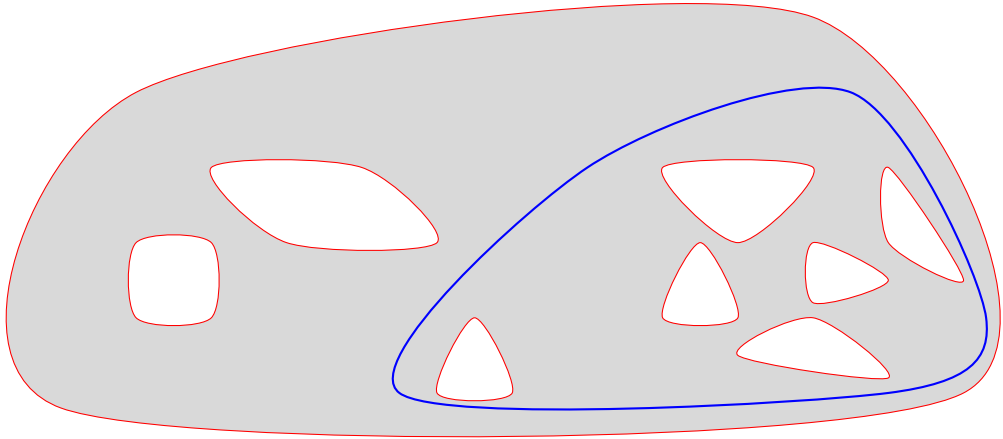
\begin{figure}
    \centering
    \begin{tikzpicture}
        \draw[red,fill=gray!30] plot[smooth cycle] coordinates {(0, -.2) (1, 4) (10, 5) (12, 0)};
        \draw[red,fill=white] plot[smooth cycle] coordinates {(3, 2) (2, 3) (4, 3) (5, 2)};
        \draw[red,fill=white] plot[smooth cycle] coordinates {(1, 1) (1, 2) (2, 2) (2, 1)};
        \draw[red,fill=white] plot[smooth cycle] coordinates {(9, .5) (10, 1) (11, .2)};
        \draw[red,fill=white] plot[smooth cycle] coordinates {(8, 1) (8.5, 2) (9, 1)};
        \draw[red,fill=white] plot[smooth cycle] coordinates {(11, 2) (11, 3) (12, 1.5)};
        \draw[red,fill=white] plot[smooth cycle] coordinates {(9, 2) (8, 3) (10, 3)};
        \draw[red,fill=white] plot[smooth cycle] coordinates {(10, 2) (10, 1.2) (11, 1.5)};
        \draw[red,fill=white] plot[smooth cycle] coordinates {(5, 0) (6, 0) (5.5, 1)};
        \draw[blue,thick] plot[smooth cycle] coordinates {(4.5, 0) (11, 0) (12.3, 1) (10.5,4) (7,3)};
    \end{tikzpicture}
    \caption{Illustration of a graph cut open by a set of cycle separators (red). The blue line shows a cycle separator that separates the red vertices in a balanced way (assuming that the number of vertices is roughly proportional to the length of the respective line).}
    \label{fig:branch}
\end{figure}

So overall, the general strategy looks as follows.
We want to solve the following generalization of \textsc{Independent Set}: Given a planar graph~$G$, a set~$B$ of vertices (the union of some cycle separators), and an integer~$k$, output a lookup table that reports for each sufficiently small set~$Z \subseteq B$ whether an independent set~$\mathcal{I}$ of size (at most)~$k$ exists in~$G$ such that~$\mathcal{I} \cap B = Z$.
Initially~$B = \emptyset$.
If the graph has some constant size, then solve it using brute force and otherwise, compute a balanced cycle separator~$C$.
Guess if an optimal solution intersects~$C$ in many or few vertices.
If it is in many, then find/guess a different cycle separator~$C'$ using the lemma by Nederlof and reduce the parameter~$k$ accordingly.
If it is few, then consider~$C' = C$ and do not reduce the parameter.
Let~$G_{\inte}$ and~$G_{\ex}$ be the graphs induced by the interior and exterior of~$C'$, respectively, and let~$B'_{\inte}$ and~$B'_{\ex}$ be the subsets of~$B \cup C'$ in either graph.
If~$B'_{\inte}$ and~$B'_{\ex}$ are small enough, then recursively solve the two subinstances and do the following.
For each sufficiently small set~$Z$, compute the optimal solution for~$Z$ by finding the optimal solutions~$Z_{\inte},Z_{\ex}$ for both subinstances where~$Z_{\inte} \cap C = Z$ and~$Z_{\ex} \cap C = Z$.
The union of~$Z_{\inte}$ and~$Z_{\ex}$ is then an optimal solution for~$Z$ (and the vertices in~$Z$ are counted twice, so an optimal solution has size~$|Z_{\inte}| + |Z_{\ex}| - |Z|$).
If~$B'_{\inte}$ and/or~$B'_{\ex}$ are too large, then find a separator that splits the set in a balanced way and solve the two subinstances instead.

We mention that the above is oversimplified in two aspects.
First, the lemma by Nederlof does not return a small cycle separator which only intersects any solution in sufficiently few vertices.
Rather, it returns a cycle separator and a 3-partition of its vertices into heavy, light, and discarded vertices.
The number of heavy vertices is small but a solution might contain any number of heavy vertices.
The number of light vertices is larger (but still bounded) and there is a guarantee that the solution only intersects it in sufficiently few vertices so that we can still afford to iterate over all subsets of the respective size in~$2^{\tilde{\Oh}(\sqrt{k})}$ time.
Lastly, the number of discarded vertices is unbounded, but we are guaranteed that the solution does not contain any discarded vertices.
The second oversimplification is the brute-force step in case the graph is sufficiently small.
In order to make progress when reducing the parameter, we also need an algorithm for the case where the parameter is a constant.

Let us now give a few more details on our main algorithm and highlight some of the details where our approach differs from that of Nederlof.
First and of least importance, Nederlof does not analyze the polynomial part of the running time.
In particular, the use of a black-box algorithm from the treewidth-pattern-covering approach~\cite{FominLMPPS22} incurs as stated a running time of~$\Oh(n^4)$ in each recursive call.
We overcome this issue by replacing costly flow and dual-LP computations in their algorithm by recently discovered near-linear-time algorithms~\cite{ChenKLPGS22}.
We then analyze the running time of both the algorithm from the treewidth-pattern-covering approach~\cite{FominLMPPS22} and the one by Nederlof including the polynomial part of the running time (for the first time).
Moreover, the base-case algorithm when the parameter becomes small enough is different in our algorithm.
While both perform dynamic programming over certain tree decompositions, we use balanced tree decompositions in order to improve the polynomial part of the running time.
Second, even ignoring the running time per recursive call, the analysis by Nederlof bounds the number of recursive calls by~$\Oh(n^4)$.
Fortunately, this issue can be resolved by a more careful analysis.
Unfortunately, we are not aware of a way to show it using statements of the form ``there exists a sufficiently large constant~$c$ such that $\ldots$''.
Instead, we carefully choose several constants in order for the analysis to work.
The third issue stems from the nature of our undertaking.
Since we wish to provide a general framework rather than solving a specific problem, some steps that can easily be shown for a specific problem like \textsc{Subgraph Isomorphism} require careful consideration in our setting.
Two examples of this are the definition of the problem variant we wish to solve and the algorithm for the base case when either the parameter or the number of vertices is sufficiently small.
Recall that for \textsc{Independent Set}, we wanted to solve a more general version where we return a lookup table.
In order for our framework to be applicable as often as possible, we do not give a strict recipe of how to generalize any particular problem, but instead give a list of minimal requirements that the problem has to fulfill in order for our framework to work.
For the algorithm for the base case when the parameter is constant, we show how to perform dynamic programming over a tree decomposition using only the minimal requirements.

Finally, we give an overview of the different applications of our framework.
We start by showing a general recipe of how to apply it and a subexponential parameterized algorithm for \fragment{} is a first simple-to-show corollary.
This one is simple as the way in which two partial solutions can be combined along a separator is quite trivial.
It is easy to see that \fragment{} generalizes \wis, \wpvc, \den, \knkcut, as well as weighted variants of the latter two problems and thus, we get subexponential parameterized algorithms for all of them.

We then turn our attention to \wif.
Here, we show how to encode connectivity and being acyclic using partitions and the join of two sets.

For \minor, the main new obstacles to overcome are the exponential number of possible models for a single node of~$H$ and how to encode information about which vertex in the separator plays which role (is part of the model of which node or edge of~$H$).
We deal with the latter issue in a similar way as Nederlof did for \textsc{Subgraph Isomorphism}, that is, we store for which subgraph of~$H$ we already have a model and also a separator of some fixed size for this subgraph to the rest of~$H$.
Since we assume that the current cycle separators only contain relatively few vertices from the solution and the separator in~$H$ also has a fixed size, we can show that the number of possible ways to assign the nodes in the separator to any sufficiently small subset of the vertices in the current cycle separator is small enough so that we can enumerate them all and find corresponding solutions in subexponential parameterized time.
For the former issue, we do not store the model for~$H$ but only which vertices in the model are already connected (plus information about which incident edges in~$H$ are already part of the model) and again rely on our encoding for connectivity for \wif.

Finally, we consider parallel minors (where the pattern graph~$H$ can have parallel edges and one wishes for a minor model with at least (or exactly in the case of induced minors) that many connections within the model).
Here, the main problem is that planar multigraphs can have an exponential number of non-isomorphic separations and hence the previous analysis cannot yield subexpontial running times.
We overcome this issue by defining a new parameter for the analysis and show that if this parameter is exponential, then no planar (simple) graph with~$k$ vertices can have~$H$ as a minor.
Thus, we have a simple win-win approach.
Either the parameter is large and we can immediately answer no (the parameter will be easy to compute) or the parameter is small (subexponential) and we can adapt the analysis to show subexponential parameterized algorithms.

\section{Preliminaries}
\label{sec:prelim}
In this section, we first give an overview of the notation and basic concepts we use.
Then, we give all formal problem definitions.
In a third step, we formalize \textbf{\emph{request variants}} of problems.
Recall that for \textsc{Independent Set}, we want to solve a more general version where we output a lookup table that answers for each set~$Z$ of certain size (and certain interaction with heavy, light, and discarded vertices of the cycle separators), whether an independent set of a certain size exists where exactly the vertices in~$Z$ are contained from the cycle separators.
This is the request variant for \textsc{Independent Set} and we give a definition (plus some intuition) of what type of generalized problem we can solve using our framework.
Finally, we give an overview of the different results from the known literature that our framework relies on.

We start with some notation and assume familiarity with standard notation from graph theory and theoretical computer science.
In particular, we denote by~$n$ and~$m$ the number of vertices and edges in a graph, respectively.
We use~$\mathds{N}$ to denote the set of \emph{\textbf{positive integers}} and~$\tilde{\Oh}$ to hide polylogarithmic factors in the Bachmann–Landau (big-O) notation.
For an integer~$\ell$, we use~$\mathbf{[\ell]}$ to denote the set~$\{1,2,\ldots, \ell\}$.
The base of all logarithmic functions in this work is~$2$ and we assume the word RAM model of computation (which means that operations like computing sums or minima over a constant number of real numbers take constant time independent of the size of the numbers).
Given a graph~$G=(V,E)$ and a vertex set~$V' \subseteq V$, we denote the graph \textbf{\emph{induced}} by~$V'$ by~\emph{$\mathbf{G[V']}$} and the graph obtained by deleting all vertices in~$V'$ (and all edges incident to those vertices) by~$\mathbf{G-V'}$.
For a set~$F \subseteq E$ of edges, we use~$\mathbf{G - F}$ as a shorthand for~$(V, E \setminus F)$.
Given a set~$W$ (not necessarily a subset of~$V$), we denote by~$W[G] = W \cap V$ the intersection of~$W$ and the vertex set of~$G$.

An embedding of a planar graph is 1-outerplanar if it is outerplanar, that is, all vertices lie on the outer face.
For $d\geq 2$, an embedding of a planar graph is~$d$-outerplanar if removing all vertices on the outer face results in a~$(d-1)$-outerplanar embedding.
A graph is~$\mathbf{d}$\emph{\textbf{-outerplanar}} if it has a~$d$-outerplanar embedding.
A planar graph is \textbf{\emph{triangulated}} if each face is bounded by three edges.
A \textbf{\emph{separation}} of $G$ is a pair~$(X,Y)$ of vertex subsets such that~$X \cup Y = V$ and there are no edges between vertices in~$X\setminus Y$ and~$Y \setminus X$.
The set~$X \cap Y$ is called the \emph{\textbf{separator}} and slightly overloading notation, we will sometimes also refer to some graph with vertex set~$X \cap Y$ as the separator.
The size of a separation~$(X,Y)$ is the number of vertices in its separator.
Two separations~$(X,Y)$ and~$(X',Y')$ are \emph{\textbf{isomorphic}} if there is an isomorphism~$f \colon V \rightarrow V$ such that~$\bigcup_{v \in X}f(v) = X'$ and~$f(v) = v$ for each~$v \in X \cap Y$.
In a planar graph together with an embedding, (the vertex set of) each simple cycle~$C$ is a separator that separates the strict \textbf{\emph{interior}} (the graph enclosed by~$C$) from the strict \textbf{\emph{exterior}}.
We therefore prefer to use the term cycle separator rather than simple cycle.
The union of the strict interior/exterior of~$C$ and~$C$ is denoted by the (non-strict) interior/exterior of~$C$.
For a cycle separator~$C$, we denote the interior by~$C_{\inte}$ and the exterior by~$C_{\ex}$.

A \emph{\textbf{tree decomposition}} of a graph~$G = (V,E)$ is a rooted tree~$\mathcal{T}=(U,A)$ in which each node~$u$ represents a subset~$X_u$ of vertices called the \emph{bag} of~$u$ such that

\begin{compactitem}
    \item each vertex in~$V$ is contained in at least one bag,
    \item for each edge~$\{v,w\} \in E$, there is a node~$u \in U$ with~$v,w \in X_u$, and
    \item for each~$v \in V$ it holds that all nodes~$u$ of~$\mathcal{T}$ with~$v \in X_u$ form a connected subtree of~$\mathcal{T}$.
\end{compactitem}
The \emph{\textbf{width}} of a tree decomposition is~$\max_{u \in U} |X_u| -1$ and the \textbf{\emph{treewidth}} of~$G$ is the minimum width of all possible tree decompositions of~$G$.
The depth of a tree decomposition is the length of a longest path from the root to any other node and a tree decomposition is binary if each node has at most two children.
Given a tree decomposition~$\mathcal{T} = (U,A)$ of a graph~$G$ and a node~$u \in U$, we denote by~$Y_u$ the union of all bags of~$u$ and descendants of~$u$ in~$\mathcal{T}$.
Note that for the root~$p$ of~$\mathcal{T}$ it holds that~$Y_p = V$. 

A \emph{\textbf{contraction}} of an edge~$\{u,v\}$ in a graph removes~$u$ and~$v$ from the graph, creates a new vertex~$uv$ and replaces each other edge~$\{x,y\}$ with~$x \in \{u,v\}$ with an edge~$\{uv,y\}$.
Depending on the context, one sometimes removes parallel edges and self loops and other times allows the resulting graph to be a multigraph.
To show the versatility of our framework, we will consider both versions and say that a contraction is simple if it removes parallel edges and self loops and it is a parallel contraction otherwise.
A (multi)graph~$H=(U,F)$ is a \emph{\textbf{simple/parallel minor}} of a simple graph~$G=(V,E)$ if there is a subgraph~$G'$ of~$G$ such that simple/parallel contractions of~$G'$ result in a graph isomorphic to~$H$.
We say that~$H$ is an \emph{\textbf{induced simple/parallel minor}} if there is an induced subgraph~$G'$ of~$G$ that can be contracted to obtain~$H$.
In all cases, the graph~$G'$ is called the \emph{\textbf{model}} of~$H$ and we also study rooted versions.
Here, additionally to the (multi)graph~$H$, we are given a function~$\eta \colon U \rightarrow 2^V$ that describes for each node\footnote{We call the vertices in the pattern graph~$H$ nodes in order to avoid confusion with the vertices in the graph~$G$.} of~$H$ a set of vertices of~$G$ that have to be contained in the desired model of~$H$.
More formally, for any~$u \in U$, all vertices in~$\eta(u)$ have to be contracted into a single vertex~$u'$ such that there is an isomorphism from the resulting graph to~$H$ mapping~$u'$ to~$u$.

A \emph{\textbf{proper weight assignment}}~$w$ is an assignment of non-negative weights at most~$\frac{1}{4}$ to vertices summing to~$1$.
For~$q \in [0,1]$, a \emph{$\mathbf{q}$\textbf{-balanced cycle separator}} for a proper weight assignment~$w$ is a cycle~$C$ such that the sum of weights of all vertices in the strict interior of~$C$ is at most~$q$ and the sum of weights of all vertices in the strict exterior of~$C$ is at most~$q$.

We say a pair~$(C,x)$ consisting of a set~$C$ of vertices and partitioned into heavy vertices~$C_{\he}$, light vertices~$C_{\li}$, and discarded vertices~$C_{\di}$ and an integer~$x$ is \emph{\textbf{almost measured}} with respect to two integers~$k$ and $\alpha$ if $x \leq \max(200 \log^7(k),1)$, ${|C_{\he}| \leq 40 x \alpha \sqrt{k}}$, and~${|C_{\li}| \leq 60 x \alpha k \log(k)}$.
It is \emph{\textbf{measured}} if it is almost measured and~$x \leq \max(40 \log^7(k),1)$ holds.
A \emph{\textbf{pebble set}} for~$(C,k,\alpha,x)$ is a set~$\peb \subseteq C_{\he} \cup C_{\li}$ with~$|\peb \cap C_{\li}| \leq 20x\alpha \sqrt{k} \log(k)$.
Intuitively, pebble sets describe the possible sets of vertices in a collection of cycle separators for which we wish to store a solution in a lookup table.

\subsection{Problem Definitions}
Here, we define the problems mentioned in the introduction.
We start with different variants of finding minors in a graph.
We deliberately chose to use both simple and parallel contractions, minors and induced minors, minimization and maximization variants to illustrate the versatility of our framework.
The other possible combinations can be solved similarly, but for the sake of conciseness, we do not explicitly study them.

\probdef{\minor}
{A graph~$G=(V,E)$, an edge-weight function~$w \colon E \rightarrow \mathds{R}^+$, a graph~$H=(U,F)$, a function~$\eta \colon U \rightarrow 2^V$, an integer~$k$, and a value~$W \in \mathds{R}$.}
{Is there a subgraph~$G'$ of~$G$ with at most~$k$ vertices and of total weight at most~$W$ such that~$G'$ is a model of~$H$, that is, it contains all vertices in~$\bigcup_{u \in U}\eta(u)$ and~$H$ is a rooted simple minor of~$G'$?}

\probdef{\im}
{A graph~$G=(V,E)$, an edge-weight function~$w \colon E \rightarrow \mathds{R}^+$, a multigraph~${H=(U,F)}$, a function~$\eta \colon U \rightarrow 2^V$, an integer~$k$, and a value~$W \in \mathds{R}$.}
{Is there an induced subgraph~$G'$ of~$G$ with at most~$k$ vertices and of total weight at least~$W$ such that~$G'$ contains all vertices in~$\bigcup_{u \in U}\eta(u)$ and~$H$ is a rooted parallel induced minor of~$G'$?}

The next problem we study is \wif.

\probdef{\wif}
{A graph~$G=(V,E)$, a vertex-weight function~$w \colon V \rightarrow \mathds{R}$, an integer~$k$, and a value~$W \in \mathds{R}$.}
{Is there a set~$S \subseteq V$ of at most~$k$ vertices of total weight at least~$W$ such that~$G[S]$ is a forest?}

The definition of \fragment{} was already shown in the introduction.
We repeat it here to have all problem definitions in one place.

\probdef{\fragment}
{A directed graph~$G=(V,A)$, a vertex-cost function~$c \colon V \rightarrow \mathds{N}$, a weight function~$w \colon V \cup A \rightarrow \mathds{R}$, an integer~$k$, and values~$\alpha_1,\alpha_2,\alpha_3,\beta, W \in \mathds{R}$.}
{Is there a set~$S$ of cost exactly~$k$, that is, $\sum_{v\in S}c(v) = k$ and \[\hspace*{-2cm}\alpha_1 \big(\sum_{\substack{(u,v) \in A\\u,v \in S}} w((u,v))\big) + \alpha_2 \big(\sum_{\substack{(u,v) \in A\\u \in S, v \notin S}} w((u,v))\big) + \alpha_3 \big(\sum_{\substack{(u,v) \in A\\u \notin S,v \in S}} w((u,v))\big) + \beta \big(\sum_{v \in S} w(S)\big) \geq W?\]}

When the parameters~$\alpha_1,\alpha_2,\alpha_3$ and~$\beta$ are clear from the context, then we call the value for a particular set~$S$ the \emph{\textbf{score}} of~$S$.
All other problems we study are simple corollaries of the above problems.
We hence only give informal problem statements for these.
In \wpvc, the input consists of a graph~$G$ with positive integer costs and real-valued edge weights, an integer~$k$, and a real value~$W$.
The question is whether there is a set of vertices of cost at most~$k$ such that the weight of covered edges is at least~$W$.
An edge is covered if at least one of its endpoints is contained in the solution.
Note that we require all costs to be integers as otherwise the problem is para-NP-hard when parameterized by~$k$ (start with an instance~$(G,k,W)$ of unweighted \textsc{Partial Vertex Cover} and set the weight of each vertex to~$\frac{1}{k}$).
In \wis, we are given a graph~$G$ with real-valued vertex weights, an integer~$k$, and a real value~$W$.
The questions is whether an independent set of size at most~$k$ and weight at least~$W$ exists.

In \den, we are given an undirected graph~$G$ and two integers~$k$ and~$W$.
The question is whether there is a set of exactly~$k$ vertices such that there are at least~$W$ edges with both endpoints in that set.
\knkcut{} can be defined on both undirected and directed graphs together with two integers~$k$ and~$W$.
In the undirected setting, the question is whether there is a set~$S$ of exactly~$k$ vertices such that there are at least~$W$ edges with exactly one endpoint in~$S$.
In the directed setting, the question is whether there is a set~$S$ of exactly~$k$ vertices such that there are at least~$W$ edges~$(u,v)$ with~$u \in S$ and~$v \notin S$.

In \wsi, the task is to find a subgraph of minimum weight that is isomorphic to a given pattern graph~$H$.
In \wisi, the task is to find an induced subgraph of maximum weight that is isomorphic to a given pattern graph~$H$.
Note that the original work by Nederlof~\cite{Ned20} showed how to solve the unweighted variants of the above two problems (and also how to count the number of solutions).

\wsp{} is a generalization of \textsc{Steiner Tree} where the solution is disconnected.
Here, we are given an undirected graph, a set~$\mathcal{T}=\{X_1,X_2,\ldots,X_t\}$ of sets of terminal vertices, an edge-weight function, an integer~$k$, and a real value~$W$.
The task is to find a subgraph~$G'$ of~$G$ with at most~$k$ vertices of total weight at most~$W$ such that for each~$i \in [t]$, all vertices in~$X_i$ are in the same connected component in~$G'$ and for any~$i \neq j$, the vertices in~$X_i$ and the vertices in~$X_j$ are in distinct connected components in~$G'$.
\wsp{} can be seen as the dual of \textsc{Cut-Uncut}, where the goal is to remove as few edges as possible to get to a subgraph satisfying the above.
The special case of \textsc{Cut-Uncut} with only two sets~$X_1$ and~$X_2$ was recently studied for planar graphs~\cite{BDFGK24}.

\subsection{Request Variants}
We now give a characterization of the type of problems our framework can solve.
We call those \emph{\textbf{request variants}} and for a regular problem~$\pi$, we call its request variant \Q.
In order to allow the framework to be applicable in as many situations as possible, we do not give a strict recipe of how to generalize a problem $\pi$ to \Q, but instead our main result (\cref{thm:main-fpt}) lists the requirements that the problem \Q{} has to satisfy in order for the framework to work.
Before we give a formal definition of requests variants, we will give some intuition regarding the different concepts we rely on and revisit the example \textsc{Independent Set} from the introduction.

The first requirement we have is a notion of \emph{locality}.
This stems from the use of the lemma by Nederlof and basically says that the famous approach by Baker is applicable.
In particular, we require that there is some constant~$\alpha_{\pi}$ depending on the problem such that the following holds.
If we are looking for a solution with at most~$k$ vertices, and if we compute a breadth-first search from an arbitrary vertex~$r$ and then partition all vertices depending on their distance from~$r$ modulo~$\alpha_{\pi}k+1$, then there exist a solution in the entire graph if and only if for at least one of the~$\alpha_{\pi}k+1$ parts, it holds that after removing all vertices from the respective part, there is a solution of size at most~$k$ in the remaining graph.
In the case of \textsc{Independent Set}, this trivially holds for~$\alpha = 1$.
The~$k$ vertices of a solution can be contained in at most~$k$ of the parts so one part does not contain any vertices from the solution.
Removing all vertices in this part does neither create new independent sets nor does it invalidate the solution we consider.
We say that the request variant of \textsc{Independent Set} has \emph{\textbf{locality}}~$\alpha = 1$.
This notion of locality will allow us to reduce the problem to a setting where the graph is~$\alpha_{\pi}k$-outerplanar.

The next ingredient for request variants are \emph{requests}.
As stated above, we wish to solve problems where the output is not a single solution but rather a lookup table.
We call the queries that this lookup table should be able to answer \textbf{\emph{requests}} and for a problem \Q, we denote the function that enumerates all possible requests by~$R_{\pi}$. 
These requests mainly depend on the following four objects: a parameter~$d$, a graph~$G$, a \emph{\textbf{portal set}}~${B = (B_{\he},B_{\li},B_{\di})}$, and a pebble set~$\peb$.

Let us give some intuition about these.
The graph~$G$ should be thought of as a subgraph of the input graph defined by cutting out some parts along certain cycle separators.
The parameter~$d$ describes an upper bound on both the number of vertices we allow in a solution and (together with~$\alpha_{\pi}$) the outerplanarity of~$G$.
We will achieve subexponential parameterized algorithms for this parameter~$d$.
As we will reduce the value of~$d$ in some recursive calls, we prefer to distinguish it from the input parameter~$k$.
The portal set~$B$ should be thought of as the boundary of~$G$ to the rest of the input graph.
For technical reasons, we always give the portal set as the union of three sets: heavy, light, and discarded vertices.
The intuition is that in order to keep the number of possible intersections of a solution with the portal set subexponential, we restrict ourselves to those solutions that do not contain any discarded vertices and sufficiently few light vertices.
On the flip side, the number of heavy vertices will be quite small and the number of light vertices will be somewhat small.
The pebble set describes which vertices in~$B$ should be contained in the solution.
This should be seen as the interface of a partial solution to the rest of the graph.
Each request variant will take an integer~$x$ as input which will limit how many vertices of the solution we allow to be contained in~$B$.
Intuitively, this simply counts how many cycle separators are contained in~$B$ but this intuition will not quite work due to technicalities when we compute cycle separators that separate other cycle separators in a balanced way.
We then consider pebble sets for the tuple~$(B,d,\alpha_{\pi},x)$.
Recall that these are sets~$\peb \subseteq B_{\he} \cup B_{\li}$ with~$|\peb \cap B_{\li}| \leq 20x\alpha_{\pi} \sqrt{d} \log(d)$.

We also allow requests to depend on one additional value, which we will use for problem specific information.
We denote these by~$S_{\pi}$ and we only have one minimal requirements regarding these: they should be maintainable in subexponential time.
In many cases, the set~$S_{\pi}$ simply contains weight functions and for \textsc{Independent Set Request}, this set will be empty.
In particular, a request for~\textsc{Independent Set Request} will only consist of a pebble set~$\peb$.
We can now formally define \textsc{Independent Set Request} and afterwards give a definition for request variants in general.

\problemdef{\textsc{Independent Set Request}}
{An integer~$d$, a~$d$-outerplanar graph~$G=(V,E)$, a portal set~$B=(B_{\he},B_{\li},B_{\di})$, and an integer~$x$.}
{Output a lookup table that returns for each request~$r = \peb$ for each pebble set~$\peb$ for~$(B,d,1,x)$, the maximum size of an independent set~$\mathcal{I}$ in~$G$ where~$\mathcal{I} \cap (B_{\he} \cup B_{\li} \cup B_{\di}) = \peb$ (or~$-\infty$ if no such set exists).}

Recall that a pair~$(B,x)$---consisting of a set~$B$ of vertices partitioned into heavy, light, and discarded vertices and an integer~$x$---is almost measured with respect to two integers~$d$ and $\alpha_{\pi}$ if~$x \leq \max(200 \log^7(d),1)$, ${|B_{\he}| \leq 40 x \alpha_{\pi} \sqrt{d}}$, and~${|B_{\li}| \leq 60 x \alpha_{\pi} d \log(d)}$.

\begin{definition}
    A \emph{request variant} \Q{} of locality~$\alpha_{\pi}$ takes a tuple~${(d,G,B,x,S_{\pi})}$ as input where~$B=(B_{\he},B_{\li},B_{\di})$ is partitioned into heavy, light, and discarded vertices, $G$ is a~$\alpha_{\pi}d$-outerplanar graph, and $(B,x)$ is almost measured with respect to~$d$ and~$\alpha_\pi$.
    It outputs a lookup table with an entry for each request~$r \in R_\pi(d,G,B,\peb,S_{\pi})$ where~$\peb$ is a pebble set for~$(B,d,\alpha_{\pi},x)$.
    The output represents some solution where the vertices in the solution contained in~$B_{\he} \cup B_{\li} \cup B_{\di}$ is exactly~$\peb$.
\end{definition}

\subsection{Results Used in our Algorithm}
Finally, we state a number of results from the literature, which we will subsequently rely on.
The first one is a well-known strategy due to Baker to transform a planar graph into several~$d$-outerplanar graphs.

\begin{lemma}[\cite{Bak94}]
    \label{lem:baker}
    Given a planar graph~$G$ and an integer~$d$, one can compute in linear time a partition of the vertices of~$G$ into sets~$A_1,A_2,\ldots,A_{d}$ such that for each~$i \leq d$, the graph~$G - A_i$ is~$(d-1)$-outerplanar. 
\end{lemma}

The second result allows us to triangulate a~$d$-outerplanar graph.

\begin{lemma}[\cite{Bie15}, Corollary 1]
    \label{lem:tri}
    Any $d$-outerplanar graph~$G$ can be triangulated in~$\Oh(n^2)$ time such that the resulting graph has outer-planarity at most~$d + 1$.
\end{lemma}

The third result will be used to compute balanced cycle separators.
Therein, a fundamental cycle of an edge~$e$ with respect to a spanning tree~$T$ (which does not contain~$e$) is the unique simple cycle that contains~$e$ and all its other edges are contained in~$T$.

\begin{lemma}[\cite{KM}, Lemma 5.3.2.]
    \label{lem:cyclehelp}
    There is a linear-time algorithm that, given a triangulated plane graph~$G$, a spanning tree~$T$ of~$G$, and a proper weight assignment, returns a nontree edge~$e$ such that the fundamental cycle of~$e$ with respect to~$T$ is a \nicefrac{3}{4}-balanced cycle separator for~$G$.
\end{lemma}

We use this lemma in the following way.

\begin{restatable}{corollary}{cycle}
    \label{lem:cycle}
    Given a~$d$-outerplanar triangulated graph~$G$ and a proper weight assignment, one can compute in~$\Oh(n^2)$ time a \nicefrac{3}{4}-balanced cycle separator of size at most~$2d+1$ for~$G$.
\end{restatable}

\begin{proof}
    First, compute in~$\Oh(n^2)$ time a~$d$-outerplanar embedding of~$G$ \cite{Kam07}.
    Let~$L(v)$ be the layer of~$v$, that is, the vertices on the outer face have~$L(v) = 0$ and when removing all vertices with~${L(v) < i}$ for some~$i \in [d-1]$, the vertices on the new outer face receive~$L(v) = i$.
    Since~$G$ is triangulated, the outer face is a triangle.
    Build a spanning tree~$T$ for~$G$ in~$\Oh(n)$ time as follows.
    Start with any two edges on the outer face and then iteratively for increasing~$i \in [d-1]$ and each vertex~$v$ with~$L(v) = i$ add an edge between~$v$ and a vertex~$u$ with~$L(u) = i-1$ to~$T$.
    Note that such an edge exists as~$G$ is triangulated.
    Now apply \cref{lem:cyclehelp} to compute in~$\Oh(n)$ time an edge~$e$ not contained in~$T$ such that the fundamental cycle of~$e$ with respect to~$T$ is a \nicefrac{3}{4}-balanced cycle separator for~$G$.
    Note that the fundamental cycle of~$e$ can be computed in~$\Oh(n)$ time and it contains at most two vertices of each layer (except for the outermost, where it can contain three) and therefore at most~$2d+1$ vertices in total.
\end{proof}

Next, we use a result that bounds the treewidth of~$k$-outerplanar graphs.

\begin{lemma}[\cite{Bodlaender98}]
    \label{lem:bodlaender}
    Each~$k$-outerplanar graph has treewidth at most~$3k-1$.
\end{lemma}

We use this in combination with a linear-time algorithm to compute an approximate tree decomposition in the case where the treewidth is a constant and a result to transform a tree decomposition into a binary tree decomposition of logarithmic depth without increasing the width too much.

\begin{theorem}[\cite{Kor21}]
\label{thm:tuukka}
There is an algorithm that, given an $n$-vertex graph $G$ and an integer $k$, in~$2^{\Oh(k)}n$~time either outputs a tree decomposition of $G$ of width at most $2k + 1$ or determines that the treewidth of~$G$ is larger than $k$.
\end{theorem}

\begin{lemma}[\cite{BH98}]
\label{lem:BH}
There is an algorithm that, given an $n$-vertex graph $G$ and a tree decomposition of width~$k$, computes a binary tree decomposition of~$G$ of width at most~$3k+2$ and depth~$\Oh(\log(n))$ in~$2^{\Oh(k)}n$~time.
\end{lemma}

The main ingredient of our subexponential parameterized algorithm is a lemma due to Nederlof~\cite{Ned20}.
He mentions the lemma without formal statement or proof and he does not analyze the constants involved in the running time.
However, the proof follows directly from the proof of his Lemma 3.7 and we will give a proof sketch for analyzing the running time.

\begin{lemma}
    \label{lem:unknownbalance}
    There is an~$2^{\Oh(\log^2(k))} n^{1+o(1)}$-time algorithm that
    \begin{itemize}
        \item given a triangulated graph~$G=(V,E)$ and a cycle separator~$C$ on at most~$k$ vertices that is~$q$-balanced for a (possibly unknown) proper weight assignment~$w$,
        \item outputs a set~$\mathcal{C}$ of at most~$2^{44 \log^2(k)}$ cycle separators each partitioned into heavy, light, and discarded vertices such that for each~$C \in \mathcal{C}$, $(C,1)$ is measured for~$k$ and 1 and for each~$P \subseteq V$ with~$|P| \leq k$, at least one~$C \in \mathcal{C}$  is~$\max(\nicefrac{3}{4},q)$-balanced for~$w$, $P\cap C_{\di} = \emptyset$, and \[|P \cap C_{\li}| \leq \frac{11 |P| \log(k)}{2 \sqrt{k}}.\]
    \end{itemize}
\end{lemma}

\begin{proof}[Analysis of Running time.]
    The algorithm by Nederlof iteratively replaces one cycle separator by~$4$.
    The total number of cycle separators considered is in~$\Oh(2^{44 \log^2(k)})$ and for each cycle separator, he applies a menger-type algorithm by Fomin, Lokshtanov, Marx, Pilipczuk, Pilipczuk, and Saurabh~\cite{FominLMPPS22}.
    Based on the outcome, only some vertices are added to either the heavy, light, or discarded vertices.
    The time needed for the latter is dominated by the menger-type algorithm, so it only remains to analyze this algorithm.
    This algorithm first computes in~$\Oh(n+m)$ time an auxiliary graph~$H$ whose number of vertices and edges is linear in the number of vertices and edges in the original graph.
    Since the input graph is planar, this takes~$\Oh(n)$ time and~$H$ also has~$\Oh(n)$ edges (but is not necessarily planar).
    It then computes a minimum-cost flow in~$H$, partitions the flow into walks, shortcuts these walks to paths if necessary, computes a solution to the dual LP of the minimum-cost flow problem, and partition the vertices based on the solution for the dual.
    Using recent improvements for almost-linear time algorithms for flows and their duals~\cite{ChenKLPGS22}, we can compute an optimal integral solution to the minimum-cost flow problem and its dual LP in~$\Oh(m^{1+o(1)})=\Oh(n^{1+o(1)})$ time.
    All other steps can be performed in~$\Oh(k n \log n)$ time, so overall, the algorithm takes~$2^{\Oh(\log^2(k))}n^{1+o(1)}$ time.
\end{proof}

Finally, we summarize a couple of simple inequalities that we will use repeatedly.

\begin{lemma}
    \label{lem:help}
    If~$d \geq 16$, then~$\log^2(d) \leq d$.
    If~$d \geq 54$, then~$\lfloor d - 3\sqrt{d}\rfloor \geq \frac{4d}{7}$.
    If~$d \geq 54$, then~$\sqrt{D-3\sqrt{D}}+1-\sqrt{D} \leq -\nicefrac{1}{10}$.
\end{lemma}

\begin{proof}
    Since the first statement is trivial, we will focus on the second and third.
    To this end, observe that~$\lfloor d - 3\sqrt{d} \rfloor \geq d - 3\sqrt{d} - 1$.
    As both sides of the inequality are positive for any~$d \geq 11$, we can square both sides and only show~$d-3\sqrt{d}-1 \geq \frac{4d}{7}$.
    Note that
    \begin{align*}
        &d - 3\sqrt{d} - 1 \geq (\nicefrac{4}{7})d\\
        \Leftrightarrow\ &(\nicefrac{3}{7})d \geq 3 \sqrt{d} + 1\\
        \Leftrightarrow\ &\sqrt{d} \geq 7 + \nicefrac{7}{(3\sqrt{d})},
    \end{align*}
    $\nicefrac{7}{(3\sqrt{d})} \leq \nicefrac{1}{3}$ for all~$d \geq 49$, and~$7+\nicefrac{1}{3} \leq \sqrt{54}$.

    Note that the third inequality is equivalent to~$\sqrt{D} - \nicefrac{11}{10} \geq \sqrt{D-3\sqrt{D}}$.
    Since both sides are positive for any~$d \geq 54$ (as~$\sqrt{D} - \nicefrac{11}{10} \geq \sqrt{54} - 2 > 5$ and~$\sqrt{D - 3\sqrt{D}} \geq \sqrt{54 - 3\sqrt{54}} > 5$), we can square both sides and only show~$(\sqrt{D} - \nicefrac{11}{10})^2 \geq D - 3\sqrt{D}$.
    This holds as
    \[(\sqrt{D} - \nicefrac{11}{10})^2 = D - \frac{22}{10}\sqrt{D} + \frac{121}{100} \geq D - \frac{30}{10} \sqrt{D} + 0 = D - 3\sqrt{D}.\qedhere\]
\end{proof}

\section{The Framework}
\label{sec:framework}
In this section, we show our main result, \cref{thm:main-fpt}.
It describes a general algorithmic framework for solving request variants on planar graphs in subexponential parameterized time.
The framework can handle non-bidimensional weighted problems and does not require the solution to consist of a small number of connected components.
It can also handle directed graphs, but in order to keep the presentation as concise as possible, we will focus on undirected graphs here and mention that just considering the underlying undirected graph allows to extend all results to directed graphs as well.
We next state our main theorem, but defer its proof to the end of this section.
Note that in the following we consider arbitrary separators, not necessarily cycle separators.
We still use the terms~$C_{\inte}$ and~$C_{\ex}$ to refer to the graphs induced by either side of the separation.
\begin{restatable}{theorem}{mainalgorithm}
    \label{thm:main-fpt}
    Let \Q{} be a computable request variant, $\rho \geq 1$ be a constant, and~$\gamma$ be a non-decreasing computable function such that
    \begin{itemize}
        \item[R1] \Q{} has constant locality~$\alpha_{\pi}$, that is, for any instance~${I=(d,G,B,x,S_{\pi})}$ of \Q, any request~${r \in R_{\pi}(d,G,B,\peb,S_{\pi})}$, and any BFS layering~$V_1,V_2,V_3,\ldots,V_i$ of~$G$ with~$i > \alpha_{\pi} d$, the following holds.
        Let~$V^j$ be the union of all sets~$V_\ell$ with~$\ell \equiv_{\alpha_{\pi}d+1} j$ for each~$j \in [\alpha_{\pi}d+1]$ and let~$I'_j = (d,G'=G[V \setminus V^j],B[G'],x,S_{\pi}[G'])$ be the instance where all vertices from~$V_j$ are removed.
        Then, the output for instance~$I$ and request~$r$ is equal to either the maximum or the minimum of the outputs for all~$I'_j$ and~$r$ containing all vertices in~$\peb$. \label{req:baker}
        \item[R2] $R_\pi(d,G,B,\peb,S_{\pi})$ can be computed in~$\Oh(\gamma(d) n^\rho)$ time. \label{req:r}
        \item[R3] Given a separator~$C$, $S_{\pi}[C_{\inte}]$ and~$S_{\pi}[C_{\ex}]$ can be computed in~$\Oh(\gamma(d) n^{\rho})$ time. \label{req:s}
        \item[R4] There is an algorithm that given a separator~$C=(C_{\he},C_{\li},C_{\di})$ such that~$(C,1)$ is measured for~$d$ and~$\alpha_{\pi}$, a pebble set~$\peb$ for $(B,d,\alpha_{\pi},x)$, and the solutions for instances
        \begin{align*}
            (d,C_{\inte},((B_{\he} \cup C_{\he})[C_{\inte}],(B_{\li} &\cup C_{\li})[C_{\inte}],(B_{\di} \cup C_{\di})[C_{\inte}]),x+1,S_{\pi}[C_{\inte}]) \text{ and}\\
            (d,C_{\ex},((B_{\he} \cup C_{\he})[C_{\ex}],(B_{\li} &\cup C_{\li})[C_{\ex}],(B_{\di} \cup C_{\di})[C_{\ex}]),x+1,S_{\pi}[C_{\ex}]),
        \end{align*}
        can compute the solutions for instance~$I$ for all requests~${r \in R_\pi(d,G,B,\peb,S_{\pi})}$ in~$\Oh(\gamma(d)n^\rho)$ time that do not intersect~$C_{\di}$ and intersect~$C_{\li}$ in at most~$20\alpha_{\pi}\sqrt{d}\log(d)$ vertices.
        \label{req:merge}
    \end{itemize}
    Then, \Q{} can be solved in~$2^{\tilde{\Oh}(\sqrt{d})} \gamma(d) n^{\max(2.49,\rho+\varepsilon)}$ time for any~$\varepsilon > 0$.
\end{restatable}

The rest of this section is structured as follows.
First, we present three algorithms that recursively call one another to solve the problem.
One of them ($\SQ$) is the main algorithm, one ($\clean$) is for the special case in which the portal set~$B$ becomes too large, and the final ($\base$) handles the base cases when the parameter becomes small enough.
We then show that a number of invariants are maintained throughout all algorithm calls.
In the following three subsections, we analyze each of the three algorithms.
There, we give an informal description of the respective algorithm and prove that it works correctly.
This will yield the correctness proof of \cref{thm:main-fpt} in the end.
In the next section, we will show how to apply this theorem to different problems.

\subsection{Algorithms and Invariants}
We start by formally stating the three aforementioned algorithms: $\SQ$ (\cref{alg:Q}), $\clean$ (\cref{alg:clean}), and $\base$ (\cref{alg:base}).
We present the algorithms for the case where~$\Q$ is a maximization problem, but the minimization version can be solved by essentially the same algorithms.
\begin{algorithm}[tbp!]
    \DontPrintSemicolon
    \KwIn{${(d,G=(V,E),x,B_{\he},B_{\li},B_{\di},S_{\pi})}$ as specified in \cref{lem:invariants}.}
    \If{$|V| < 2^{7\alpha_{\pi}}$\label{lin:smallnif}}{
        \Return brute force solution for instance~$(d,G,B,x,S_{\pi})$ \label{lin:bruteforcesmalln}\;
    }
    \If{$d < 2^{7\alpha_{\pi}}$\label{lin:basecaseif}}{
        \Return $\base(d,G,x,B_{\he},B_{\li},B_{\di},S_{\pi})$\label{lin:basecase}\;
    }
    \If{$x \geq 40\log^7(d)$\label{lin:cleanif}}{
        \Return $\clean(d,G,x,0,B_{\he},\emptyset,B_{\li},\emptyset,B_{\di},\emptyset,S_\pi)$ \label{lin:clean}\;
    }
    $\res \gets \overrightarrow{\infty}$\; \label{lin:sqini}
    Use \cref{lem:tri} to compute a set~$\Delta$ of edges such that~$G^{\Delta} = G+\Delta$ is triangulated and~$(\alpha_{\pi}d+1)$-outerplanar\;
    Use \cref{lem:cycle} to find a~$\nicefrac{3}{4}$-balanced cycle separator~$C$ with a vertex set~$V_C$ of at most~$2\alpha_{\pi}d+3$ vertices in~$G^{\Delta}$ for the weight functions that assigns~$\nicefrac{1}{|V|}$ to each~$v \in V$\; \label{lin:cycle}
    \tcc{case where $C$ contains at most~$20\sqrt{d} \log(d)$ vertices of the solution}
    Compute $S_{\pi}[C_{\inte}]$ and $S_{\pi}[C_{\inte}]$ using Requirement R3\;
    $\res_1 \gets \SQ(d,C_{\inte}-\Delta,x+1,B_{\he}[C_{\inte}],(B_{\li} \cup C)[C_{\inte}],B_{\di}[C_{\inte}],S_\pi[C_{\inte}])$\; \label{lin:rec1}
    $\res_2 \gets \SQ(d,C_{\ex}-\Delta,x+1,B_{\he}[C_{\ex}],(B_{\li} \cup C)[C_{\ex}],B_{\di}[C_{\ex}],S_\pi[C_{\ex}])$\; \label{lin:rec2}
    \ForEach{pebble set~$\peb$ for~$(B,d,\alpha_{\pi},x)$}{
                    Compute~$R_\pi(d,G,B,\peb,S_{\pi})$ using Requirement R2\;
                    Use Requirement R4 to compute a solution~$s(r)$ for each~$r \in R_\pi(d,G,B,\peb,S_{\pi})$ based on~$\res_1$ and~$\res_2$\; \label{lin:upSQ}
                    $\res(r) = \max(\res(r),s(r))$ \label{lin:upSQ2}\;
            }
    \tcc{case where $C$ contains more than~$20\sqrt{d} \log(d)$ vertices of the solution}
    Use \cref{lem:unknownbalance} with input~$G^{\Delta}$ and~$C$ to compute a set~$\mathcal{S}$ of separators\; \label{lin:balance}
    $d' = \lfloor d - 3\sqrt{d} \rfloor$\;
    \ForEach{$C' \in \mathcal{S}$ with partition~$(C'_{\he},C'_{\li},C'_{\di})$}{
        Use \cref{lem:baker} with~$\alpha_{\pi}d'$ on the interior of~$C'$ to compute sets $A_{\inte}^1,A_{\inte}^2,\ldots,A_{\inte}^{\alpha_{\pi}d'+1}$\; \label{lin:baker1}
        Use \cref{lem:baker} with~$\alpha_{\pi}d'$ on the exterior of~$C'$ to compute sets $A_{\ex}^1,A_{\ex}^2,\ldots,A_{\ex}^{\alpha_{\pi} d'+1}$\; \label{lin:baker2}
        \ForEach{$A_{\inte}$, $A_{\ex}$ \label{lin:bakerfor}}{
            $\res_1 \gets \SQ(d',{C'_{\inte} - A_{\inte} - \Delta}, {2x+2}, {(B_{\he} \cup C'_{\he})[C'_{\inte}]}, {(B_{\li} \cup C'_{\li})[C'_{\inte}]}, \allowbreak (B_{\di} \cup C')[C'_{\inte}],S_{\pi}[C'_{\inte}])$\; \label{lin:rec3}
            $\res_2 \gets \SQ(d',C'_{\ex} - A_{\ex} - \Delta,2x+2,(B_{\he} \cup C'_{\he})[C'_{\ex}], (B_{\li} \cup C'_{\li})[C'_{\ex}], \allowbreak (B_{\di} \cup C')[C'_{\ex}], S_{\pi}[C'_{\ex}])$\; \label{lin:rec4}
            \ForEach{pebble set~$\peb$ for~$(B,d,\alpha_{\pi},x)$}{
                    Compute~$R_\pi(d,G,B,\peb,S_{\pi})$ using Requirement R2\;
                    Use Requirement R4 to compute a solution~$s(r)$ for each~$r \in R_\pi(d,G,B,\peb,S_{\pi})$ based on~$\res_1$ and~$\res_2$\; \label{lin:critical}
                    $\res(r) = \max(\res(r),s(r))$ \label{lin:critical2}\;
            }
        }
        \Return $\res$\;
    }
    \caption{$\SQ$ (maximization)}
    \label{alg:Q}
\end{algorithm}%
\begin{algorithm}[t]
    \DontPrintSemicolon
    \KwIn{$(d,G=(V,E),x,y,X_{\he},Y_{\he},X_{\li},Y_{\li},X_{\di},Y_{\di},S_{\pi})$  as specified in \cref{lem:invariants}.}
    $\res \gets \overrightarrow{\infty}$\;  \label{lin:cleanini}
    \eIf{$|X_{\he} \cup X_{\li}| \leq 20 \sqrt{d}$ \label{lin:solveif}}{
        \Return $\SQ(d,G,2y+1,X \cup Y)$\; \label{lin:solve}
    }{
        Use \cref{lem:tri} to compute a set~$\Delta$ of edges such that~$G^{\Delta} = G+\Delta$ is triangulated and~$(\alpha_{\pi}d+1)$-outerplanar\;
        Use \cref{lem:cycle} to find a~$\nicefrac{3}{4}$-balanced cycle separator~$C$ with at most~$2\alpha_{\pi} d+3$ vertices in~$G^{\Delta}$ for the weight functions that assigns~$\nicefrac{1}{|X|}$ to each vertex in~$X=X_{\he}\cup X_{\li}$.\;
        Use \cref{lem:unknownbalance} with input~$G^{\Delta}$ and~$C$ to compute a set~$\mathcal{S}$ of separators\; \label{lin:cleanbalance}
        \ForEach{$C' \in \mathcal{S}$ with partition~$(C'_{\he},C'_{\li},C'_{\di})$}{
            Compute $S_{\pi}[C'_{\inte}]$ and~$S_{\pi}[C'_{\ex}]$ using Requirement R3\;\label{lin:cleans}
            $\res_1 \gets \clean(d,C'_{\inte} - \Delta,x,y+1,X_{\he}[C'_{\inte}], (Y_{\he} \cup C'_{\he})[C'_{\inte}], \allowbreak X_{\li}[C'_{\inte}],(Y_{\li} \cup C'_{\li})[C'_{\inte}],X_{\di}[C'_{\inte}],(Y_{\di}\cup C'_{\di})[C'_{\inte}],S_{\pi}[C'_{\inte}])$\; \label{lin:rec5}
            $\res_2 \gets \clean(d,C'_{\ex} - \Delta,x,y+2,X_{\he}[C'_{\ex}], (Y_{\he} \cup C'_{\he})[C'_{\ex}], \allowbreak X_{\li}[C'_{\ex}],(Y_{\li} \cup C'_{\li})[C'_{\ex}],X_{\di}[C'_{\ex}],(Y_{\di}\cup C'_{\di})[C'_{\ex}],S_{\pi}[C'_{\ex}])$\; \label{lin:rec6}
            \ForEach{pebble set~$\peb$ for~$(B,d,\alpha_{\pi},x)$}{
                    Compute~$R_\pi(d,G,B,\peb,S_{\pi})$ using Requirement R2\; \label{lin:cleanr}
                    Use Requirement R4 to compute a solution~$s(r)$ for each~$r \in R_\pi(d,G,B,\peb,S_{\pi})$ based on~$\res_1$ and~$\res_2$\; \label{lin:upclean}
                    $\res(r) = \max(\res(r),s(r))$\;
            }
        }
        \Return $\res$\;
    }
    \caption{$\clean$ (maximization)}
    \label{alg:clean}
\end{algorithm}%
\begin{algorithm}[tp]
    \DontPrintSemicolon
    \KwIn{$(d,G=(V,E),x,B_{\he},B_{\li},B_{\di},S_{\pi})$ as specified in \cref{lem:invariants}.}
    Use \cref{thm:tuukka} to compute a tree decomposition of width at most~$6\alpha_{\pi}d-1$ \label{lin:tuukka}\;
    Use \cref{lem:BH} to modify the tree decomposition into a binary tree decomposition~$\mathcal{T}$ of width~$w \leq 18\alpha_{\pi}d-1$ and depth~$\Oh(\log(n))$ \label{lin:BH}\;
    $D \gets 4d^2$\;
    \tcc{Build relevant instances and brute force solution for pairs of bags}
    \ForEach{node~$u$ except the root of the tree decomposition\label{lin:for1}}{
        Let~$u'$ be the parent of~$u$ in~$\mathcal{T}$\;
        $Z_u \gets X_u \setminus (B_{\li} \cup B_{\di})$\;
        $X'_u \gets X_u \cup X_{u'}$\;
        $Z'_u \gets X'_u \setminus (B_{\li} \cup B_{\di})$\;
        $I_u \gets (D,G[Y_u],(Z_u, B_{\li} \cap X_u, B_{\di} \cap X_u),\infty,S_{\pi}[G[Y_u]])$\;
        $I'_u \gets (D,G[X'_u],(Z'_u, B_{\li} \cap X'_u, B_{\di} \cap X'_u),\infty,S_{\pi}[G[X'_u]])$\;
        $\res'_u \gets$ brute force solution for $I'_u$ \label{lin:bruteforcebase}\;
    }
    \tcc{Bottom-Up dynamic programming}
    \ForEach{node~$u$ of the tree decomposition with at least one child from bottom to top\label{lin:for2}}{
        Let~$v_1$ be the first child of~$u$ in~$\mathcal{T}$\;
        $I^1_u \gets (D,G_1 = G[X_u \cup Y_{v_1}],(B_{\he}[G_1] \cup Z_{v_1},B_{\li}[G_1],B_{\di}[G_1]),\infty,S_{\pi}[G_1])$\;
        Use Requirement R4 with separator~$(Z_{v_1},B_{\li} \cap X_{v_1},B_{\di} \cap X_{v_1})$ and solutions~$\res'_{v_1}$ (for instance~$I'_{v_1}$) and~$\res_{v_1}$ (for instance~$I_{v_1}$) to get solution~$\res_u^1$ for instance~$I^1_u$\;
        \If{$u$ has only one child}{
            $\res_u \gets \res_u^1$\;
        }\Else{
            Let~$v_2$ be the second child of~$u$ in~$\mathcal{T}$\;
            $J_u \gets (D,G' = G[X_u \cup Y_{v_2}],(B_{\he}[G'] \cup Z_{v_2},B_{\li}[G'],B_{\di}[G']),\infty,S_{\pi}[G'])$\;
            $I^2_u \gets (D,G_2 = G[Y_u],(B_{\he}[G_2] \cup Z_{v_2},B_{\li}[G_2],B_{\di}[G_2]),\infty,S_{\pi}[G_2])$\;
            Use Requirement R4 with separator~$(Z_{v_2},B_{\li} \cap X_{v_2},B_{\di} \cap X_{v_w})$ and solutions~$\res'_{v_2}$ (for instance~$I'_{v_2}$) and~$\res_{v_2}$ (for instance~$I_{v_2}$) to get solution~$s_u$ for instance~$J_u$\;
            Use Requirement R4 with separator~$(Z_{u},B_{\li} \cap X_{u},B_{\di} \cap X_{u})$ and solutions~$\res_u^1$ (for instance~$I^1_u$) and~$s_u$ (for instance~$J_u$) to get solution~$\res_u^2$ for instance~$I^2_u$ \label{lin:I2}\;
            $\res_u \gets \res_u^2$\;
        }
    }
    \tcc{Removing~$Z_p$ for root~$p$ from portal set}
    $\res'_r \gets$ brute force solution for instance~$I'_z=(D,G[X_p],(X_p \setminus (B_{\li} \cup B_{\di}),B_{\li} \cap X_p,B_{\di} \cap X_p),\infty,S_{\pi}[G[X_p]])$\;
    Use Requirement R4 with separator~$(X_p \setminus (B_{\li} \cup B_{\di}),\emptyset,\emptyset)$ and solutions~$\res'_p$ for instance~$I'_p$ and solution~$\res_p$ for~$(D,G,(B_{\he} \cup (X_p \setminus (B_{\li} \cup B_{\di})),B_{\li},B_{\di}),\infty,S_{\pi})$ to compute solution~$\res^*$ for~$(D,G,(B_{\he},B_{\li},B_{\di}),\infty,S_{\pi})$ \label{lin:remove}\;
    \tcc{Discarding solutions for too large pebble sets}
    \ForEach{pebble set~$\peb$ for~$(B,d,\alpha_{\pi},x)$\label{lin:pebblefor}}{
        Compute~$R=R_\pi(d,G,B,\peb,S_{\pi})$ using Requirement R2\;
        \ForEach{$r \in R$}{
            $\res(r) \gets \res^*(r)$\;
        }
    }
    \Return $\res$\;
    \caption{$\base$}
    \label{alg:base}
\end{algorithm}%
The aforementioned invariants are summarized in the following lemma.
\begin{restatable}{lemma}{invariants}
    \label{lem:invariants}
    Throughout all calls of $\SQ$, $\clean$, and~$\base$, the following invariants are maintained.
    \begin{itemize}
        \item The input to~$\SQ$ is a valid instance of~\Q{} and~$x\leq 200\log^7(d)-44\log(d)$.
        \item The input~$(d,G,x,y,X_{\he},Y_{\he},X_{\li},Y_{\li},X_{\di},Y_{\di},S_{\pi})$ given to $\clean$ satisfies the following: $(d,G,x+2y,X_{\he} \cup Y_{\he},X_{\li} \cup Y_{\li},X_{\di} \cup Y_{\di},S_{\pi})$ is a valid instance of \Q,~$d \geq 2^{7\alpha_{\pi}}$, $x \leq 200\log^7(d)-44\log(d)$, $y < 22 \log(d)$, ${|X_{\he} \cup X_{\li}| \leq (\nicefrac{3}{4})^y \cdot 50d^6}$, $|Y_{\he}| \leq 40 (2y) \alpha_{\pi} \sqrt{d}$, and~$|Y_{\li}| \leq 60 (2y) \alpha_{\pi} d \log(d)$.
        \item When $\SQ$ with values~$d_1$ and~$z_1 = x_1$ or $\clean$ with values~$d_1$ and~${z_1 = x_1 + 2y_1}$ calls~$\SQ$ with values~$d_2$ and~$z_2 = x_2$ or~$\clean$ with values~$d_2$ and~$z_2 = x_2 + 2y_2$ and the results are used to combine results, then~$z_1 \sqrt{d_1} \log(d_1) \leq z_2 \sqrt{d_2} \log_{d_2}$ holds.
        \item The input to~$\base$ is a valid instance of~\Q{} and~$d \leq 2^{7\alpha_{\pi}+2}$.
    \end{itemize}
\end{restatable}
{
    \begin{proof}
    In the following, we show that if the invariant holds in a recursive call, then this recursive call will only invoke recursive calls for which the invariant also holds.
    Note that the initial call of $\SQ$ fulfills the requirement that the input is a valid instance of \Q{} by definition.
    Since $\base$ does not call any other algorithm, we only need to show that~$\SQ$ and~$\clean$ maintain the invariants. 
    
    \paragraph*{Main Algorithm ($\SQ$).}
    We start with recursive calls made in $\SQ$.
    To this end, assume that we are in a recursive call of $\SQ$ with input~$(d,G,x,B_{\he},B_{\li},B_{\di},S_{\pi})$, where~${x \leq 200\log^7(d)-22\log(d)}$ and~$(d,G,(B_{\he},B_{\li},B_{\di}),x,S_{\pi})$  corresponds to a valid instance of \Q.
    Note that~$\base$ is only called in \cref{lin:basecase} and only if the check in \cref{lin:basecaseif} confirms that~$G$ has at most~$d \leq 2^{7\alpha_{\pi}}$.
    It is called with the same input~$(d,G,x,B_{\he},B_{\li},B_{\di},S_{\pi})$, which is a valid instance of \Q{} by assumption.
    Hence, all requirements are met in this case.

    We next analyze the case where~$\SQ$ calls~$\clean$.
    Note that the only recursive call to $\clean$ is in \cref{lin:clean} with input~$(d,G,x,0,B_{\he},\emptyset,B_{\li},\emptyset,B_{\di},\emptyset,S_{\pi})$.
    In this case, all invariants of $\clean$ except~$d \geq 2^{7\alpha_{\pi}}$ and $|B_{\he} \cup B_{\li}| \leq (\nicefrac{3}{4})^0 \cdot 50d^6$ hold by the assumption that~$(d,G,(B_{\he},B_{\li},B_{\di}),x,S_{\pi})$ is a valid instance of \Q{}.
    Note that since we passed \cref{lin:basecaseif}, we may assume that~$d \geq 2^{7\alpha_{\pi}}$.
    Since~$(B,x)$ is almost measured for~$k$ and~$\alpha_{\pi}$, we have that~$x \leq 200\log^7(d)$ and therefore~$|B_{\he}| \leq 200\log^7(d)\cdot 40\alpha_{\pi}\sqrt{d} = 8000\alpha_{\pi}\sqrt{d}\log^7(d)$ and~${|B_{\li}| \leq 200\log^7(d)\cdot 60\alpha_{\pi}d\log(d) = 12000\alpha_{\pi}d\log^8{d}}$.
    These are at most~$20000\alpha_{\pi}d\log^7{d}$~vertices combined, which is upper bounded by~$31d^6$ by \cref{lem:help} as~$d>2^{7\alpha_{\pi}}\geq 100$.
    The third invariant also clearly holds here as~$z_1 = x < x+1 = x_2$ and~$d_1 = d = d_2$.
    
    Let next~$I'=(d',G',x',B'_{\he},B'_{\li},B'_{\di},S'_{\pi})$ be the input to a recursive call of $\SQ$ in \cref{lin:rec1} or~\ref{lin:rec2}.
    Recall that~$I'$ is a valid instance of \Q{} if
    \begin{multicols}{2}
        \begin{enumerate}
            \item $G'$ is a~$\alpha_{\pi}d'$-outerplanar graph,
            \item $1 \leq x' \leq \max(200 \log^7(d') - 44 \log(d),1)$,
            \item $|B_{\he}'| \leq 40x\alpha_{\pi}\sqrt{d}$, and
            \item $|B'_{\li}| \leq 60 x \alpha_{\pi}d\log(d)$.
        \end{enumerate}
    \end{multicols}
    Since~$d' = d$ and~$G'$ is a subgraph of~$G$, the first point trivially holds.
    Note that since we reached past \cref{lin:cleanif,lin:clean}, it holds that~${1 \leq x' = x+1 \leq 40 \log^7(d) + 1 \leq 200 \log^7(d)- 22 \log(d)}$.
    As we only potentially removed vertices from~$B_{\he}$, the third point also holds trivially.
    Finally, since we only added~$2\alpha_{\pi}d+3$ vertices to~$B_{\li}$, we have that~$|B_{\li}'| = |B_{\li} \cup C_{\he}| \leq |B_{\li}| + |C_{\li}| \leq 40x\alpha_{\pi}d\log(d) + 2\alpha_{\pi}d+3 \leq 40(x+1)\alpha_{\pi}d\log(d)$, which shows the fourth condition.
    The third invariant again trivially holds as~$z_1 < z_2$ and~$d_1 = d_2$.
    
    Let now~$I'=(d',G',x',B'_{\he},B'_{\li},B'_{\di},S'_{\pi})$ be the input to a recursive call in \cref{lin:rec3} or~\ref{lin:rec4}.
    Note that the application of \cref{lem:baker} in \cref{lin:baker1,lin:baker2} ensures that $C'_{\inte} - A_{\inte}$ and~$C'_{\ex}-A_{\ex}$ are~$d'\alpha_{\pi}$-outerplanar graphs.
    Moreover, since we passed \cref{lin:cleanif,lin:clean}, we have~$x \leq 40 \log^7(d)$ implying
    \begin{align*}
        &1 \leq x' = 2x+2 \leq 2 \cdot 40\log^7(d)+2\\
        &\leq 80\log^7((\nicefrac{7}{4})d') + 2\\
        &= 80 (\log(d') + \log(\nicefrac{7}{4}))^7 + 2 \\
        &< 80 (\log(d') + 1)^7 + 2\\
        &= 80 (\log^7(d') + 7\log^6(d') + 21\log^5(d') + 35\log^4(d') + 35\log^3(d') + 21\log^2(d') + 7\log(d') + 1) + 2\\
        &\leq 80 \log^7(d') (1 + \nicefrac{7}{7} + \nicefrac{21}{7^2} + \nicefrac{35}{7^3} + \nicefrac{35}{7^4}+\nicefrac{21}{7^5}+\nicefrac{7}{7^6}+\nicefrac{3}{7^7})\\
        &\leq 80\cdot 2.45 \log^7(d')\leq 200\log^7(d')- 44\log(d'),
    \end{align*}
    where the third and second from last inequality hold as~$\log(d) \geq 7$ as~$d \geq 2^{7\alpha_{\pi}} \geq 2^7$.
    Moreover, \cref{lem:unknownbalance,lem:help} yield
    \begin{align*}
        |B'_{\he}| &\leq |B_{\he} \cup C'_{\he}| \leq 40 \alpha_{\pi} x \sqrt{d} + 40 \sqrt{\alpha_{\pi}d} \leq 40(x+1)\alpha_{\pi}\sqrt{d}\\
        &\leq 40(2x+2)\alpha_{\pi}\sqrt{\frac{d}{4}} \leq 40x'\alpha_{\pi}\sqrt{\frac{4d}{7}}\\
        &\leq 40x'\alpha_{\pi}\sqrt{d'}
    \end{align*}
    We will next show that~$\frac{7\alpha_{\pi}}{4} \leq (\frac{4d}{7})^{\frac{1}{7}}$ or equivalently~$d \geq \alpha_{\pi}^7 (\frac{7}{4})^6$.
    Note that when~$\alpha_{\pi}=1$, then~${d \geq 2^7 \geq 32 \geq \frac{7}{4}^6 = \frac{7}{4}^6 \alpha_{\pi}^7}$.
    For all positive integers~$\alpha_{\pi} \geq 2$, it holds that~$\alpha_{\pi} \leq 2^{\alpha_{\pi}-1}$.
    Hence,~${\alpha_{\pi}^7 (\frac{7}{4})^6 \leq (2^{\alpha_{\pi}-1})^7 \cdot 32 = 2^{7\alpha_{\pi}-7+5} \leq d}$.
    Combined with \cref{lem:unknownbalance,lem:help}, this implies that
    \begin{align*}
        |B'_{\li}| &\leq |B_{\li} \cup C'_{\li}| \leq 60 \alpha_{\pi} x d \log(d) + 60 \alpha_{\pi} d \log(\alpha_{\pi}d)\\
        &\leq 60(x+1)\alpha_{\pi}d\log(\alpha_{\pi}d)\\
        &\leq 30(2x+2)\frac{7}{4}\alpha_{\pi}\frac{4d}{7}\log(\frac{7\alpha_{\pi}}{4}\frac{4d}{7})\\
        &\leq 30(2x+2)\frac{7}{4}\alpha_{\pi}\frac{4d}{7}(\log(\frac{4d}{7})^{\nicefrac{8}{7}})\\
        &\leq 30x'\frac{7\cdot 8}{4 \cdot 7}d'\log(d')\\
        &\leq 60 x' d' \log(d').
    \end{align*}
    To show that the third invariant is also maintained in this case, note that~$\frac{4d}{7} \geq d^{\nicefrac{3}{4}}$ for all~$d>10$.
    Hence, we have
    \begin{align*}
        x'\sqrt{d'}\log(d') &\geq (2x+2)\sqrt{\frac{4d}{7}}\log(\frac{4d}{7})\\
        &\geq 2x\sqrt{\frac{4}{7}}\sqrt{d}\log(d^{\nicefrac{3}{4}})\\
        &\geq x\sqrt{d}\log(d)(2\sqrt{\frac{4}{7}}\frac{3}{4})\\
        &\geq x \sqrt{d} \log(d).
    \end{align*}    
    
    \paragraph*{Reduce Portals ($\clean$).}
    We now consider a recursive call of $\clean$ with input~$(d,G,x,y,X_{\he},Y_{\he},X_{\li},Y_{\li},X_{\di},Y_{\di},S_{\pi})$ which satisfies our invariant and show that the input to each call of $\SQ$ and $\clean$ also satisfies the respective invariants.
    Observe that $\SQ$ is only called in \cref{lin:solve} and with input~$(d,G,y+1,X_{\he}\cup Y_{\he},X_{\li}\cup Y_{\li},X_{\di}\cup Y_{\di},S_{\pi})$, which is by assumption a valid instance of \Q{} since~$x \geq 1$ and~$|X_{\he} \cup X_{\li}| \leq 20 \sqrt{d}$.
    Moreover, the result is never used to combine solutions using Requirement R4, so the third invariant is not applicable here.
    
    We now consider the input~$(d',G',x',y',X'_{\he},Y'_{\he},X'_{\li},Y'_{\li},X'_{\di},Y'_{\di},S'_{\pi})$ given to $\clean$ in \cref{lin:rec5,lin:rec6}.
    We need to show that
    \begin{multicols}{2}
        \begin{enumerate}
            \item $G'$ is a~$\alpha_{\pi}d'$-outerplanar graph,
            \item $1 \leq x' \leq 200 \log^7(d')-22\log(d')$,
            \item $y' < 22 \log(d')$,
            \item $|X'_{\he}| \leq 40 x' \alpha_{\pi} \sqrt{d'}$,
            \item $|Y'_{\he}| \leq 40 (2y') \alpha_{\pi} \sqrt{d'}$,
            \item $|X'_{\li}| \leq 60 x'\alpha_{\pi} d' \log(d')$,
            \item $|Y'_{\li}| \leq 60 (2y')\alpha_{\pi} d' \log(d')$,
            \item $d' \geq 2^{7\alpha_{\pi}}$,
            \item $|X'_{\he} \cup X_{\li}| \leq (\nicefrac{3}{4})^{y'} \cdot 50d'^6$, and
            \item $(x+2y) \sqrt{d}\log(d) \leq (x'+2y') \sqrt{d'} \log(d')$.
        \end{enumerate}
    \end{multicols}
    Note that points~$1,2,4,6,8,$ and~$10$ hold as~$d'=d, x' = x, y'=y+1, {X'_{\he} \subseteq X_{\he}}, {X'_{\li} \subseteq X_{\li}}$, and~$G'$~is a subgraph of~$G$.
    For point 9, note that~$|X_{\he}' \cup X'_{\li}| \leq \nicefrac{3}{4}|X_{\he} \cup X_{\li}|$ as~$C'$ is $\nicefrac{3}{4}$-balanced with respect to the weight function that assign weight~$\nicefrac{1}{|X_{\he} \cup X_{\li}|}$ to each vertex in~$X_{\he} \cup X_{\li}$ and zero to all other vertices.
    Hence, \[|X'_{\he} \cup X'_{\li}| \leq (\nicefrac{3}{4})|X_{\he} \cup X_{\li}| \leq \nicefrac{3}{4} \cdot (\nicefrac{3}{4})^{y} \cdot 50d^6 = (\nicefrac{3}{4})^{y+1} \cdot 50d^6 = (\nicefrac{3}{4})^{y'} \cdot 50d'^6.\]
    For the third point, notice that since~$|X_{\he} \cup X_{\li}| \leq (\nicefrac{3}{4})^{y} \cdot 50d^6$ by assumption and~$|X_{\he} \cup X_{\li}| \geq 1$ as we passed \cref{lin:solveif}, it holds that~$(\nicefrac{3}{4})^y\cdot 50d^6 \geq 1$, which implies~${y \leq \frac{-\log(50d^6)}{\log(\nicefrac{3}{4})} \leq \frac{-\log(d^7)}{-\nicefrac{1}{3}} \leq 21 \log(d)}$.
    Thus, it holds that~${y' = y+1 < 22 \log(d)}$.
    Using \cref{lem:unknownbalance}, we next get that
    \begin{align*}
        |Y_{\he}'| &\leq |Y_{\he} \cup C'_{\he}|\\
        &\leq 40(2y)\alpha_{\pi}\sqrt{d}+40 \sqrt{\alpha_{\pi}d}\\
        &\leq 40(2y+1)\alpha_{\pi}\sqrt{d} < 40(2y')\alpha_{\pi}\sqrt{d} \text{ and}
    \end{align*}
    \begin{align*}
        |Y_{\li}'| &\leq |Y_{\li} \cup C'_{\li}| \leq 60 (2y)\alpha_{\pi} d \log(d) + 60 \alpha_{\pi} d \log(\alpha_{\pi}d)\\
        &\leq 60 (2y)\alpha_{\pi} d \log(d) + 60 \alpha_{\pi}d \log(d^2)\\
        &\leq 60 (2y+2)\alpha_{\pi} d \log(d)\\
        &\leq 60 (2y')\alpha_{\pi} d \log(d).
    \end{align*}
    This concludes the proof.
\end{proof}
}

\subsection{Base Case}
Here, we prove that $\base$ works as intended.
Before doing so, let us give an intuitive description of what the algorithm does.
As stated in \cref{lem:invariants}, the algorithm is only called when~$d$ is constant (because~$\alpha_{\pi}$ is a constant).
Recall that for a node~$u$ of the tree decomposition, $X_u$ denotes the vertices in the bag of~$u$ and~$Y_u$ denotes the union of~$X_{v}$ for all descendants of~$v$ (including~$u$).
We use standard dynamic programming ideas on the tree decomposition to solve request for pebble sets \emph{of any size} where we consider all vertices in the bag of the current node as heavy vertices (in addition to the input portal set~$B$); this can be viewed as setting~$x = \infty = d$.
In order to improve the polynomial part of the running time, we will consider balanced binary tree decompositions, that is, binary tree decompositions of depth~$\Oh(\log n)$.

For the dynamic program over the balanced binary tree decomposition, we first solve certain small instances using brute force and then use Requirement R4 to combine solutions.
The small instances are those induced by vertices in~$X_u \cup X_{u'}$ for any node~$u$ (except the root) in the tree decomposition and where~$u'$ is the parent of~$u$.
Since each bag has constant size and we assume that \Q{} is computable, this results in a constant-time algorithm for each edge in the tree decomposition.
We then combine solutions as follows.
For each node~$u$ and each child~$v_i$ of~$u$, we combine the solutions for~$G[Y_{v_i}]$ and~$G[X_u \cup X_{v_i}]$ using~$X_{v_i}$ as a separator consisting of heavy vertices using Requirement R4.
We call these instances~$I^1_u$ (for~$v_1$) and~$J_u$ (for~$v_2$).
Therein, we keep the vertices in~$X_{u}$ as heavy portal vertices.
If~$u$ has two children, then  we combine the solutions for both instances to get a solution for~$G[Y_{u}]$.
After this is done for all nodes, we are left with two minor inconveniences.
On the one hand, the vertices in~$X_p$ ($p$ is the root of the tree decomposition) are currently treated as heavy portal vertices.
However, using Requirement R4, we can view~$X_p \setminus B$ as a separator for the separation~$(G,X_p \setminus B)$, that is, no vertices are separated from~$(V \setminus X_p) \cup B$.
Applying Requirement R4 then allows us to drop the vertices in~$X_p \setminus B$ from the portal set.
A key step in the proof will be to show that~$((X_u \cup X_{u'},\emptyset,\emptyset),1)$ is measured for~$4d^2$ and~$\alpha_{\pi}$ and hence Requirement R4 can be applied.
On the other hand, we now have a lookup table storing \emph{too much} information.
We did not yet consider the size of possible pebble sets as we imagined~$x = \infty = d$.
However, this is trivial to fix.
We simple iterate over each pebble set for~$(B,d,\alpha,x)$, compute all requests for it and copy the solution for it.
We will show that all of the above can be done in~$\Oh(n^{\rho}\log(n))$ time assuming~$d$ is a constant.

\begin{lemma}
    \label{lem:basecase}
    Let~$(d,G=(V,E),x,B_{\he},B_{\li},B_{\di},S_{\pi})$ be the input given to~$\base$ satisfying the requirements of \cref{lem:invariants} and let~\Q{} satisfy all requirements of \cref{thm:main-fpt}.
    Then, $\base$ solves the instance~$(d,G,(B_{\he},B_{\li},B_{\di}),x,S_{\pi})$ of \Q{} in~$\Oh(n^{\rho}\log(n))$ time.
\end{lemma}

\begin{proof}
    We will use the same notation as in \cref{alg:base}.
    We first prove that for each separator~${C=(C_{\he},C_{\li},C_{\di})}$ chosen by $\base$, it holds that~$(C,1)$ is almost measured with respect to~$D = (2\alpha_{\pi}d)^2$ and~$\alpha_\pi \geq 1$.
    Note that the largest separator consists of the vertices of two bags in the binary tree decomposition~$\mathcal{T}$.
    We next show that each bag contains at most~$18\alpha_{\pi}d$ vertices and hence, the largest separator contains~$36\alpha_{\pi}d$ vertices.
    By \cref{lem:invariants}, the graph~$G$ is~$\alpha_{\pi}d$-outerplanar.
    \cref{lem:bodlaender} then states that it has treewidth at most~$3\alpha_{\pi}d-1$.
    \cref{thm:tuukka} then allows us to compute a tree decomposition of width at most~$2(3\alpha_{\pi}d-1)+1 = 6\alpha_{\pi}d-1$ and \cref{lem:BH} allows us to compute a binary tree decomposition of width at most~$3(6\alpha_{\pi}d-1)+2 = 18 \alpha_{\pi}d -1$ and depth~$\Oh(n)$.
    Hence, each bag contains at most~$18\alpha_{\pi}d$ vertices.
    Recall that~$(C,1)$ is almost measured for~$D$ and~$\alpha_{\pi}$ if~$1 \leq \max(1,200\log^7(d))$, $|C_{\he}| \leq 40 \alpha_{\pi}\sqrt{D}$, and~$|C_{\li}| 60 \alpha_{\pi} \sqrt{D}$.
    The first point holds trivially and since~$|C_{\he} \cup C_{\li} \cup C_{\di}| \leq 36 \alpha_{\pi} d < 40 \alpha_{\pi} \sqrt{D} = 80 \alpha_{\pi} d$, the other two requirements are also met and~$(C,1)$ is measured with respect to~$D$ and~$\alpha_{\pi}$ and Requirement R4 can be applied in all cases.

    Requirement R4 then guarantees that all solutions for pebble sets that do not intersect~$C_{\di}$ and that intersect~$C_{\li}$ in at most~$20\alpha_{\pi}\sqrt{D}\log(D)$ vertices are computed correctly.
    We next show that these are all subsets of~$C_{\he} \cup C_{\li}$.
    To this end, we show~$|C_{\he} \cup C_{\li}| \leq 20\alpha_{\pi}\sqrt{D}\log(D)$.
    Since~${|C_{\he} \cup C_{\li}| \leq |C_{\he} \cup C_{\li} \cup C_{\di}|}$, we can reuse the above argument to conclude the following: $|C_{\he} \cup C_{\li}| < 36\alpha_{\pi}d < 40 \alpha_{\pi} d = 20 \alpha_{\pi} \sqrt{D} < 20 \alpha_{\pi} \sqrt{D} \log(D)$.
    Hence, in each case where we apply Requirement R4, it correctly returns solutions for all subsets of heavy and light vertices of the separator independent of their size.

    We next show that $\base$ correctly computes the result for the instance of \Q{} it is given as input and afterwards analyze the running time.
    The correctness of the brute-force step in \cref{lin:bruteforcebase} follows from the fact that \Q{} is computable.
    Since the correctness of each single combination of results follows from Requirement R4, it only remains to show that Requirement R4 is applied in each step so that it gives the result for the instance it is supposed to solve.
    To this end, we make a case distinction whether a node~$u$ has 0, 1, or two children (it cannot have more as~$\mathcal{T}$ is binary).
    If it has 0 children, then we do not apply Requirement R4, so nothing needs to be shown.
    If~$u$ has exactly one child~$v_1$, then we already computed the solutions for instances~$I_{v_1}$ and~$I'_{v_1}$.
    Note that~$I_{v_1}$ is the instance for~$G[Y_{v_1}]$ and~$I'_{v_1}$ is the instance for~$G[X_{v_1} \cup X_{u}]$.
    Since these two are separated by~$X_{v_1}$ since $\mathcal{T}$ is a tree decomposition, we correctly solve the instance for~$G[Y_{v_1} \cup X_u] = G[Y_u]$ as~$u$ only has~$v_1$ as a child, which is precisely~$I_u = I_u^1$.
    If~$u$ has two children~$v_1$ and~$v_2$, then the same argument as above shows that we correctly compute solutions for instances~$I_u^1$ (for~$G[Y_{v_1} \cup X_u]$) and~$J_u$ (for~$G[Y_{v_2} \cup X_u]$).
    As these two are separated by~$X_u$, we correctly compute the solution for the instance~$I_u^2 = I_u$ for~$G[Y_{v_1} \cup X_u \cup Y_{v_2}] = G[Y_u]$ in \cref{lin:I2}.
    In \cref{lin:remove}, we only remove some heavy vertices from the portal set by using a trivial separator~$(X_u \setminus (B_{\li} \cup B_{\di}), \emptyset,\emptyset)$.
    Finally, we output all solutions for pebble sets for~$(B,d,\alpha_{\pi},x)$.
    As argued above, all of these were already correctly computed before, so the output of $\base$ is correct.

    It remains to analyze the running time.
    We will do so in three steps.
    First, we will analyze the running time outside the for-loops in \cref{lin:for1,lin:for2}.
    Second, we will show that the time spend for each node~$u$ in the tree decomposition~$\mathcal{T}$ computed in \cref{lin:BH} is in~$\Oh(|Y_u|^{\rho})$.
    Third, we will argue that summed over all nodes in~$\mathcal{T}$, this is still contained in~$\Oh(n^\rho \log(n))$.
    For the first step, note that \cref{lem:invariants} states that~$d$ is a constant (as~$\alpha_{\pi}$ is a constant).
    This implies that~$D$ is a constant and by the definition of request variants, $x$ is at most~$200 \log^7(d)$ and is therefore also a constant.
    Finally, this yields that~$|B_{\he}|$ and~$|B_{\li}|$ are constant.
    Hence, the number of iterations in \cref{lin:pebblefor} is a constant.
    It is now easy to verify that each single line outside the for-loops in \cref{lin:for1,lin:for2} take at most~$\Oh(\gamma(D)n^{\rho}) = \Oh(n^{\rho})$ time.

    For the second step, note that in an iteration of the two for-loops for a node~$u$ in~$\mathcal{T}$, each single line takes at most~$\Oh(\gamma(D) |Y_u|^{\rho})$ time as subgraphs can be computed in linear time and updating~$S_{\pi}$ and computing~$R_{\pi}$ take~$\Oh(\gamma(D) |Y_u|^{\rho})$ time by requirements R2 and R3.
    The use of Requirement~R4 also takes~$\Oh(\gamma(D) |Y_u|^{\rho})$ time and brute-forcing the solutions in \cref{lin:bruteforcebase} takes constant time as each bag contains a constant number of vertices.

    It remains to analyze the running time for nodes in~$\mathcal{T}$.
    To this end, we consider all nodes at a given depth~$t$ together.
    Let~$U_t$ be the set of all these nodes and note that we can assume without loss of generality that the number of nodes in~$\mathcal{T}$ is at most~$n$ and thus~$|U_t|\leq n$.
    Note that since~$\mathcal{T}$ is a tree decomposition, it holds that~${\sum_{u \in U_t} |Y_u| \leq n + \sum_{u \in U_t} |X_u| \leq n (1+18\alpha_{\pi}d) \in \Oh(n)}$.
    Hence, the running time for all nodes at depth~$t$ combined is in~$\Oh(n^{\rho})$ and summed over the~$\Oh(\log(n))$ different depths, this yields the final running time of~$\Oh(n^{\rho}\log(n))$.    
\end{proof}

\subsection{Reduce Portals}

In this section, we show that $\clean$ works as intended, that it, it reduces an instance of~\Q{} with~$x \geq 40\log^7(d)$ to~$2^{\Oh(\log^3(d))}$~instances of~\Q{} where~$x \leq 22\log(d)$ holds in each instance.
We start with a description of the main idea of the algorithm.
We are given a partition of the portal sets~$(B_{\he},B_{\li},B_{\di})$ into sets~$(X_{\he},X_{\li},X_{\di})$ and~$(Y_{\he},Y_{\li},Y_{\di})$.
Initially~$X_i = B_i$ and~$Y_i = \emptyset$ for each~$i \in \{\he,\li,\di\}$ and the goal is to call~$\SQ$ with a small set~$B_{\he} \cup B_{\li}$.
We do this by iteratively finding a~$\nicefrac{3}{4}$-balanced cycle separator for the weight function that assigns weight~$\nicefrac{1}{|X_{\he} \cup X_{\li}|}$ to each vertex in~$X=X_{\he} \cup X_{\li}$ and zero to all other vertices.
For this cycle separator, we apply \cref{lem:unknownbalance} to find a set of cycle separators that are balanced for the mentioned weight function.
We then guess the cycle separator~$C'$ that does not intersect the solution in too many vertices, add its heavy vertices~$C'_{\he}$ to~$Y_{\he}$, its light vertices~$C'_{\li}$ to~$Y_{\li}$, and its discarded vertices~$C'_{\di}$ to~$Y_{\di}$ and call~$\clean$ on the interior and the exterior.
Clearly, the set~$X$ shrinks by a factor of~$\nicefrac{3}{4}$ in each iteration.
Hence, after calling~$\clean$ for~$\Oh(\log(|X|))$ times, the set~$X$ will have constant size.
We also show that the set~$Y$ does not become too large in the process.
Hence, we can then call~$\SQ$ with~$B_i = X_i \cup Y_i$ for each~$i \in \{\he,\li,\di\}$ (modeled as an oracle in the following lemma).

\begin{lemma}
    \label{lem:clean}
    Let~$(d,G=(V,E),x,y,X_{\he},Y_{\he},X_{\li},Y_{\li},X_{\di},Y_{\di},S_{\pi})$ be given to~$\clean$ as input, and let it satisfy the requirements of \cref{lem:invariants}. Let also \Q{} satisfy the requirements of \cref{thm:main-fpt}.
    Given access to an oracle for instances~${(d,G'=(V',E'),B',x',S'_{\pi})}$ with~${|V'| < |V|}$ and~${x' \leq 44\log(d)}$, $\clean$ solves the instance~$(d,G,(X_{\he}\cup Y_{\he},X_{\li}\cup Y_{\li},X_{\di}\cup Y_{\di}),x+2y,S_{\pi})$ of \Q{} with at most $2^{2000\log^3(d)}$ oracle calls in~$2^{\tilde{\Oh}(\sqrt{d})} \gamma(d) n^{\max(\rho,2)}$~time.
\end{lemma}

\begin{proof}
    We will prove the statement by induction over~$z = 22\log(d) - y$.
    If~$z = 1$, then we have~${y= 22\log(d) - 1}$ and \cref{lem:invariants} yields~$2y+1 \leq 44\log(d)$ and 
    $$|X_{\he} \cup X_{\li}| \leq (\nicefrac{3}{4})^{22 \log(d) - 1} \cdot 50d^6 \leq \nicefrac{4}{3} \cdot 2^{22\log(\nicefrac{3}{4})\log(d)+7\log(d)} \leq \nicefrac{4}{3} \cdot 2^{- 2\log(d)} \leq 1$$ and \cref{lin:solve} is therefore reached.
    Note that if \cref{lin:solve} is reached, then the oracle calls give the correct solution for the subinstances by assumption as~$2y+1 \leq 44\log(d)$.
    
    If \cref{lin:solve} is not reached, then we let~$\mathcal{S}$ be the set of cycle separators computed in \cref{lin:cleanbalance}.
    We call $\clean$ for each separator in~$\mathcal{S}$ and in each call, the value of~$y$ is increased by~$1$, that is, the value of~$z$ is reduced by one.
    Hence, the induction hypothesis states that the recursive calls in \cref{lin:rec5,lin:rec6} are computed correctly.
    For each cycle separator~$C \in \mathcal{S}$, it holds by \cref{lem:unknownbalance} that~$(C,1)$ is measured for~$d$ and~$1$ and therefore also for~$d$ and~$\alpha_\pi$.
    Thus, we can apply Requirement~R4 to guarantee that we return the correct solution.
    This concludes the proof of correctness.
    
    We next analyze the running time and the number of oracle calls.
    Consider any recursive call of~$\clean$.
    By \cref{lem:unknownbalance}, this calls~$\clean$ at most~$2 \cdot 2^{22\log^2(2\alpha_{\pi}d+3)} \leq 2^{88\log^2(d)}$ times and in each such call, the value of~$y$ is increased by one.
    Hence, the number of recursive calls of~$\clean$ corresponds to the number of vertices in a tree of depth~$22 \log(d) - 1$ and degree~$2^{88\log^2(d)}$.
    Hence, the number of oracle calls (calls of $\SQ$ and the number of leaves in the tree) is at most
    \begin{align*}
        ({2^{88\log^2(d)}})^{(22\log(d)-1+1)} &\leq 2^{88\log^2(d) \cdot 22\log(d)}
        \leq 2^{2000\log^3(d)}
    \end{align*}
    and the number of recursive calls of~$\clean$ (the number of internal vertices in the tree) is at most~$2^{2000\log^3(d)}-1$.

    It remains to analyze the running time of each recursive call.
    Note that each single line except for \cref{lin:cleanini,lin:cleanbalance,lin:cleans,lin:cleanr,lin:upclean} can be executed in~$\Oh(n^2)$ time.
    
    \cref{lin:cleanbalance} takes~${2^{\Oh(\log^2(\alpha_{\pi} d))} n^{1+o(1)}}\subseteq 2^{\Oh(\log^2(d))} n^2$ time by \cref{lem:unknownbalance}.
    Since~$|V'|\leq n$ and~$d'=d$ for each recursive call, Requirements R2, R3, and~R4 state that \cref{lin:cleans,lin:cleanr,lin:upclean} each take~$\Oh(\gamma(d)n^{\rho})$ time.
    Since the number of separators in~$\mathcal{S}$ is at most~$2^{44\log^2(2d\alpha_{\pi}+3)} \in 2^{\tilde{\Oh}(1)}$ and the number of possible pebble sets for~$((X_{\he} \cup Y_{\he},X_{\he} \cup Y_{\he},X_{\he} \cup Y_{\he}),x+2y)$ is by \cref{lem:invariants} at most
    \begin{align*}
        2^{40(200\log^7(d))\alpha_{\pi}\sqrt{d}} \cdot \binom{60(200\log^7(d))\alpha_{\pi}d\log(d)}{20(200\log^7(d))\alpha_{\pi}\sqrt{d}\log(d)} &\in 2^{\tilde{\Oh}(\sqrt{d})}\\(\text{since }|X_{\he} \cup Y_{\he}| \leq 40(200\log^7(d))\alpha_{\pi}\sqrt{d} \text{ and }|X_{\li} \cup Y_{\li}| &\leq 60(200\log^7(d))\alpha_{\pi}d\log(d)),
    \end{align*}
    we have that the running time for one recursive call is in~$2^{\tilde{\Oh}(\sqrt{d})} \gamma(d) n^{\max(2,\rho)}$.
\end{proof}

\subsection{Main Algorithm}
We are finally in a position to prove our main theorem by showing that $\SQ$ solves \Q.
We start with a description of the algorithm and then proceed with the formal proof.

Following the intuition presented in the introduction, $\SQ$ works by first checking whether we can solve the base case.
If~$|V| \leq 2^{7\alpha_{\pi}}$, then it can afford to brute-force the solution and if~$d \leq 2^{7\alpha_{\pi}}$, then it calls~$\base$.
It then checks whether~$(B,x)$ is measured for~$d$ and~$\alpha_\pi$, that is, whether the size of~$B$ and the value of~$x$ are small enough to proceed or whether $\clean$ has to be called.
If $(B,x)$ is measured, then $\SQ$ finds a cycle separator~$C$ and guesses whether~$C$ contains at most~$20\sqrt{d}\log(d)$~vertices of the (unknown) solution~$Z$ or more.
In the first case, it solves the corresponding instances in the interior and exterior of~$C$ and combines the results (Requirement R4 ensures that this can be done).
Otherwise, it uses the fact that~$C$ contains many vertices of the solution to reduce the value of~$d$ as follows.
Since~$C$ contains at least~$20\sqrt{d}\log(d)$~vertices of~$Z$, both the strict interior and the strict exterior of~$C$ contain at most~$|Z|-20 \sqrt{d} \log(d)$~vertices.
Hence~$C$ is~$\frac{|Z|-20 \sqrt{d} \log(d)}{|Z|}$-balanced for the weight function that assigns weight~$\nicefrac{1}{|Z|}$ to each vertex in~$Z$ and zero to all other vertices.
Using \cref{lem:unknownbalance}, the algorithm computes a set of cycle separators where one of those is also~$\frac{|Z|-20 \sqrt{d} \log(d)}{|Z|}$-balanced for the mentioned weight function and does not intersect the solution in more than~$17 \sqrt{d} \log(d)$~vertices.
Thus, the algorithm applies \cref{lem:baker} (which is permitted by Requirement R1) to compute a~$\lfloor d - 3\sqrt{d} \rfloor$-outerplanar subgraph~$G'$ of both the interior and the exterior of each computed cycle separator, solves the problem recursively, and combines the solutions; again Requirement R4 allows us to combine solutions.
Intuitively, we make progress in each recursive call as we either reduce the number of vertices by a constant factor or reduce the parameter~$d$ by roughly~$3\sqrt{d}$.
We will prove that this is sufficient for the claimed running time and now proceed with the formal proof of \cref{thm:main-fpt}, which we restate here for convenience. 
\mainalgorithm*
\begin{proof}
    We first show that $\SQ$ solves the instance~$(d,G,x,B=(B_{\he},B_{\li},B_{\di}),S_{\pi})$ of~\Q{} it is given as input.
    We will later show that $\SQ$ terminates after a finite number of recursive calls.
    We can therefore show the claim by induction and assume that all recursive subcalls of $\SQ$ indeed solve the corresponding instance of \Q{} correctly.
    For the sake of notational convenience, we will assume that~$\rho \geq 2$.

    If \cref{lin:bruteforcesmalln} is reached, then the computation is correct as we assume \Q{} to be computable.
    If \cref{lin:basecase} is reached, then the computation is correct by \cref{lem:basecase}.
    If \cref{lin:clean} is reached, then \cref{lem:clean} shows that the solution is correctly computed.

    So assume that \cref{lin:sqini} is reached and let~$C$ be the cycle separator computed in \cref{lin:cycle} and let~$\mathcal{S}$ be the set of cycle separators computed in \cref{lin:balance}.
    Let~$\peb$ be a pebble set for~$(B,x)$ and let~$r \in R_\pi(d,G,B,\peb,S_{\pi})$ be a request.
    If the (unknown) solution~$Z$ that determines the optimal value for request~$r$ intersects~$C$ in at most~$20\sqrt{d}\log(d)$ vertices, then Requirement R4 guarantees that the solution is computed correctly in \cref{lin:upSQ,lin:upSQ2} as the solutions~$\res_1$ and~$\res_2$ are correct by the induction hypothesis.
    Moreover, as we only update~$\res$ with valid solutions, the result is not overwritten in a later iteration of \cref{lin:upSQ2}.
    If~$Z$ and~$C$ intersect in more than~$20\sqrt{d}\log(d)$ vertices, then both the strict interior and the strict exterior of~$C$ contain at most~$d - 20 \sqrt{d} \log(d)$ vertices of~$Z$, that is,~$C$ is~$\frac{|Z|-20 \sqrt{d} \log(d)}{|Z|}$-balanced with respect to the weight function that assigns weight~$\nicefrac{1}{|Z|}$ to each vertex in~$Z$ and zero to all other vertices.
    By \cref{lem:unknownbalance}, there exists a~$\frac{|Z|-20 \sqrt{d} \log(d)}{|Z|}$-balanced cycle separator~$C' \in \mathcal{S}$ (computed in \cref{lin:balance}) whose intersection with~$Z$ is of size at most
    \begin{align*}
        40 \sqrt{2\alpha_{\pi}d+3} + \frac{11|Z|\log(\alpha_{\pi} d)}{2\sqrt{2\alpha_{\pi}d+3}}
        \leq&\ 40 \sqrt{\frac{5d\alpha_{\pi}}{2}} + \frac{11 d \log(d^2)}{2 \sqrt{2d}}\\
        \leq&\ 40 \sqrt{\frac{5d\log(d)}{14}} + \frac{11 d \log(d)}{\sqrt{2d}}\\
        \leq&\ (40 \sqrt{\frac{5}{14\log(d)}} + \frac{11}{\sqrt{2}}) \sqrt{d} \log(d)\\
        \leq&\ (\frac{40}{7}\sqrt{\frac{5}{2}}+\frac{11}{\sqrt{2}})\leq 17 \sqrt{d} \log(d).
    \end{align*}
    Here, the second inequality is due to the fact that~$\alpha_{\pi} \leq \frac{\log(d)}{7}$ as~$d \geq 2^{7\alpha_{\pi}}$ and the second to last inequality is due to the fact that~$\log(d)\geq 7$ as~$d>2^7$.
    The (non-strict) interior and exterior of~$C'$ hence contain at most 
    \[\lfloor |Z| - 20 \sqrt{d} \log(d) + 17 \sqrt{d} \log(d) \rfloor \leq \lfloor d - 3\sqrt{d}\log(d) \rfloor \leq \lfloor d-3\sqrt{d} \rfloor \text{ vertices from~$Z$.}\]
    Due to the induction hypothesis, the solutions~$\res_1$ and~$\res_2$ in \cref{lin:rec3,lin:rec4} are correct.
    It only remains to show that the algorithm behind Requirement R4 can combine them.
    This is slightly non-trivial in this case as the value for~$d$ is reduced in the recursive calls, so we need to argue that all pebble sets for~$(B,d,\alpha_\pi,x)$ are also pebble sets for~$(B,d'=\lfloor d - 3\sqrt{d} \rfloor,\alpha_\pi,x'=2x+1)$, that is, the algorithm can access all required results.
    By definition, each pebble set for~$(B,d,\alpha_\pi,x)$ contains at most~$20x\alpha_{\pi}\sqrt{d}\log(d)$ vertices in~$B_{\li}$.
    By \cref{lem:invariants}, we have that~$x \sqrt{d} \log(d) \leq x' \sqrt{d'} \log(d')$ and each such pebble set therefore contains at most~$20x'\alpha_{\pi}\sqrt{d'}\log(d')$ vertices from~$B_{\li}$.
    Thus, each such pebble set is also a pebble set for~$(B,d',\alpha_{\pi},x')$ and Requirement R4 applies.
    This concludes the proof of correctness.

    It remains to analyze the running time of $\SQ$.
    To this end, let~$\varepsilon > 0$ be any constant.
    We start with the running time of one recursive call.
    In the following, we will again use the fact that~$\alpha_{\pi} \leq \frac{\log(d)}{7}$.
    First, note that the number of possible pebble sets is upper bounded by
    \begin{align*}
    2^{|B_{\he}|}\cdot \binom{|B_{\li}|}{20x\alpha_{\pi}\sqrt{d}\log(d)} &\leq 2^{40(200\log^7(d))\alpha_{\pi}\sqrt{d}} \cdot 2^{(20(200\log^7(d)\alpha_{\pi}\sqrt{d}\log(d))\cdot \log(60(200\log^7(d))\alpha_{\pi}d\log(d))}\\
    &\leq 2^{\frac{8000}{7}\sqrt{d}\log^8(d)+\frac{3200}{7}\sqrt{d}\log^8(d)\log(d^8)}\\
    &\leq 2^{(\frac{8000}{7\log^2(d)}+\frac{8\cdot 3200}{7})\sqrt{d}\log^9(d)}\\
    &\leq 2^{3681\sqrt{d}\log^6(d)}.
    \end{align*}
    The number of separators in~$\mathcal{S}$ is upper-bounded by~$2^{44\log^2(2\alpha_{\pi}d+3)} \leq 2^{44\log^2(d^{1.5})} = 2^{99\log^2(d)}$ by \cref{lem:unknownbalance}.
    This bounds the number of iterations of the for-loops in one recursive call.
    If \cref{lin:bruteforcesmalln} is reached, then the number of operations is constant.
    By \cref{lem:baker,lem:unknownbalance}, \cref{lem:cycle}, and Requirements~R2, R3, and~R4, all lines can be computed in~$\Oh(\gamma(d) n^{\max(2,\rho)}\log(n))\subset\Oh(\gamma(d) n^{\rho+\varepsilon})$ time each.
    Thus, the running time of a single recursive call is in \[\Oh(2^{3681\sqrt{d}\log^7(d)+99\log^2(d)}\gamma(d)n^{\rho+\varepsilon}) \subseteq \Oh(2^{3682\sqrt{d}\log^7(d)}\gamma(d)n^{\rho+\varepsilon}).\]
    
    We next analyze the running time of all calls of $\SQ$ and~$\clean$ combined.
    To this end, we define three functions~$\phi$, $T$, and~$\xi$.
    The function~$\Phi(d,x) = \min(\frac{x}{40\log^{7}(d)},1)$ is a potential function measuring how close a recursive call is to calling $\clean$.
    
    The function~$T[n,D,\phi]$ upper-bounds the number of operations performed by the current recursive call and all future recursive calls of~$\SQ$ and~$\clean$ when the current recursive call solves an instance with~$\nu$~vertices,~$d= D$, and~$\Phi$ evaluating to~$\phi$.
    Let~$\kappa$ be a constant such that each recursive call of~$\SQ, \clean$ and~$\base$ performs at most~$\kappa 2^{3682\sqrt{D}\log^7(D)}\gamma(D)\nu^{\rho+\varepsilon}$ operations and let~$\lambda = 2^{2000\log^3(D)} \kappa$.
    We will show that
    \[T[\nu,D,\phi] \leq \begin{cases}
        2\kappa 2^{3682\sqrt{D}\log^7(D)}\gamma(D)\nu^{\rho+\varepsilon} & \text{ if \cref{lin:basecase} is reached}\\
        2^{2000\log^3(D)}T[\nu,D,\nicefrac{1}{\log^{6}(D)}] + \lambda 2^{3682\sqrt{D}\log^7(D)}\gamma(D)\nu^{\rho+\varepsilon} & \text{ if \cref{lin:clean} is reached}\\
        2 T[\frac{3}{4}\nu,D,\phi+\frac{1}{40\log^7(D)}] + 2^{99 \log^2(D)}T[\nu,\lfloor D-3\sqrt{D} \rfloor,1] &\\
        \quad\quad\quad\quad\quad\quad\quad\quad\quad\quad\ \ + \kappa 2^{3682\sqrt{D}\log^7(D)}\gamma(D)\nu^{\rho+\varepsilon}& \text{ else.}
    \end{cases}\]
    Finally,~$\xi(\nu,D,\phi) = \lambda 2^{4000\log^7(D)(\sqrt{D}+\frac{21\phi}{40\log^4(D)}))} \gamma(D) \nu^{\delta}$, where~$\delta = \max(\rho+\varepsilon,2.49)$, will be an upper bound for~$T$ in closed form.

    Let us start by showing that~$T$ is indeed an upper bound on the number of operations performed by all future recursive calls (including the current one) combined.
    The number of operations performed by the current recursive call was analyzed above.
    If \cref{lin:basecase} is reached, then only one further call to~$\base$ is made.
    If \cref{lin:clean} is reached, then $\clean$ is called and the proof of~\cref{lem:clean} yields that this leads to at most~$2^{2000 \log^3(D)}-1$ calls of~$\clean$, these calls combined call~$\SQ$ at most~$2^{2000\log^3(D)}$ times, and in each of these calls, the value of~$x$ is at most~$44\log(D)$ by \cref{lem:invariants}.
    This implies that~$\phi$ is at most~$\frac{44}{40\log^7(D)} \leq \frac{1}{\log^6(D)}$.
    Moreover, the number of operations performed by the current recursive call of~$\SQ$ and the at most~$2^{2000\log^3(D)}$ recursive calls of~$\clean$ perform at most
    $$2^{2000\log^3(D)} \kappa 2^{3682\sqrt{D}\log^7(D)}\gamma(D)\nu^{\rho+\varepsilon} \leq \lambda 2^{3682\sqrt{D}\log^7(D)}\gamma(D)\nu^{\delta}.$$
    This shows that~$T$ is indeed an upper bound in the second case.
    If neither of \cref{lin:basecase,lin:clean} are reached, then we call $\SQ$ once each in \cref{lin:rec1,lin:rec2} with the same value of~$D$, the value of~$\nu$ shrunk by a factor of~$\nicefrac{3}{4}$, and~$x$ increased by one.
    Hence, the value of~$\phi$ increases by at most~$\frac{1}{40\log^7(D)}$.
    Moreover, we call~$\SQ$ in \cref{lin:rec3,lin:rec4} once for each cycle in~$\mathcal{S}$.
    As analyzed above, there are at most~$2^{99 \log^2(D)}$ cycles in~$\mathcal{S}$ and in each of these recursive calls, the number~$\nu$ of vertices does not increase, the value of~$D$ is decreased by~$\lceil 3\sqrt{D} \rceil$, and in the worst case the value of~$\phi$ increases to~$1$.
    The number of operation in the current recursive call is again at most~$\kappa 2^{3682\sqrt{D}\log^7(D)}\gamma(D)\nu^{\rho+\varepsilon}$ and hence~$T$ is an upper bound as claimed in all cases.

    To conclude the proof, we show that~$T \leq \xi$.
    We prove this by induction and therefore assume that~$T \leq \xi$ in each recursive call.
    Note that if \cref{lin:basecase} is reached, then only one recursive call is made to $\base$ which does not incur any other recursive calls.
    Hence, it holds trivially that~$T[\nu,D,\phi] = 2\kappa 2^{3682 \sqrt{D} \log^7(D)} \gamma(D) \nu^{\rho+\varepsilon} \leq \lambda 2^{4000\sqrt{D}\log^7(D)}\gamma(D)n^{\delta} \leq \xi(\nu,D,\phi)$ (note that~$\frac{21\Phi}{40\log^4(D)}$ might be very small but it is non-negative).
    If \cref{lin:clean} is reached, then we have that~$T[\nu,D,\phi] = 2^{2000 \log^3(D)} T[\nu,D,\nicefrac{1}{\log^{6}(D)}]$ and~$\phi = 1$.
    The induction hypothesis then yields
    \begin{align*}
        T[\nu,D,\phi] &= 2^{2000 \log^3(D)} T[\nu,D,\nicefrac{1}{\log^{6}(D)}] + \lambda 2^{3682\sqrt{D}\log^7(D)}\gamma(D)\nu^{\rho+\varepsilon}\\
        &\leq 2^{2000 \log^3(D)} \cdot \lambda 2^{4000\log^7(D)(\sqrt{D}+\frac{1}{2\log^{10}(D)})} \gamma(D)\nu^{\delta} + \lambda 2^{3682\sqrt{D}\log^7(D)}\gamma(D)\nu^{\delta}\\
        &\leq \lambda \gamma(D) \nu^{\delta}(2^{4000\log^7(D)(\sqrt{D}+\frac{21}{40\log^{10}(D)}+\frac{2000\log^3(D)}{4000\log^7(D)})}+2^{3682 \sqrt{D}\log^7(D)})\\
        &\leq \lambda (2\cdot2^{4000\log^7(D)(\sqrt{D}+\frac{21}{40\log^{10}(D)}+\frac{20}{40\log^4(D)})})\gamma(D) \nu^{\delta}\\
        &\leq \lambda 2^{4000\log^7(D)(\sqrt{D}+\frac{\frac{21}{\log^6(D)}+20+\frac{1}{100\log^3(D)}}{40\log^4(D)})}\gamma(D) \nu^{\delta}\\
        &\leq \lambda 2^{4000\log^7(D)(\sqrt{D}+\frac{21}{40\log^4(D)})} \gamma(D) \nu^{\delta} = \xi(\nu,D,1) = \xi(\nu,D,\phi).
    \end{align*}
    Recall that $\log(D) \geq 7$ and \cref{lem:help} states that~$\sqrt{D - 3\sqrt{D}}+1 \leq \sqrt{D}-\frac{1}{10}$.
    This yields for the last case
    \begin{align*}
        T[\nu,D,\phi] &\leq 2 T[\frac{3}{4}\nu,D,\phi+\frac{1}{40\log^7(D)}] + 2^{99 \log^2(D)}T[\nu,\lfloor D-3\sqrt{D} \rfloor,1] + \lambda 2^{3682\sqrt{D}\log^7(D)}\gamma(D)\nu^{\rho+\varepsilon}\\
        &\leq 2 \cdot \lambda 2^{4000\log^7(D)(\sqrt{D}+\frac{21(\phi+\frac{1}{40\log^7(D)})}{40 \log^4(D)})} \gamma(D)(\frac{3}{4}\nu)^{\delta} \\&\quad + 2^{99 \log^2(D)} \cdot \lambda 2^{4000\log^7(D-3\sqrt{D})(\sqrt{D-3\sqrt{D}}+\frac{21}{40\log^4(D)})} \gamma(D-3\sqrt{D})\nu^{\delta} \\
        &\quad + \kappa 2^{3682\sqrt{D}\log^7(D)}\gamma(D)\nu^{\delta}\\
        &\leq 2 \cdot (\frac{3}{4})^\delta \cdot \lambda 2^{4000\log^7(D)(\sqrt{D}+\frac{21(\phi+\frac{1}{40\log^7(D)})}{40 \log^4(D)})} \gamma(D)\nu^{\delta}
        \\&\quad + 2^{99 \log^2(D)} \cdot \lambda 2^{4000\log^7(D)(\sqrt{D-3\sqrt{D}}+1)} \gamma(D)\nu^{\delta} \\
        &\quad + \kappa 2^{3682\sqrt{D}\log^7(D)}\gamma(D)\nu^{\delta}\\
        &\leq 2 \cdot (\frac{3}{4})^\delta \cdot \lambda 2^{4000\log^7(D)(\sqrt{D}+\frac{21\phi}{40\log^4(D)}+\frac{21}{1600\log^{11}(D)})} \gamma(D)\nu^{\delta}
        \\&\quad + 2^{99 \log^2(D)} \cdot \lambda 2^{4000\log^7(D)(\sqrt{D}-\frac{1}{10})} \gamma(D)\nu^{\delta} \\
        &\quad + \kappa 2^{3682\sqrt{D}\log^7(D)}\gamma(D)\nu^{\delta}\\
        &\leq \lambda 2^{4000\log^7(D) (\sqrt{D}+\frac{21\phi}{40\log^4(D)})} \gamma(D) \nu^{\delta} \cdot \\
        &\quad \quad \quad \quad \quad \quad \Big (2 \cdot (\nicefrac{3}{4})^{\delta} \cdot 2^{\frac{21 \cdot 4000}{1600 \log^4(D)}} + 2^{99 \log^2(D) + 4000 \log^7(D) (-\frac{1}{10})} + \frac{\kappa}{\lambda} \Big)\\
        &\leq \lambda 2^{4000 \log^7(D)(\sqrt{D}+\frac{21\phi}{40\log^4(D)})} \gamma(D) \nu^{\delta}  \cdot \Big(2 \cdot (\nicefrac{3}{4})^{2.49} \cdot 2^{\frac{210}{4\cdot 7^4}} + 2^{-300\log^7(D)} + 2^{-2000\log^3(D)}\Big)\\
        &\leq \lambda 2^{4000 \log^7(D)(\sqrt{D}+\frac{21\phi}{40\log^4(D)})} \gamma(D) \nu^{\delta} = \xi(\nu,D,\phi).
    \end{align*}
    This concludes the proof.
\end{proof}

Note that \cref{thm:main-fpt} only solves the decision version whether a set of at most~$k$ vertices with a certain property exists.
However, the above algorithm can also easily be modified to output a set of vertices of an optimal solution by simply storing one representative of each partial solution.

\section{Applications}
\label{sec:applications}
In this section, we apply \cref{thm:main-fpt} to different problems to achieve subexponential-in-$k$ and subcubic-in-$n$ time algorithms for planar input graphs.
We start with \fragment, which acts both as an easy example of to how to use the framework and a convenient way to avoid using the framework.
In particular, we will list the different steps required to apply the framework and then show how this is done in the case of \fragment.
We will then show that problems like \wis, \wpvc, \den, and \knkcut{} are all special cases of \fragment.
Hence, in order to show subexponential parameterized algorithms for these (as well as weighted and directed variant of the last two), one does not need to apply the framework but can just show how to use an algorithm for \fragment{} instead.
Afterwards, we apply the framework to \wif, \minor, and \im.
These examples will get progressively more technical and we recommend to read them in order.
Our aim is to provide the necessary ingredients one after another so that each step on its own is as simple as possible to follow.
Most of the aforementioned problems also have relevant special cases which we mentioned in the introduction.
We will discuss those along the way.

\subsection{A Simple-to-use Fragment}
\label{sec:fragment}
Here, we give a general recipe for how to apply \cref{thm:main-fpt}.
We then show how to use this recipe for showing that \fragment{} can be solved in~$2^{\tilde{\Oh}(\sqrt{k})}n^{2.49}$~time.
Finally, we show how to use this fact to show that problems like \wis, \wpvc, \den, and \knkcut{} can be solved in the same time.
We note that subexponential parameterized algorithms were already known for \wis~\cite{masterthesis} and \knkcut{} for undirected planar graphs~\cite{koana}, but was asked as a specific open question for \den~\cite{koana} and was not known for \wpvc{} or \knkcut{} for directed planar graphs.
This result also generalizes in a straight-forward way to weighted variants of \knkcut{} and \den{} and we suspect that there are many more applications of \fragment.

We now present the recipe for any problem~$\pi$ under consideration.
\begin{enumerate}
    \item Decide for an appropriate function~$\gamma$ and constants~$\rho,\alpha_{\pi} \geq 1$.
    \item Define~$S_{\pi}$ and~$R_{\pi}$ and show that they can be computed/maintained in~$\Oh(\gamma(k)n^{\rho})$ time (Requirements R2 and R3 in \cref{thm:main-fpt}).
    \item Define an appropriate request version~$\Q$.
    \item Show that Requirement R1 of \cref{thm:main-fpt} holds, that is, there is an index~$i$ such that we can remove each layer~$i+j\cdot(\alpha_{\pi}k+1)$ of a BFS layering for any~$j$.
    \item Show that Requirement R4 of \cref{thm:main-fpt} is fulfilled, that is, given a separator~$C$ and solutions to subproblems on either side of the separator, we can compute in~$\Oh(\gamma(k)n^\rho)$ time all partial solutions that do not intersect~$C$ too much.
    \item Design an algorithm that given access to an oracle for \Q{} solves~$\pi$ and analyze its running time and how many oracle calls it makes.
\end{enumerate}
One also has to show that \Q{} is computable but since this is obvious in all examples studied in this paper, we will ignore this step.
We next show how to use this recipe and \cref{thm:main-fpt} to show a subexponential parameterized algorithm for \fragment.
Before doing so, we prove a small lemma that will be useful for all applications of \cref{thm:main-fpt}.

\begin{lemma}
    \label{lem:numberofpebbles}
    Let~$(C,1)$ be a pair that is measured with respect to two integers~$d$ and~$\alpha$, where~$\alpha$ is a constant.
    Then, the number of possible pebble sets~$\peb$ for~$(C,1)$ is in~$\Oh(2^{80\alpha \sqrt{d}\log^2(d)})$ and all pebble sets can be computed in~$\Oh(2^{80\alpha \sqrt{d}\log^2(d)}d)$ time.
\end{lemma}

\begin{proof}
    Note that by definition,~$|C_{\he}| \leq 40\alpha\sqrt{d}$, $|C_{\li}| \leq 60\alpha d \log(d)$, and the number of possible pebble sets for~$(C,1)$ is at most~$2^{|C_{\he}|} \binom{|C_{\li}|}{20\alpha\sqrt{d}\log(d)}$.
    If~$d \leq (60\alpha)^2$, that is, $d$ is constant, we get a constant number of possible pebble sets.
    Otherwise, it holds that~$60 \alpha \log(d) < d$ because~${\log(d) < \sqrt{d}}$ for any~$d \geq 16$.
    Thus, the number of pebble sets is upper bounded by
    \[2^{40 \alpha \sqrt{d}} (60 \alpha d \log(d))^{20 \alpha \sqrt{d} \log(d)} \leq 2^{40 \alpha \sqrt{d} + 20 \alpha \sqrt{d} \log(d) \log(d^2)} \leq 2^{80\alpha\sqrt{d}\log^2(d)}.\]

    We can compute all such sets by iterating over all subsets~$S_1$ of~$C_{\he}$ and all subsets~$S_2$ of~$C_{\li}$ of size at most~$60\alpha d \log(d)$ in~$\Oh(2^{|C_{\he}|} \binom{|C_{\li}|}{20\alpha\sqrt{d}\log(d)}) \subseteq \Oh(2^{80\alpha\sqrt{d}\log^2(d)})$ time and computing~$S_1 \cup S_2$ in~$\Oh(d)$ time.
\end{proof}

We are finally in a position to prove our result for \fragment.
To this end, recall that a request variant \Q{} of locality~$\alpha_{\pi}$ takes a tuple~${(d,G,B,x,S_{\pi})}$ as input where~${B=(B_{\he},B_{\li},B_{\di})}$ is partitioned into heavy, light, and discarded vertices, $G$ is a~$\alpha_{\pi}d$-outerplanar graph, and $(B,x)$ is almost measured with respect to~$d$ and~$\alpha_\pi$.
The output is a lookup table with an entry for each request~$r \in R_\pi(d,G,B,\peb,S_{\pi})$, where~$\peb$ is a pebble set for~$(B,d,\alpha_{\pi},x)$ and the output represents some solution where the vertices in the solution contained in~$B_{\he} \cup B_{\li} \cup B_{\di}$ is exactly~$\peb$.

\begin{theorem}
    \fragment{} can be solved in~$2^{\tilde{\Oh}(\sqrt{k})}n^{2.49}$~time for planar input graphs.
\end{theorem}

\begin{proof}
    We follow the recipe from above and first choose~$\gamma(d) = 2^{240\sqrt{d}\log^2(d)}d^2$, $\rho=2$, and~$\alpha_{\pi} = 3$.
    The set~$S_{\pi}$ will contain the weight function~$w$ and the cost function~$c$ for all edges and vertices in the currently considered graph as well as the values~$\alpha_1,\alpha_2,\alpha_3$, and~$\beta$.
    A request~$r \in R_{\pi}(d,G,B,\peb,S_{\pi})$ will be a pair~$(\peb,\ell)$ where~$\ell \leq d$.
    The answer for such a request should correspond to the maximum score achieved by a set of vertices intersecting~$B$ in exactly~$\peb$ and of cost exactly~$\ell$.
    It is easy to verify that~$S_{\pi}$ can be maintained in~$\Oh(n^2)$ time and~$R_{\pi}(d,G,B,\peb,S_{\pi})$ can be computed in~$\Oh(nk)$ time by iterating over all~$\ell \leq d$ and copying the vertices in~$\peb$.
    For the third step of our recipe, we define \fragmentrequest{} as follows.
    
    \problemdef{\fragmentrequest}
    {An integer~$d$, a directed~$3d$-outerplanar graph~$G=(V,A)$, a portal set~${B=(B_{\he},B_{\li},B_{\di})}$, an integer~$x$, a vertex-cost function~$c \colon V \rightarrow \mathds{N}$, an arc-weight function~$w: A \rightarrow \mathds{R}$, and values~$\alpha_1,\alpha_2,\alpha_3,\beta \in \mathds{R}$ such that~$(B,x)$ is almost measured with respect to~$d$ and~$3$.}
    {Output a lookup table that returns for each request~$r = (\peb,\ell)$ for each pebble set~$\peb$ for~$(B,d,3,x)$ and each~$\ell \leq d$, the maximum value
    \[\alpha_1 \big(\sum_{\substack{(u,v) \in A\\u,v \in S}} w((u,v))\big) + \alpha_2 \big(\sum_{\substack{(u,v) \in A\\u \in S, v \notin S}} w((u,v))\big) + \alpha_3 \big(\sum_{\substack{(u,v) \in A\\u \notin S,v \in S}} w((u,v))\big) + \beta \big(\sum_{v \in S} w(S)\big)\] of any set~$S$ of cost exactly~$\ell$ and where~$S \cap (B_{\he} \cup B_{\li} \cup B_{\di}) = \peb$ (or~$-\infty$ if no such set exists).}

    Next, we show that Requirement R1 of \cref{thm:main-fpt} holds, that is, there is an index~$i$ such that we can remove the vertices from each layer~$i+j\cdot(3k+1)$ for any~$j$ of any given BFS layering of the underlying undirected graph~$G_{\downarrow}$.
    To this end, note that the score of a set~$S$ only depends on the vertices in~$S$ and the arcs incident to~$S$.
    For each vertex~$v \in S$ and any BFS-layering of~$G_{\downarrow}$, it holds that all arcs incident to~$v$ have all their endpoints in at most~$3$ (consecutive) layers: the layer of~$u$, the one before, and the one after.
    Thus, all arcs relevant for a set~$S$ of size at most~$k$ are incident to at most~$3k$ layers and if there are more than~$3k$ layers, then there is an index~$i$ such that removing all layers~$i + j(3k+1)$ for any~$j$ does not invalidate~$S$.
    Moreover, no new solutions are created in the process, so Requirement~R1 holds for~$\alpha_{\pi} = 3$ for \fragmentrequest{} (and \fragment).

    The penultimate step is to show that Requirement R4 holds, that is, there is an algorithm that given a separator~$C=(C_{\he},C_{\li},C_{\di})$ such that~$(C,1)$ is measured for~$d$ and~$3$, a pebble set~$\peb$ for~$(B,d,3,x)$, and the solutions for instances
        \begin{align*}
            (d,C_{\inte},((B_{\he} \cup C_{\he})[C_{\inte}],(B_{\li} &\cup C_{\li})[C_{\inte}],(B_{\di} \cup C_{\di})[C_{\inte}]),x+1,S_{\pi}[C_{\inte}]) \text{ and}\\
            (d,C_{\ex},((B_{\he} \cup C_{\he})[C_{\ex}],(B_{\li} &\cup C_{\li})[C_{\ex}],(B_{\di} \cup C_{\di})[C_{\ex}]),x+1,S_{\pi}[C_{\ex}]),
        \end{align*}
        can compute the solutions for all requests~${r \in R_\pi(d,G,B,\peb,S_{\pi})}$ that do not intersect~$C_{\di}$ and intersect~$C_{\li}$ in at most~$60\sqrt{d}\log(d)$ vertices in~$\Oh(\gamma(d)n^2)$ time.
    So assume that we are given~$C$, solutions~$\res_1$ for the interior instance and~$\res_2$ for the exterior instance, and a pebble set~$\peb$ for~$(B,d,3,x)$.
    Let~$\peb_{\inte}$ and~$\peb_{\ex}$ be the subsets of~$\peb$ which appear in~$C_{\inte}$ and~$C_{\ex}$, respectively.
    We iterate over all possible pebble sets~$\peb'$ for~$(C,d,3,1)$, that is, sets which do not contain any vertices in~$C_{\di}$ and at most~$60 \sqrt{d} \log(d)$ vertices in~$C_{\li}$.
    Let~$A_C$ be the set of arcs with both endpoints in~$C_{\he} \cup C_{\li} \cup C_{\di}$.
    Let~$A_1, A_2, A_3 \subseteq A_C$ be the sets of arcs~$(u,v) \in A_C$ such that (1) both $u$ and~$v$ belong to~$\peb'$, (2) $u \in \peb'$ and~$v \notin \peb'$, and (3) $u \notin \peb'$ and~$v \in \peb'$.
    Then, for each such set~$\peb'$, we iterate over all~$\ell' \leq \ell$ and lookup the result for the interior for request~$(\peb_{\inte} \cup \peb',\ell')$ and for the exterior for~$(\peb_{\ex} \cup \peb',\ell-\ell'+\sum_{v \in \peb'}c(v))$.
    Let~$s_1$ and~$s_2$ be these scores.
    We then compute~$s = s_1 + s_2 - (\alpha_1 (\sum_{a \in A_1}w(a)) + \alpha_2 (\sum_{a \in A_2}w(a)) + \alpha_3 (\sum_{a \in A_3}w(a)) + \beta (\sum_{v \in \peb'}w(v)))$ and return the highest computed~$s$ as a solution for~$r$.

    We next prove that our computation is correct, that is, (i) we show for each request and any solution for that request that we store a value that is at least as large as the score of the solution and (ii) we show that each time we update a request, there is a partial solution of that score.
    Before doing so, note that~${((B_{\he} \cup C_{\he},B_{\li} \cup C_{\li},B_{\di} \cup C_{\di}),x+1)}$ is measured for~$d$ and~$3$ as~$(B,x)$ and~$(C,1)$ are both measured for~$d$ and~$3$.
    Thus, the solutions~$\res_1$ and~$\res_2$ have entries for~$(\peb_{\inte} \cup \peb',\ell')$ and~$(\peb_{\ex} \cup \peb',\ell')$ for each~$\peb'$ with~${\peb' \cap C_{\di} = \emptyset}$ and~${|\peb' \cap C_{\li}| \leq 60\sqrt{d}\log(d)}$ and for each~$\ell' \leq d$, respectively.

    Now, consider a set~$Z$ of vertices corresponding to a solution for~$r$ with no vertices in~$C_{\di}$ and at most~$60\sqrt{d}\log(d)$ vertices in~$C_{\li}$.
    Let~$Z_{\inte}$ and~$Z_{\ex}$ be the subsets of~$Z$ in~$C_{\inte}$ and~$C_{\ex}$, respectively, and let~$s_{\inte}$ and~$s_{\ex}$ be the scores achieved by~$Z_{\inte}$ in~$C_{\inte}$ and~$Z_{\ex}$ in~$C_{\ex}$, respectively, where~$\peb' = Z_{\inte} \cap Z_{\ex}$.
    Note that the queries~$(\peb_{\inte} \cup \peb',\sum_{v\in Z_{\inte}}c(v))$ and~$(\peb_{\ex} \cup \peb',\sum_{v\in Z_{\ex}}c(v))$ return values at least~$s_{\inte}$ and~$s_{\ex}$, respectively.
    Moreover, all arcs in~$A_1 \cup A_2 \cup A_3$ as well as~$w(v)$ for all~$v \in \peb'$ are counted in both~$s_{\inte}$ and~$s_{\ex}$ and no other arc is counted in both solutions as~$C$ is a separator.
    Thus, the score of~$Z$ is~${s_{\inte} + s_{\ex} - (\alpha_1 (\sum_{a \in A_1}w(a)) + \alpha_2 (\sum_{a \in A_2}w(a)) + \alpha_3 (\sum_{a \in A_3}w(a)) + \beta (\sum_{v \in \peb'}w(v)))}$, which is precisely what we computed.
   
    In the other direction, note that any value we compute corresponds to a set of vertices of cost precisely~$\ell$ which does not intersect~$C_{\di}$ and which intersects~$C_{\li}$ in at most~$60\sqrt{d}\log(d)$ vertices.
    This is true as both~$\res_1$ and~$\res_2$ guarantee this property for the two solutions we combine and both solutions contain precisely the same vertices in~$C$.
    The value we subtract from the sum of the two scores is precisely the amount we counted twice (once in each of the two scores~$s_{1}$ and~$s_2$).
    Thus, the computation is correct.

    We next analyze the running time of our algorithm for Requirement R4.
    The number of possible pebble sets~$\peb'$ for~$(C,1)$ is by \cref{lem:numberofpebbles} in~$\Oh(2^{240\sqrt{d}\log^2(d)})$ and they can be enumerated in~$\Oh(2^{240\sqrt{d}\log^2(d)}d)$ time.
    We then iterate over all possible values of~$\ell \leq d$ and all other computations can be performed in~$\Oh(n^2)$ time.
    The overall running time is therefore bounded by~$\Oh(2^{240 \sqrt{d}\log^2(d)}d^2 n^2) = \Oh(\gamma(d)n^2)$ and Requirement R4 thus holds for \fragmentrequest.
   
    Finally, we present a reduction from \fragment{} to \fragmentrequest{} and analyze the running time.
    To this end, we apply \cref{lem:baker} with~$d=3k+1$ to compute~$3k+1$ $3k$-outerplanar graphs.
    We then solve \fragmentrequest{} for each of these graphs using \cref{thm:main-fpt} with~$d=k$, portal set~$B=(\emptyset,\emptyset,\emptyset)$, $x=1$, the input weight and cost functions, and input values~$\alpha_1,\alpha_2,\alpha_3$ and~$\beta$.
    For each such iteration, we query the optimal solution for request~$(\emptyset,k)$, which corresponds to an optimal solution consisting of vertices of cost exactly~$k$.
    If any iteration finds a solution with score at least~$W$, then we output yes and otherwise we output no.
    Since the correctness can easily be verified, it remains to analyze the running time.
    Using \cref{thm:main-fpt} and the arguments above, solving each instance of \fragmentrequest{} takes~$2^{\tilde{\Oh}(\sqrt{3k})}\gamma(k)n^{2.49}$ time.
    The time for computing all inputs is in~$\Oh(kn)$ yielding an overall running time in~$\Oh((3k+1) 2^{\tilde{\Oh}(\sqrt{k})} 2^{240\sqrt{k}\log^2(k)}k^2n^{2.49}) + \Oh(nk) = 2^{\tilde{\Oh}(\sqrt{k})}n^{2.49}$, concluding the proof.
\end{proof}

We next show that \wis, \wpvc, \den, and \knkcut{} are all special cases of \fragment{} and hence we can show subexponential parameterized algorithms for all of them without needing to apply the framework for each of them.
Subexponential parameterized algorithms for planar graphs were already known for \wis~\cite{masterthesis}, \knkcut{} in undirected planar graphs~\cite{koana}, and (unweighted) \textsc{Partial Vertex Cover}~\cite{koana} (even for apex-minor free graphs).
They were not known for the directed variant of \knkcut{} and \wpvc.
Moreover, the existence of such an algorithm was asked as an open problem for (unweighted) \den~\cite{koana}, which we resolve here in the affirmative.
We mention is passing that an optimal solution for \den{} might contain many connected components and hence the treewidth-pattern-covering approach does not work.

\begin{corollary}
    \wis, \wpvc, \den, and \knkcut{} can all be solved in~$2^{\tilde{\Oh}(\sqrt{k})}n^{2.49}$ time for planar input graphs.
\end{corollary}

\begin{proof}
    For the directed variant of \knkcut, we simply set~$c(v) = 1$ for all vertices~$v$, $\alpha_1 = \alpha_3 = \beta = 0$, and~$\alpha_2 = 1$.
    Then, a set~$S$ of cost exactly~$k$ is a set of size exactly~$k$ and the score of~$S$ is precisely the sum of weights of edges from~$S$ to the rest of the graph.
    
    For the undirected variant of \knkcut{} and for all remaining problems, we first arbitrarily direct all edges of the input graph (without changing the weight in case the problem has edge weights).
    For \knkcut, \den, and \wis, we again set~$c(v) = 1$ for all vertices.
    For \knkcut, we set~${\alpha_1 = \beta = 0}$ and~${\alpha_2 = \alpha_3 = 1}$.
    For \den, we set~$\alpha_1 = 1$ and~$\alpha_2 = \alpha_ 3 = \beta = 0$.
    It is easy to verify that \fragment{} for the chosen values of~$\alpha_1,\alpha_2,\alpha_3$ and~$\beta$ corresponds to \knkcut{} and \den, respectively.

    For \wpvc, we set the cost of a vertex to its input weight and set~${\alpha_1 = \alpha_2 = \alpha_3 = 1}$ and~${\beta = 0}$.
    We then solve for each~$\ell \leq k$ an instance of \fragment{} to find whether there is a set of cost exactly~$\ell$ that covers edges of weight at least~$W$.
    The additional factor of~$k$ in the running time is subsumed in~$2^{\tilde{\Oh}(\sqrt{k})}$.

    Finally, we consider \wis.
    Here, we set~$\alpha_2 = \alpha_3 = 0$, $\beta = 1$, and~$\alpha_1 = -n w$, where~$w$ is the largest weight of any vertex.
    This ensures that any set which induces at least one edge, that is, it is not an independent set, has a negative score.
    Again, we iterate over all~$\ell \leq k$ and compute the maximum weight of any independent set of size exactly~$\ell$.
    The maximum over all of those gives the maximum-weight independent set of size at most~$k$ and the additional factor of~$k$ in the running time is again subsumed in~$2^{\tilde{\Oh}(\sqrt{k})}$.
    This concludes the proof.
\end{proof}

\subsection{Maximum-Weight Induced Forest}

We next study \wif.
Before giving the formal proof, let us first give a high-level overview.
When combining solutions along a separator, we want to verify that the two partial solution (induced forests) are together also a forest.
To this end, consider a cycle in the union of two partial solutions and let~$\peb'$ be the set of vertices from the separator that are part of the cycle.
Each consecutive pair is connected in (at least) one of the two partial solutions.
Thus, the main idea will be to keep track of (i) which vertices of the portal set are contained in the partial solution (a pebble set) and (ii) which of these vertices are in the same connected component in the partial solution.
We encode this information in a \emph{\textbf{partition}} of a pebble set.
To this end, we denote the set of all possible partitions of a set~$S$ by~$\mathbf{\mathcal{B}(S)}$ and each set in~$S$ a \emph{\textbf{block}}.
For a given subset~$A$ of vertices and a given partition~$P \in \mathcal{B}(A)$, we say that a graph~$G'$ has connected components for~$A$ as described by~$P$ if for each pair~$u,v \in A$ of vertices in~$A$, it holds that~$u$ and~$v$ belong to the same connected component in~$G'$ if and only if~$u$ and~$v$ are contained in a common block~$p \in P$.
We also say that~$P$ describes the connected components of~$G'$ for~$A$.
We will then be interested in keeping track of the connectivity of the union of two partial solutions.
We will show that this is equivalent to the \emph{\textbf{join}} of two partitions~$P$ and~$Q$ in case~$P$ and~$Q$ partition the same set of vertices.
For notational convenience, we will also define the join for partitions of different sets of vertices and show that this generalized notion of join precisely describes the connectivity of the union of two partial solutions we desire.

We now give a formal definition of joins of partitions.
We denote the join of~$P$ and~$Q$ by~$P \sqcup Q$.
If~$P$ and~$Q$ partition the same set of vertices, then the join of~$P$ and~$Q$ is the finest common coarsening of~$P$ and~$Q$ and is obtained by starting with~$P$ and iteratively merging two sets~$p_1,p_2 \in P$ if there is a set~$q \in Q$ with~$p_1 \cap q \neq \emptyset$ and~$p_2 \cap q \neq \emptyset$.
If~$P$ partitions a set~$V_1$ and~$Q$ partitions a different set~$V_2$, then we define the join of~$P$ and~$Q$ as the join of~$P'$ and~$Q'$, where~${P' = P \cup \{\{x\} \mid x \in V_2 \setminus V_1\}}$ and~${Q' = Q \cup \{\{x\} \mid x \in V_1 \setminus V_2\}}$, that is, we first add all elements which are not partitioned by~$P$ or~$Q$ as singletons.
The join of two partitions of sets of size~$n$ can be computed in~$\Oh(n\alpha(n))$ time, where~$\alpha(n)$ is the inverse Ackermann's function~\cite{TL84}.
We next show that the join operation accurately describes the connectivity of the union of two partial solutions when these are joined along a common separator.
This result is folklore but since we could not find a formal proof anywhere, we provide a proof for the sake of completeness.

\begin{lemma}[folklore]
    \label{lem:connectivity}
    Let~$G=(V,E)$ be an undirected graph and~$(X,Y)$ be a separation of~$G$.
    Let~$Z = X \cap Y$, $\peb \subseteq Z$, and let~$B \subseteq V$ be a set of vertices with~$B \cap Z = \peb$.
    Let~${B_X = B \cap X}$ and~${B_Y = B \cap Y}$.
    Let~$P_X \in \mathcal{B}(B_X)$ and~$P_Y \in \mathcal{B}(B_Y)$ be partitions of~$B_x$ and~$B_Y$, respectively.
    Let~$G_X=(V_X,E_X)$ and~$G_Y=(V_Y,E_Y)$ be subgraphs of~$G[X]$ and~$G[Y]$, respectively, such that~$B_X \subseteq V_X$, $V_X \cap Z = \peb$, $G_X$ has connected components for~$B_X$ as described by~$P_X$, $B_Y \subseteq V_Y$, $V_Y \cap Z = \peb$, and $G_Y$ has connected components for~$B_Y$ as described by~$P_Y$.
    Then, the graph~$G'=(V_X \cup V_Y, E_X \cup E_Y)$ has connected components for~$B$ as described by~$P_X \sqcup P_Y$.
\end{lemma}

\begin{proof}
    First, we show that if two vertices~$u,v \in B$ are in the same block of~$P_X \sqcup P_Y$, then they are in the same connected component in~$G'$.
    By definition, there is a sequence~$u=w_1, w_2, \dots, w_{\ell}=v$ in~$B$ such that for each~$i \in [\ell-1]$, $w_i$ and~$w_{i+1}$ are contained in a common block of~$P_X$ or in a common block of~$P_Y$.
    If $w_i$ and~$w_{i+1}$ are in a common block of~$P_X$, then note that~$w_i,w_{i+1} \in B_X$ and there exists a path in~$G_X$ connecting them, since $G_X$ has connected components for~$B_X$ as described by~$P_X$.  
    Similarly, if they are in a common block of~$P_Y$, then there exists a path in~$G_Y$ connecting them.  
    Since~$G_X$ and~$G_Y$ are both subgraphs of~$G'$ and connectivity is transitive, this shows that~$u$ and~$v$ are in the same connected component of~$G'$.
    
    Conversely, we next show that if~$u$ and~$v$ are in the same connected component of~$G'$, then they are contained in a common block of~$P_X \sqcup P_Y$.
    So assume that~$u$ and~$v$ are in the same connected component in~$G'$.
    Then there exists a path~$P = (u=p_1, p_2, \dots, p_q=v)$ in $G'$.
    Let~$w_1,w_2,\ldots,w_\ell$ be the vertices in~$Z$ that are contained in~$P$ (in the same order) and let~$w_0 = u$ and~$w_{\ell+1}=v$.
    Each edge of $P$ is contained in~$E_X$ or in~$E_Y$ as~$G'$ is the union of~$G_X$ and~$G_Y$ and since~$(X,Y)$ is a separation of~$G$, all edges between~$w_{i-1}$ and~$w_{i}$ are contained in the same set for each~$i \in [\ell+1]$.
    Since~$P_X$ and~$P_Y$ describe the connected components of~$G_X$ and~$G_Y$ for~$B_X$ and~$B_Y$, respectively, it holds for each~$i \in [\ell+1]$ that~$w_{i-1}$ and~$w_{i}$ are contained in a common block in~$P_X$ or in~$P_Y$ (or both).
    Thus, $u$ and~$v$ are contained in a common block of $P_X \sqcup P_Y$ by definition of joins.
    This concludes the proof.
\end{proof}

The next lemma will help us decide when the union of two forests~$G_1$ and~$G_2$ is acyclic.
Intuitively, this is the case if there is no sequence~$(w_1,w_2,\ldots,w_{2\ell})$ of vertices such that~$w_{2i-1}$ and~$w_{2i}$ are connected in~$G_1$ for each~$i \in [\ell]$ and~$w_1$ and~$w_{2\ell}$ as well as~$w_{2i}$ and~$w_{2i+1}$ for each~$i \in [\ell-1]$ are connected in~$G_2$.
We will next define a notion of acyclic pairs of partitions that captures this intuition.
Unfortunately, the intuition does not quite work when~$G_1$ and~$G_2$ are not edge-disjoint.
As a simple example, consider just two vertices~$u$ and~$v$.
If they are connected by the same path in~$G_1$ and~$G_2$, then the union of the two is just a path and therefore acyclic.
If however, they are connected by two paths that differ by at least one edge, then the union is not acyclic.
We circumvent this issue by keeping track of the set of edges in the intersection and ignoring those for one of the two graphs.
We next define our notion of acyclic pairs of partitions.
A pair~$(P,Q)$ of partitions is \emph{\textbf{acyclic}} if the following graph~$F_{P,Q}$ is acyclic.
Start with a vertex~$v_x$ for each element~$x$ and a vertex~$u_p$ for each block~$p \in P \cup Q$.
Then, add an edge between~$v_x$ and~$u_p$ if~$x \in p$.
Note that deciding whether a pair~$(P,Q)$ is acyclic can be done in~$\Oh(n)$ time, where~$n$ is the number of elements partitioned by~$P$ and~$Q$, as we can build the bipartite graph in~$\Oh(n)$ time and then check for a cycle in linear time. Since each element vertex is incident to at most~$2$ edges, the number of edges is at most~$2n$ and thus this can also be done in~$\Oh(n)$ time.
We next show that acyclic partitions accurately describes when the union of two partial solutions is acyclic.

\begin{lemma}
    \label{lem:acyclic}
    Let~$G=(V,E)$ be an undirected graph and~$(X,Y)$ be a separation of~$G$ with~$Z = X \cap Y$.
    Let~$\peb \subseteq Z$ and~$B \subseteq V$ be sets of vertices such that~$B \cap Z = \peb$.
    Let~$B_X = B \cap X$ and~$B_Y = B \cap Y$.
    Let~$P_X \in \mathcal{B}(B_X)$ and~$P_Y \in \mathcal{B}(B_Y)$ be partitions of~$B_x$ and~$B_Y$, respectively.
    Let~$E_{\peb}$ be the set of edges of~$G[\peb]$.
    Let~$G_X=(V_X,E_X)$ and~$G_Y=(V_Y,E_Y)$ be acyclic induced subgraphs of~$G[X]$ and~$G[Y]$, respectively, such that~$B_X \subseteq V_X$, $V_X \cap Z = \peb$, $G_X$ has connected components for~$B_X$ as described by~$P_X$, $B_Y \subseteq V_Y$, $V_Y \cap Z = \peb$, and $G_Y$ has connected components for~$B_Y$ as described by~$P_Y$.
    Let~$G'_Y=(V_Y,E_Y \setminus E_{\peb})$ and let~$P'_Y$ be the partition that describes the connected components of~$G'_Y$ for~$B_Y$.
    Then, the graph induced by~$V_X \cup V_Y$ is an acyclic induced subgraph of~$G$ with connected components for~$B$ as described by~$P_X \sqcup P_Y$ if and only if the pair~$(P_x,P'_Y)$ is acyclic.
\end{lemma}
    
\begin{proof}
    Let~$H = G[V_X \cup V_Y]$ be the graph induced by~$V_X \cup V_Y$.
    For the first direction, assume that~$H$ is acyclic and has connected components for~$B$ as described by~$P_X \sqcup P_Y$.
    We will show that~$(P_X, P'_Y)$ is acyclic.
    To this end, assume towards a contradiction that the graph~$F_{P_X,P'_Y}$ contains a cycle~$(v_{x_1},u_{p_1},v_{x_2},u_{p_2},\ldots,u_{p_{\ell}})$, where~$x_1,x_2,\ldots,x_\ell \in B$ and~$p_1,p_2,\ldots,p_{\ell} \in P_X \cup P'_Y$.
    Since we may assume without loss of generality that the cycle we consider is simple and each vertex~$v_{x_i}$ is only incident to two nodes~$u_{p_j}$ (one for a block in~$P_X$ and one for a block in~$P'_Y$), it holds that the nodes~$u_{p_j}$ alternate between being blocks for~$P_X$ and~$P_{Y'}$.
    By construction, each consecutive pair $(v_{x_i}, v_{x_{i+1}})$ lies within the same connected component of~$G_X$ or of~$G'_Y$, that is, there are paths in~$G_X$ or~$G'_Y$ connecting~$x_i$ and~$x_{i+1}$.
    Moreover, these paths are pairwise edge-disjoint as~$E_X \cap (E_{Y} \setminus E_{\peb}) = \emptyset$ and two paths in the same graph are edge disjoint as they belong to different connected components (we may assume that each node~$u_{p_j}$ appears at most once in the cycle).
    Hence, we can concatenate all paths to get a closed walk in~$H$ that does not repeat any edges, contradicting the fact that~$H$ is acyclic.

    In the other direction, assume that~$(P_X, P'_Y)$ is acyclic.
    We will show that~$H$ is acyclic and that it has connected components for~$B$ as described by~$P_X \sqcup P_Y$.
    Note that the latter is shown by \cref{lem:connectivity}.
    Hence, assume towards a contradiction that $H$ contains a cycle~$C$.
    Since~$G_X$ and~$G_Y$ are each acyclic, $C$ must contain edges from both~$G_X$ and~$G_Y$.
    Let~$C_X$ and~$C'_Y$ denote the subset of edges in~$E_X$ and~$E_Y \setminus E_{\peb}$ and let~$w_1,w_2,\ldots,w_{\ell}$ be the vertices at the end of maximal subpaths of~$C$ in~$C_X$ and/or~$C'_Y$.
    Note that~$w_i \in \peb$ for each~$i \in [\ell]$ as~$Z$ is a separator and~$G_X$ and~$G_Y$ only contain the vertices in~$\peb$ in~$Z$.
    By definition, $w_i$ and~$w_{i+1}$ (where~$w_{\ell+1} = w_1$) are contained in the same block~$p_i$ in~$P_X$ or~$P'_Y$, respectively.
    Note that~$(w_1,p_1,w_2,\ldots,w_\ell,p_{\ell})$ forms a cycle in~$F_{P_X,P'_Y}$, a contradiction to~$(P_X,P'_Y)$ being acyclic.
    Thus, $H$ is acyclic.
    This concludes the proof.
\end{proof}

As a last step before showing how to solve \wif{} in subexponential parameterized time, we analyze the number of possible partitions for a given pebble set.

\begin{lemma}
    \label{lem:numberofpartitions}
    Let~$(B,x)$ be a pair that is almost measured with respect to two integers~$d$ and~$\alpha$, where~$\alpha$ is a constant.
    Let~$\peb$ be a pebble set for~$(B,d,\alpha,x)$.
    Then, the number of possible partitions of~$\peb$ is in~$\Oh(2^{66000 \sqrt{d}\log^{10}(d)})$ and all partitions can be enumerated in the same time.
\end{lemma}

\begin{proof}
    Since~$(B,x)$ is almost measured with respect to~$d$ and~$\alpha$, it holds that~$x \leq 200\log^7(d)$.
    Moreover, $\peb \subseteq (B_{\he} \cup B_{\li})$, $|B_{\he}| \leq 40x\alpha\sqrt{d}$ and~$|\peb \cap B_{\li}| \leq 20x\alpha \sqrt{d} \log(d)$ because~$\peb$ is a pebble set for~$(B,d,\alpha,x)$.
    Thus, $|\peb| \leq 60 (200 \log^7(d)) \alpha \sqrt{d} \log(d) = 12000 \alpha \sqrt{d} \log^8(d)$.

    The number of partitions for a set of size~$n$ is known as the $n$\textsuperscript{th} Bell's number and is known to be at most~$n^n$.
    There are even slightly better upper bounds~\cite{Bell}, but this simple bound will be sufficient for our purpose.

    If~$d < 256 \cdot 2^{\alpha}$, then the number of partitions is constant, that is, in~$\Oh(1)$.
    So assume that~${d \geq 256 \cdot 2^{\alpha}}$.
    Then, the number of possible partitions of~$\peb$ is upper-bounded by
    \begin{align*}
        |\peb|^{|\peb|} &\leq (12000 \alpha \sqrt{d} \log^8(d))^{12000 \alpha \sqrt{d} \log^8(d)}\\
        &\leq (d^2 \sqrt{d} \ d^3)^{12000 \alpha \sqrt{d} \log^8(d)}\\
        &= 2^{12000 \alpha \sqrt{d} \log^8(d) \log(d^{5.5})}\\
        &\leq 2^{66000 \sqrt{d} \log^{10}(d)}.
    \end{align*}
    Here, the second inequality is due to the fact that~$(256\alpha)^2 > 12000\alpha$ and~$\log^8(d) \leq d^3$ for all~$d \geq 256$.
    The last inequality is due to the fact that~$\alpha \leq \log(d)$.

    Note that we can enumerate all partitions of~$\peb$ efficiently, that is, with constant delay on average by iteratively deciding for each vertex~$v \in \peb$ to which already existing block it is added or whether a new block consisting only of~$v$ is added.
    This concludes the proof.
\end{proof}

With the above three lemma at hand, we can now show how to apply \cref{thm:main-fpt} to \wif.
A request will consist of a pebble set, a partition of this pebble set, and an integer~$\ell$.
The output describes the maximum weight of a forest with at most~$\ell$ vertices containing exactly the pebble set in the portal set and having connected components for the pebbles as described by the partition.
That is, we basically keep track of how many vertices are contained in the partial solution, which vertices in the portal set make up the interface to the rest of the graph, and which of these vertices are connected in the partial solution.

\begin{theorem}
    \wif{} can be solved in~$2^{\tilde{\Oh}(\sqrt{k})}n^{2.49}$ time for planar input graphs.
\end{theorem}

\begin{proof}
    We again follow the recipe presented at the start of this section.
    To this end, we first choose~$\gamma(d) = 2^{80\sqrt{d}\log^2(d)} \big(2^{66000 \sqrt{d}\log^{10}(d)}\big)^2 d$, $\rho=2$, and~$\alpha_{\pi} = 1$.
    The set~$S_{\pi}$ will contain the weight function~$w$.
    A request~$r \in R_{\pi}(d,G,B,\peb,S_{\pi})$ will be a triple~$(\peb,P,\ell)$, where~$P \in \mathcal{\peb}$ is a partition of~$\peb$ and~$\ell \leq d$.
    The output for a request is the maximum weight of an induced forest with at most~$\ell$ vertices and which contains exactly the vertices in~$\peb$ from~$B$.
    It is easy to verify that~$S_{\pi}$ can be maintained in~$\Oh(n^2)$ time.
    We next show that~$R_{\pi}(d,G,B,\peb,S_{\pi})$ can be computed in~$\Oh(\gamma(d)n^{\rho})$ time.
    By \cref{lem:numberofpartitions}, the number of possible partitions for a given pebble set is in~$\Oh(2^{66000 \sqrt{d}\log^{10}(d)})$ and all these partitions can be enumerated in~$\Oh(2^{66000 \sqrt{d}\log^{10}(d)}d)$ time.
    The number of possible values for~$\ell \leq d$ is~$d$, yielding a running time of~$\Oh(2^{66000 \sqrt{d} \log^{10}(d)}d^2)$ for computing~$R_{\pi}(d,G,B,\peb,S_{\pi})$.

    For the third step of our recipe, we define \wifr{} as follows.

    \problemdef{\wifr}
    {An integer~$d$, a $d$-outerplanar graph~$G=(V,E)$, a portal set~${B= (B_{\he},B_{\li},B_{\di})}$, an integer~$x$, and a vertex-weight function~$w \colon V \rightarrow \mathds{R}$, such that~$(B,x)$ is almost measured with respect to~$d$ and~$1$.}
    {Output a lookup table that returns for each request~$r=(\peb,P,\ell)$, where~$P \in \mathcal{B}(\peb)$ is a partition of~$\peb$, the maximum weight of a set~$S$ with~$S \cap (B_{\he} \cup B_{\li} \cup B_{\di}) = \peb$, $|S| \leq \ell$, and~$G[S]$ is a forest that has connected components for~$\peb$ as described by~$P$ (or~$-\infty$ if no such set exists).}
    
    Next, we show that Requirement R1 of \cref{thm:main-fpt} holds, that is, there is an index~$i$ such that we can remove each layer~$i+j(k+1)$ of a BFS layering for any~$j$.
    To this end, note that whether a set~$S$ is a partial solution or not only depends on the vertices in~$S$ and edges between these vertices.
    Thus, all edges relevant for a set~$S$ of size at most~$k$ are incident to at most~$k$ layers and if there are more than~$k$ layers, then there is an index~$i$ such that removing all layers~$i + j(k+1)$ for any~$j$ does not invalidate~$S$.
    Moreover, no new solutions are created in the process, so Requirement R1 holds for~$\alpha_{\pi} = 1$ for \wifr{} (and \wif).

    The penultimate step is to show that Requirement R4 holds, that is, there is an algorithm that given a separator~$C=(C_{\he},C_{\li},C_{\di})$ such that~$(C,1)$ is measured for~$d$ and~$1$, a pebble set~$\peb$ for~$(B,d,1,x)$, and the solutions for instances
    \begin{align*}
        (d,C_{\inte},((B_{\he} \cup C_{\he})[C_{\inte}],(B_{\li} &\cup C_{\li})[C_{\inte}],(B_{\di} \cup C_{\di})[C_{\inte}]),x+1,S_{\pi}[C_{\inte}]) \text{ and}\\
        (d,C_{\ex},((B_{\he} \cup C_{\he})[C_{\ex}],(B_{\li} &\cup C_{\li})[C_{\ex}],(B_{\di} \cup C_{\di})[C_{\ex}]),x+1,S_{\pi}[C_{\ex}]),
    \end{align*}
    can compute the solutions for all requests~${r \in R_\pi(d,G,B,\peb,S_{\pi})}$ that do not intersect~$C_{\di}$ and intersect~$C_{\li}$ in at most~$20\sqrt{d}\log(d)$ vertices in~$\Oh(\gamma(d)n^2)$ time.
    So assume that we are given~$C$, solutions~$\res_1$ for the interior instance and~$\res_2$ for the exterior instance, and a pebble set~$\peb$ for~$(B,d,1,x)$.
    
    We iterate over all values~$\ell \leq d$ and possible partitions~$P \in \mathcal{B}(\peb)$ using \cref{lem:numberofpartitions}.
    We then compute the optimal solution for request~$r = (\peb,P,\ell)$ as follows.
    Let~$\peb_{\inte}$ and~$\peb_{\ex}$ be the subsets of~$\peb$ which appear in~$C_{\inte}$ and~$C_{\ex}$, respectively.
    We iterate over all possible pebble sets~$\peb'$ for~$(C,d,1,1)$, that is, sets which do not contain any vertices in~$C_{\di}$ and at most~$20 \sqrt{d} \log(d)$ vertices in~$C_{\li}$.
    We also iterate over all possible values~$|\peb'| \leq \ell' \leq \ell$.
    Let~${\peb'_{\inte} = \peb_{\inte} \cup \peb'}$ and~${\peb'_{\ex} = \peb_{\ex} \cup \peb'}$.
    We then iterate over all possible partitions~$P_1 \in \mathcal{B}(\peb'_{\inte})$ and~${P_2 \in \mathcal{B}(\peb'_{\ex})}$ of these two sets, respectively.
    Let~$E_{\peb'}$ be the set of edges in~$G[\peb']$, that is, edges in the input graph between two vertices in~$\peb'$.
    Let~$G'=(V',E')$ be any graph with~$\peb'_{\ex} \subseteq V'$, ${V' \cap (C_{\he} \cup C_{\li} \cup C_{\di}) = \peb'}$, $E_{\peb'} \subseteq E'$, and~$G'$ has connected components for~$\peb'_{\ex}$ as described by~$P_2$.
    Let~$G'' = (V', E' \setminus E_{\peb'})$ and let~$P'_2$ be the partition that describes the connected components of~$G''$ for~$\peb'_{\ex}$.
    If~$P_1 \sqcup P_2 = P$ and~$(P_1,P'_2)$ is acyclic, then let~$W_1 = \res_1((\peb'_{\inte},P_1,\ell'))$ and~${W_2 = \res_2((\peb'_{\ex}),P_2,\ell-\ell'+|\peb'|)}$. 
    We update the output for~$r$ with~${W_1 + W_2 - \sum_{v \in \peb'}w(v)}$ in case this value is larger than the current value.
    This concludes the construction of our algorithm for Requirement R4.

    We next prove that our computation is correct and afterwards analyze the running time of the above algorithm.
    In order to prove the correctness, we first show that for any solution~$Z$ for a request~$r$, we update the output for request~$r$ with the weight of~$Z$ unless a larger value was already found.
    Afterwards, we show that any value we update the output with corresponds to a solution of that weight.
    Before we show the correctness, first note that~${((B_{\he} \cup C_{\he},B_{\li} \cup C_{\li},B_{\di} \cup C_{\di}),x+1)}$ is measured for~$d$ and~$1$ as~$(B,x)$ and~$(C,1)$ are both measured for~$d$ and~$1$.
    Thus, the solutions~$\res_1$ and~$\res_2$ have entries for~$(\peb'_{\inte},P_1,\ell')$ and~$(\peb'_{\ex},P_2,\ell-\ell'+|\peb'|)$ for each~$|\peb'| \leq \ell' \leq \ell \leq d$, respectively.

    We next show that any valid solution is considered.
    To this end, consider a set~$Z$ of vertices corresponding to a solution for~$r$.
    Let~$Z_{\inte}$ and~$Z_{\ex}$ be the subsets of~$Z$ in~$C_{\inte}$ and~$C_{\ex}$, respectively, and let~$W_{\inte}$ and~$W_{\ex}$ be the total weights of vertices in~$Z_{\inte}$ and~$Z_{\ex}$, respectively.
    Let~$P_{\inte}$ and~$P_{\ex}$ be the partitions of~$\peb'_{\inte}$ and~$\peb'_{\ex}$ describing the connected components of~$G[Z_{\inte}]$ and~$G[Z_{\ex}]$ for~$\peb'_{\inte}$ and~$\peb'_{\ex}$, respectively.
    Note that the queries for requests~$(\peb'_{\inte},P_{\inte},|Z_{\inte}|)$ and~$(\peb'_{\ex},P_{\ex},|Z_{\ex}|)$ return values which are at least~$W_{\inte}$ and~$W_{\ex}$, respectively.
    Moreover, as~$Z_{\inte}$ and~$Z_{\ex}$ intersect precisely in~$\peb'$, it holds that~$|\peb'| \leq |Z_{\inte}| \leq \ell'$ and~$|Z_{\ex}| \leq \ell - |Z_{\inte}| + |\peb'|$.
    We will next show that~($P_{\inte}$,$P'_{\ex}$) (where~$P'_{\ex}$ is defined analogously to~$P'_2$ above) is acyclic.
    This will conclude the first direction of the proof of correctness as follows.
    By \cref{lem:acyclic}, $P_{\inte} \sqcup P_{\ex} = P$.
    Moreover, as~$P_1 = P_{\inte}$ and~$P_2 = P_{\ex}$ is then considered by our algorithm, we store the value~$W_{\inte} + W_{\ex} - \sum_{v \in \peb'}w(v)$.
    Note that~$Z_{\inte}$ and~$Z_{\ex}$ intersect precisely in~$\peb'$ and hence
    \[\sum_{v \in Z_1 \cup Z_2} = (\sum_{v\in Z_1}w(v)) + (\sum_{v\in Z_2}w(v)) - (\sum_{v\in \peb'}w(v)) = W_{1} + W_2 - \sum_{v\in \peb'}w(v).\]
    This is both the weight of all vertices in~$Z$ and the value we compute to update the output for request~$r$.
    Hence, each solution~$Z$ is considered.
    
    So it remains to show that~($P_{\inte}$,$P'_{\ex}$) is acyclic.
    Assume towards a contradiction that~$F_{P_{\inte},P'_{\ex}}$ contains a cycle~$(v_{x_1},u_{p_1},v_{x_2},u_{p_2},\ldots,u_{p_{\ell}})$, where~$x_1,x_2,\ldots,x_\ell \in \peb'$ and~${p_1,p_2,\ldots,p_{\ell} \in P_{\inte} \cup P'_{\ex}}$.
    Since each vertex~$v_{x_i}$ is only incident to two nodes~$u_{p_j}$ (one for a block in~$P_{\inte}$ and one for a block in~$P'_{\ex}$), it holds that the nodes~$u_{p_j}$ alternate between being blocks for~$P_X$ and~$P_{Y'}$.
    By construction, each consecutive pair~$(v_{x_i}, v_{x_{i+1}})$ lies within the same connected component of~$G[Z_1]$ or of~$G[Z_2] - E_{\peb'}$, where~$E_{\peb'}$ is the set of all edges between two vertices in~$\peb'$.
    That is, there are paths in~$G[Z_1]$ or~$G[Z_2] - E_{\peb'}$ connecting~$x_i$ and~$x_{i+1}$.
    Moreover, these paths are pairwise edge-disjoint as the two graphs are edge disjoint and two paths in the same graph are edge disjoint as they belong to different connected components.
    Hence, the concatenation of all these paths yields a closed walk in~$G[Z]$ (the union of~$G[Z_1]$ and~$G[Z_2] - E_{\peb'}$), a contradiction to~$G[Z]$ being acyclic.
    
    In the other direction, we show that whenever we update a solution for a request~$r$, this corresponds to an actual solution.
    So assume there is an~$\ell'$ and partitions~$P_1$ and~$P_2$ such that~$P_1 \sqcup P_2 = P$, ($P_1$,$P'_2$) is acyclic, and we update the result for request~$r$.
    This means that neither~$W_1$ nor~$W_2$ are~$-\infty$.
    Hence, there are vertex sets~$Z_1$ and~$Z_2$ corresponding to solutions for~$(\peb'_{\inte},P_1,\ell')$ and~$(\peb'_{\ex},P_2,\ell-\ell'+|Z_1 \cap Z_2|)$.
    By \cref{lem:acyclic}, $G[Z_1 \cup Z_2]$ is then acyclic and has connected components for~$\peb$ as described by~$P_1 \sqcup P_2 = P$.
    Moreover, the weight of the union is exactly the sum of weights of~$Z_1$ and~$Z_2$ minus the weight of the intersection, that is, the weight of~$\peb'$.
    Since this is precisely what we compute, this shows that every update we do is valid, concluding the proof of correctness.
    
    We next analyze the running time of our algorithm for Requirement R4.
    The number of possible pebble sets~$\peb'$ for~$(C,1)$ is by \cref{lem:numberofpebbles} in~$\Oh(2^{80\sqrt{d}\log^2(d)})$ and the pebble sets can be enumerated in~$\Oh(2^{80\sqrt{d}\log^2(d)}d)$~time.
    Iterating over all~$\ell \leq d$ adds an additional factor of~$d$.
    Computing~$\peb_{\inte}, \peb_{\ex}, \peb'_{\inte}, \peb'_{\ex}$ can all easily be done in~$\Oh(n^2)$ time.
    We then iterate over all possible partitions~$P_1 \in \mathcal{B}(\peb'_{\inte})$ and~${P_2 \in \mathcal{B}(\peb'_{\ex})}$.
    The number of possibilities is in~$\Oh((2^{66000 \sqrt{d}\log^{10}(d)})^2)$ by \cref{lem:numberofpartitions} and all partitions can be computed in~$\Oh(2^{66000 \sqrt{d}\log^{10}(d)}d)$~time.
    Computing~$E_{\peb'}$ can be done in~$\Oh(n^2)$ time.
    Computing~$G'$ and~$G''$ can also be done in the same time as follows.
    For~$G'$, start with a graph which has a vertex for each~$v \in \peb'_{\ex}$ and each~$p \in P_2$.
    Then, add all edges in~$E_{\peb'}$.
    For each~$p \in P_2$ and each~$v \in p$, add an edge between the two corresponding vertices if (and only if) this edge does not create a cycle.
    Using union-find data structures, each such incidence can be checked in~$O(\alpha(n))$ time, where~$\alpha$ is the inverse Ackermann's function.
    Since each vertex appears in exactly one set in~$P_2$, this results in an overall running time of~$\Oh(n\alpha(n)) \subseteq \Oh(n^2)$.
    For~$G''$, simply remove the edges in~$E_{\peb'}$.
    Computing~$P'_2$ can then also be done in the same time by computing the connected components in~$G''$.
    Finally, checking whether~$P_1 \sqcup P_2 = P$ and whether~($P_1$,$P'_2$) is acyclic also take~$\Oh(n\alpha(n)) \subseteq \Oh(n^2)$ time.
    Recall that~$\gamma(d) = 2^{80\sqrt{d}\log^2(d)} \big(2^{66000 \sqrt{d}\log^{10}(d)}\big)^2 d$ and~$\rho=2$.
    The overall running time for our algorithm for Requirement R4 is in
    \[\Oh(2^{80\sqrt{d}\log^2(d)} d + 2^{80\sqrt{d}\log^2(d)} d (n^2 + 2^{66000 \sqrt{d}\log^{10}(d)}d + (2^{66000 \sqrt{d}\log^{10}(d)})^2 n^2)) \subseteq \Oh(\gamma(d)n^{\rho}).\]
        
    The final step is to show a reduction from \wif{} to \wifr{} and analyze the running time.
    To this end, we apply \cref{lem:baker} with~$d=k+1$ to compute~$k+1$ $k$-outerplanar graphs.
    We then solve \wifr{} for each of these graphs using \cref{thm:main-fpt} with~$d=k$, portal set~${B=(\emptyset,\emptyset,\emptyset)}$, $x=1$, and the input weight function.
    For each such iteration, we query the optimal solution for request~$(\emptyset,\emptyset,k)$, which corresponds to an optimal solution consisting of at most~$k$ vertices.
    If any iteration finds a solution with weight at least~$W$, then we output yes and otherwise we output no.
    Since the correctness can easily be verified, it remains to analyze the running time.
    Using \cref{thm:main-fpt} and the arguments above, solving each instance of \wifr{} takes time
    \[2^{\tilde{\Oh}(\sqrt{k})} 2^{80\sqrt{k}\log^2(k)} \big(2^{66000 \sqrt{k}\log^{10}(k)}\big)^2 k n^{2.49} \subseteq 2^{\tilde{\Oh}(\sqrt{k})} 2^{132081 \sqrt{d}\log^{10}(k)} n^{2.49} \subseteq 2^{\tilde{\Oh}(\sqrt{k})}n^{2.49}.\]
    The time for computing all inputs is in~$\Oh(kn)$ by \cref{lem:baker} and because induced subgraphs can be computed in~$O(n+m)$ time and since~$m \in \Oh(n)$ for planar graphs.
    The overall running time is therefore in~${\Oh((k+1) 2^{\tilde{\Oh}(\sqrt{k})} n^{2.49}) + \Oh(nk) = 2^{\tilde{\Oh}(\sqrt{k})}n^{2.49}}$.
    This concludes the proof.
\end{proof}

\subsection{Simple Minors}
In this section, we study \minor.
We again start with an intuitive description of the rough idea.
To this end, we first briefly sketch how the algorithm works for \textsc{Subgraph Isomorphism}.
This proof sketch follows the algorithm by Nederlof but ignores the use of efficient inclusion-exclusion which is needed to also count the number of solutions.
Let~$H=(U,F)$ be the (simple) pattern graph we try to find and let~$k = |U|$.
Note that since planar graphs are closed under taking minors, we can assume without loss of generality that~$H$ is planar (otherwise the answer is no).
Following the proof by Bodlaender, Nederlof, and van der Zanden~\cite{BNZ16}, Nederlof~\cite{Ned20} observes that the number of non-isomorphic separations of~$H$ of size~$\sqrt{k}$ is in~$2^{\Oh(\nicefrac{k}{\log(k)})}$.
A request (or subproblem in the language of Nederlof) is then a triple consisting of a pebble set~$\peb$, a separation~$(X,Y)$, and a bijection~$c$ mapping the vertices in~$\peb$ to the nodes in~$X \cap Y$.
The output indicates whether there exists a subgraph~$G'=(V',E')$ and an isomorphism~$f$ that maps the vertices in~$G'$ to the nodes in~$X$ such that~$f(v) = c(v)$ for each~$v \in \peb$.
Two partial solutions for requests~$(\peb_1,(X_1,Y_1),c_1)$ and~$(\peb_2,(X_2,Y_2),c_2)$ and a separator~$C$ can then be combined to a solution for a request~$(\peb,(X,Y),c)$ if~$\peb_1 \cup \peb_2 = \peb$, $\peb_1 \cap C = \peb_2 \cap C$, $c_1(v) = c(v)$ for each~$v \in \peb_1$, $c_2(v) = c(v)$ for each~$v \in \peb_2$, and~$H[X_1 \cup X_2] = H[X]$. 

Let us now describe how to generalize the above approach for \minor.
We also keep track of non-isomorphic separations~$(X,Y)$ of the pattern graph~$H$, but instead of a yes/no-answer, we store the minimum weight of a subgraph that is a model of~$X$.
Additionally, there are three facts that complicate the algorithm slightly when considering (rooted) minors rather than subgraphs.
First and of least importance, we can no longer assume that~$c$ is a bijection and we need to make sure that~$c$ maps each root~$v$ to the correct node in~$U$.
To this end, we consider \emph{coloring functions}, which we introduce after the high-level overview.
Second, we need to ensure that the model for a single node in~$H$ is connected.
We handle this issue similar to how we handled connectivity-related issues in the previous subsection, that is, by including a partition of the pebble set in a request.
Lastly, some edges in~$H[X \cap Y]$ might be part of the current partial solution while others might be implemented later.
To keep track of this information, we store for a given separation~$(X,Y)$ a subset of the edges in~$H[X \cap Y]$ as part of each request and allow a solution to not model these.
An additional technicality is that pebble sets will have size~$\tilde{\Oh}(\sqrt{|U|})$, so the lemma by Nederlof does not apply and we will need a slightly stronger version of his statement that also works for slightly larger separations.
For this reason and because Nederlof only states that the result follows from a proof by Bodlaender, Nederlof, and van der Zanden~\cite{BNZ16} (which is slightly non-trivial), we will include a proof for the sake of completeness.

\begin{restatable}{corollary}{sigmasimple}
    \label{lem:sigmasimple}
    Given a (simple) planar graph~$H=(U,F)$ and a non-negative integer~$t \in \Oh(|U|^{\frac{2}{3}})$, the number~$\sigma_t(H)$ of non-isomorphic separations of size at most~$t$ of~$H$ is in~$2^{\Oh(\nicefrac{|U|}{\log(|U|)})}$.
    The separations can be enumerated in~$\Oh(|U|^3\sigma_t(H))$~time.
\end{restatable}

Unfortunately \cref{lem:sigmasimple} does not hold for multigraphs.
We show in the next subsection that there are simple examples where even the number of separations of size zero are exponential in~$|U|$.
Instead, we will introduce a parameter~$\psi$ that bounds the number of non-isomorphic separations of planar multigraphs and show that it leads to subexponential parameterized algorithms.
Since \cref{lem:sigmasimple} is a simple corollary of the statement for multigraphs, we defer its proof to the next subsection.

We now define \emph{\textbf{coloring functions}}.
Given two sets~$U$ and~$V$, a coloring function~$c \colon U \rightarrow 2^V$ maps each element in~$U$ to a non-empty subset of~$V$ such that each~$v \in V$ appears in exactly one set~$c(u)$.
Equivalently, each~$v \in V$ is assigned one ``color''~$u \in U$.
As suggested by the variable names, we will use coloring functions from nodes in~$H$ to subsets of vertices in~$G$.
However, we will usually restrict the vertex sets in both~$G$ and~$H$ to subsets~$\peb \subseteq V$ and~$Z \subseteq U$.
We then consider coloring functions from~$Z$ to subsets of~$\peb$.
These should be thought of as extensions of the root function~$\mu \colon U \rightarrow 2^V$, that is, we assign each vertex~$v \in \peb$ to a node~$u \in Z$ and assume that~$v$ is part of the model for node~$u$.
Whenever~$\mu$ requires that~$v$ is part of the model for some node~$u'$, then we also want to enforce~$v \in c(u')$.
Formally, we say that~$c$ is an \emph{\textbf{extension}} of~$\mu$ on the vertex set~$\peb$ if for all~$v \in \peb$ and all~$u \in U$ it holds that if~$v \in \mu(u)$, then~$v \in c(u)$.
In particular if~$v \in \mu(u'')$ for some~$u'' \in U \setminus Z$, then there is no extension of~$\mu$ as~$v$ must be assigned to~$u''$ so it cannot be assigned to any node in~$Z$.
We next analyze the number of possible coloring functions and show how to enumerate them.
In order to apply the lemma more conveniently later, we will focus on the case where~$\peb$ is a pebble set.
We mention that the proof is basically the same as for \cref{lem:numberofpartitions}.

\begin{lemma}
    \label{lem:coloringfunctions}
    Let~$(B,x)$ be a pair that is almost measured with respect to two integers~$d$ and~$\alpha$ and let~$\peb$ be a pebble set for~$(B,d,\alpha,x)$.
    Let~$V, U,$ and $Z$ be sets such that~$\peb \subseteq V$ and~${Z \subseteq U}$.
    Let~$\mu \colon U \rightarrow 2^V$ be a function.
    Then, the number of possible coloring functions~$c \colon Z \rightarrow 2^{\peb}$ that are extensions of~$\mu$ is in~$\Oh(2^{66000\sqrt{d}\log^{10}(d)})$.
    All such coloring functions can be enumerated in~$\Oh(2^{66000\sqrt{d}\log^{10}(d)}d)$ time.
\end{lemma}

\begin{proof}
    Since~$(B,x)$ is almost measured with respect to~$d$ and~$\alpha$, it holds that~$x \leq 200\log^7(d)$.
    Moreover, $\peb \subseteq (B_{\he} \cup B_{\li})$, $|B_{\he}| \leq 40x\alpha\sqrt{d}$ and~$|\peb \cap B_{\li}| \leq 20x\alpha \sqrt{d} \log(d)$ because~$\peb$ is a pebble set for~$(B,d,\alpha,x)$.
    Thus, $|\peb| \leq 60 (200 \log^7(d)) \alpha \sqrt{d} \log(d) = 12000 \alpha \sqrt{d} \log^8(d)$.

    We next show that the number of coloring functions from~$Z$ to~$\peb$ is in~$\Oh(2^{66000\sqrt{d}\log^{10}(d)})$.
    This clearly also upper bounds the number of coloring functions that are extensions of~$\mu$.
    First, note that~$|Z| \leq |\peb|$ as otherwise there is no coloring function from~$Z$ to~$\peb$ as each vertex in~$Z$ is assigned a non-empty subset of~$\peb$ and each vertex in~$\peb$ is assigned exactly one vertex in~$Z$.
    Recall that an equivalent interpretation of coloring functions is to assign each vertex in~$\peb$ to exactly one vertex in~$Z$.
    Thus, the number of coloring functions is at most~$|Z|^{|\peb|} \leq |\peb|^{\peb}$.
    If~$d < 256 \cdot 2^{\alpha}$, then this number is constant, that is, in~$\Oh(1)$.
    So assume that~${d \geq 256 \cdot 2^{\alpha}}$.
    Then,
    \begin{align*}
        |\peb|^{|\peb|} &\leq (12000 \alpha \sqrt{d} \log^8(d))^{12000 \alpha \sqrt{d} \log^8(d)}\\
        &\leq (d^2 \sqrt{d} \ d^3)^{12000 \alpha \sqrt{d} \log^8(d)}\\
        &\leq 2^{66000 \sqrt{d} \log^{10}(d)}.
    \end{align*}
    The second inequality is true because~$(256\alpha)^2 > 12000\alpha$ and~$\log^8(d) \leq d^3$ for all~$d \geq 256$.
    The third inequality is due to the fact that~$\alpha \leq \log(d)$.

    Note that we can enumerate all coloring functions in~$\Oh(|Z|^{|\peb|}d) \subseteq \Oh(2^{66000 \sqrt{d} \log^{10}(d)}d)$ time by iterating over all possible choices of assigning a ``color'' to each node in~$\peb$.
    This takes~$\Oh(|Z|^{|\peb|})$ time.
    We then check in~$\Oh(d)$ time whether the function extends~$\mu$ and assigns a non-empty set to each node~$u \in Z$.
    This concludes the proof.
\end{proof}

The last ingredient for the proof of the subexponential parameterized algorithm for \minor{} is a restriction of partitions of~$\peb$ for coloring functions.
Formally, we say that a partition~$P \in \mathcal{B}(\peb)$ \emph{\textbf{refines}} a coloring function~$c \colon Z \rightarrow 2^{\peb}$ if the following holds for each pair~$v_1,v_2 \in \peb$ of vertices in~$\peb$.
Let~$u_1,u_2 \in Z$ be the nodes such that~$v_1 \in c(u_1)$ and~$v_2 \in c(v_2)$.
If~$u_1 \neq u_2$, then~$u_1$ and~$u_2$ are contained in different blocks in~$P$.
Intuitively, $P$ refines~$c$ if whenever~$c$ assigns~$v_1$ and~$v_2$ different colors, then~$v_1$ and~$v_2$ are contained in different blocks in~$P$.
Note that we \emph{do not} require that whenever~$c$ assigns two vertices the same color, then they also appear in the same block in~$P$.
It will be convenient for us to interpret such a partition as a collection of partitions, one for each node~$u$.
For a partition, we denote the subset of~$P$ containing vertices in~$c(u)$ for some~$u \in Z$ by~$P^u$.
Note that~$P = \bigcup_{u \in Z}P^u$.
We next show how to use the above to design a subexponential parameterized algorithm for \minor{} parameterized by both~$k$ and the number~$\sigma_t(H)$ of non-isomorphic separations of~$H$.
Afterwards, we show how to use this result to get a subexponential parameterized algorithm for \minor{} parameterized by only~$k$.

\begin{theorem}
    \label{thm:minor}
    Let~$t = 12000\sqrt{k}\log^8(k)$.
    \minor{} can be solved in~$2^{\tilde{\Oh}(\sqrt{k})}(\sigma_t(H))^2n^{2.49}$ time for planar graphs, where~$\sigma_t(H)$ denotes the number of non-isomorphic separations of~$H$ of size at most~$t$.
\end{theorem}

\begin{proof}
    We follow the same recipe as usual.
    We choose~$\rho = 2$ and~$\alpha_{\pi}=1$.
    The function~$\gamma$ depends on the pattern graph~$H=(F,U)$ and we choose~${\gamma_H(d) = (\sigma_t(H))^2 2^{5t} |U|^3 2^{264080\sqrt{d} \log^{10}(d)}d}$.
    The set~$S_{\pi}$ contains the weight function~$w$ and the function~$\mu$ defining the roots of the input graph~$G$ (both restricted to the current subgraph).
    A request~$r \in R_{\pi}(d,G,B,\peb,S_{\pi})$ will be a 6-tuple~$(\peb,(X,Y),c,P,\mathcal{E},\ell)$.
    Let us give a description of the different entries.
    The set~$\peb$ is a pebble set and~$(X,Y)$ is a separation of~$H$ of size at most~$t$.
    Let~$Z = X \cap Y$ be the separator and~$H[Z] = (Z,E_Z)$ be the graph induced by~$Z$.
    The entry~$c \colon Z \rightarrow 2^{\peb}$ is a coloring function that is an extension of~$\mu$ and~$P \in \mathcal{B}(\peb)$ is a partition of~$\peb$ that refines~$c$.
    The set~$\mathcal{E} \subseteq E_Z$ is a subset of the edges in~$H[Z]$ which are not yet realized by the model, and~$\ell \leq d$ is an upper bound on the number of vertices in a partial solution.
    The output for a request~$r=(\peb,(X,Y),c,P,\mathcal{E},\ell)$ is defined as follows.
    Therein, we use~$G_P = (V, E \cup \{\{u,v\} \mid c(u) = c(v) \land \forall p \in P.\ \{u,v\} \not\subseteq p\})$ as a shorthand for the graph~$G$ where we additionally add all edges between two vertices that are mapped to the same node in~$H$ but do not appear together in a block in~$P$.
    Recall that for a set~$E'$ of edges, $G-E'$ denotes the graph where all edges in~$E'$ have been removed.
    We can now state the conditions for when a subgraph~$G'=(V',E')$ of~$G$ fulfills~$r$:
    \begin{itemize}
        \item $|V'| \leq \ell$, and
        \item there exists a function~$f \colon X \rightarrow 2^{V'}$ such that
        \begin{enumerate}
            \item for each~$u \in X$, $c(u) \subseteq f(u)$,
            \item when contracting all edges~$e$ in~$G'_P$ for which there exists a~$u \in X$ with~$e \subseteq f(u)$, then the resulting graph is isomorphic to~$H[X] - \mathcal{E}$ and for each node~$u \in X$, the isomorphism maps the result of contracting all vertices in~$f(u)$ to~$u$, and
            \item for each~$u \in X$, the graph~$G'[f(u)]$ has connected components for~$c(u)$ as described by~$P^u$.
        \end{enumerate}
    \end{itemize}
    The weight of~$G'$ is~$\sum_{e\in E'}w(e)$ and the output for~$r$ is the minimum weight of any graph fulfilling~$r$ (or~$\infty$ if no such graph exists).
    
    We next show that~$S_{\pi}[G']$ for any subgraph~$G'$ and~$R_{\pi}(d,G,B,\peb,S_{\pi})$ can be computed in~$\Oh(\gamma_H(d)n^{\rho})$ time.
    Since~$S_{\pi}$ only contains the weight function and the root function~$\mu$, it can be maintained trivially in~$\Oh(n^2) = \Oh(n^\rho)$ time.
    For any given 5-tuple~$(d,G,B,\peb,S_{\pi})$, we can compute~$R_{\pi}(d,G,B,\peb,S_{\pi})$ as follows.
    Since~$\peb$ is part of the input for~$R_{\pi}$, we only need to copy the vertices, which takes~$\Oh(n)$ time.
    By \cref{lem:sigmasimple}, we can enumerate all separations~$(X,Y)$ of~$H$ of size at most~$t$ in~$\Oh(|U|^3\sigma_t(H))$ time.
    All coloring functions that are an extension of~$\mu$ can be computed in~$\Oh(2^{66000\sqrt{d}\log^{10}(d)}d)$ time by \cref{lem:coloringfunctions}.
    All possible partitions of~$P$ can be enumerated in~$\Oh(2^{66000\sqrt{d}\log^{10}(d)})$~time by~\cref{lem:numberofpartitions} and we can check for each of them whether they refine~$c$ in~$\Oh(n)$ time.
    Finally, since~$\mathcal{E}$ is a subset of the edges in~$H[Z]$, we can enumerate all possible choices in~$\Oh(2^{3t})$ time because~$H[Z]$ contains at most~$t$ vertices and is planar and therefore contains at most~$\nicefrac{5t}{2} < 3t$ edges.
    Iterating over all possible values~$\ell \leq d$ takes~$\Oh(d)$ time and iterating over all possible combinations of all of the above takes~$\Oh(|U|^3\sigma_t(H)2^{66000\sqrt{d}\log^{10}(d)}2^{66000\sqrt{d}\log^{10}(d)}2^{3t}dn^2) \subseteq \Oh(\gamma_H(d)n^{\rho})$ time.
    
    For the third step of our recipe, we define \minorr{} as follows.

    \problemdef{\minorr}
    {An integer~$d$, a $d$-outerplanar graph~$G=(V,E)$, a portal set~${B=(B_{\he},B_{\li},B_{\di})}$, an integer~$x$, an edge-weight function~$w \colon E \rightarrow \mathds{R}$, a planar graph~$H=(U,F)$, and a function~$\mu \colon U \rightarrow 2^V$, such that~$(B,x)$ is almost measured with respect to~$d$ and~$1$.}
    {Output a lookup table that returns for each request~$r=(\peb,(X,Y),c,P,\mathcal{E},\ell)$, where~$r \in R_{\pi}(d,G,B,\peb,\{w,\mu\})$ for any pebble set~$\peb$ for~$(B,d,1,x)$, the minimum weight of a subgraph~$G'$ of~$G$ fulfilling~$r$ (or~$\infty$ if no such graph exists).}

    Next, we show that Requirement R1 of \cref{thm:main-fpt} holds, that is, there is an index~$i$ such that we can remove each layer~$i+j(k+1)$ of a BFS layering for any~$j$.
    To this end, note that whether subgraph~$G'$ fulfills a request~$r$ or not only depends on the subgraph itself and not on the rest of the graph.
    Thus, all vertices and edges in a partial solution of size at most~$k$ are incident to at most~$k$ layers in any BFS layering and if there are more than~$k$ layers, then there is an index~$i$ such that removing all layers~$i + j(k+1)$ for any~$j$ does not invalidate the partial solution.
    Moreover, no new solutions are created in the process, so Requirement R1 holds for~$\alpha_{\pi} = 1$ for \minorr{} and \minor.

    The penultimate step is to show that Requirement R4 holds, that is, there is an algorithm that given a separator~$C=(C_{\he},C_{\li},C_{\di})$ such that~$(C,1)$ is measured for~$d$ and~$1$, a pebble set~$\peb$ for~$(B,d,1,x)$, and the solutions for instances
    \begin{align*}
        (d,C_{\inte},((B_{\he} \cup C_{\he})[C_{\inte}],(B_{\li} &\cup C_{\li})[C_{\inte}],(B_{\di} \cup C_{\di})[C_{\inte}]),x+1,S_{\pi}[C_{\inte}]) \text{ and}\\
        (d,C_{\ex},((B_{\he} \cup C_{\he})[C_{\ex}],(B_{\li} &\cup C_{\li})[C_{\ex}],(B_{\di} \cup C_{\di})[C_{\ex}]),x+1,S_{\pi}[C_{\ex}]),
    \end{align*}
    can compute the solutions for all requests~${r \in R_\pi(d,G,B,\peb,S_{\pi})}$ that do not intersect~$C_{\di}$ and intersect~$C_{\li}$ in at most~$20\sqrt{d}\log(d)$ vertices in~$\Oh(\gamma_H(d)n^2)$ time.
    So assume that we are given~$C$, solutions~$\res_1$ for the interior instance and~$\res_2$ for the exterior instance, and a pebble set~$\peb$ for~$(B,d,1,x)$.
    
    We first present the algorithm, then prove its correctness and finally analyze its running time.
    Let~$\peb_{\inte}$ and~$\peb_{\ex}$ be the subsets of~$\peb$ which appear in~$C_{\inte}$ and~$C_{\ex}$, respectively.
    We iterate over all possible pebble sets~$\peb'$ for~$(C,d,1,1)$, that is, sets which do not contain any vertices in~$C_{\di}$ and at most~$20 \sqrt{d} \log(d)$ vertices in~$C_{\li}$.
    We also iterate over all possible values~$|\peb'| \leq \ell' \leq \ell$.
    Let~${\peb'_{\inte} = \peb_{\inte} \cup \peb'}$ and~${\peb'_{\ex} = \peb_{\ex} \cup \peb'}$.
    In the following, we iterate over all possible requests for~$\res_1$ and~$\res_2$ and check whether they can be combined.
    Whenever one of the following conditions does not hold, we simply disregard the iteration and continue with the next.
    Hence, we implicitly assume that all of the following conditions hold.
    We start by iterating over all pairs of non-isomorphic separations~$(X_1,Y_1)$ and~$(X_2,Y_2)$ of~$H$ of size at most~$t$ and check whether~$H[X_1 \cup X_2]$ is isomorphic to~$H[X]$ and whether~$X_1 \cap X_2 \subseteq Y_1 \cap Y_2$.
    Let~$Z_1 = X_1 \cap Y_1$, $Z_2  = X_2 \cap Y_2$, $H[Z_1] = (V_1,E_1)$, and~$H[Z_2] = (V_2,E_2)$.
    Next, we iterate over all possible pairs of coloring functions~$c_1 \colon Z_1 \rightarrow 2^{\peb'_{\inte}}$ and~$c_2 \colon Z_2 \rightarrow 2^{\peb'_{\ex}}$ that are extensions of~$\mu$.
    We then check whether~${c(u) \cap \peb_{\inte} \subseteq c_1(u)}$ for all~$u \in Z$, $c(u) \cap \peb_{\ex} \subseteq c_2(u)$, and whether for all~$u \in U$ and all~$v \in \peb'$ it holds that~$v \in c_1(u)$ if and only if~$v \in c_2(u)$.
    Intuitively, the first two requirements ensure that whenever~$c$ assigns a vertex to the model of a particular node~$u$ in~$H$, then~$c_1$ and~$c_2$ also assign the vertex to the model of~$u$ whenever the vertex~$v$ is part of their domain.
    The last check ensures that both~$c_1$ and~$c_2$ agree on the vertices in~$\peb'$.
    We also iterate over all pairs of partitions~$P_1 \in \mathcal{B}(\peb'_{\inte})$ and~$P_2 \in \mathcal{B}(\peb'_{\ex})$ that refine~$c_1$ and~$c_2$, respectively.
    Here, we also perform two checks.
    First, we check for all nodes~$u \in Z$ whether~$P^u = P_1^u \sqcup P_2^u$.
    Second, for all nodes~$u \in (X_1 \cap X_2) \setminus Y$, we check whether~$P_1^u \sqcup P_2^u$ consists of a single set.
    Finally, we iterate over all pairs of set~$\mathcal{E}_1 \subseteq E_1$ and~$\mathcal{E}_2 \subseteq E_2$ and check whether~$\mathcal{E}_1 \cap \mathcal{E}_2 = \mathcal{E}$.
    If all of the above conditions hold, then we compute~$\res_1(\peb'_{\inte},(X_1,Y_1),c_1,P_1,\mathcal{E}_1,\ell') + \res_2(\peb'_{\ex},(X_2,Y_2),c_1,P_1,\mathcal{E}_2,\ell-\ell'+|\peb'|)$ and update the result for~$r$ if the computed value is smaller than the current entry.
    This concludes the description of the algorithm.

    We next show that the algorithm for Requirement R4 is correct.
    To this end, we first show that an optimal solution is found by the algorithm and then show that whenever the algorithm updates a solution, then this corresponds to a subgraph fulfilling~$r$.
    So consider the minimum-weight subgraph~$G'=(V',E')$ of~$G$ such that~$G'$ fulfills request~$r=(\peb,(X,Y),c,P,\mathcal{E},\ell)$, $G'$ does not contain any vertices in~$C_{\di}$, and~$G'$ contains at most~$20\sqrt{d}\log(d)$ vertices from~$C_{\li}$.
    Moreover, let~$f$ be a function that maps~$X$ to subsets of~$V'$ that witnesses that~$G'$ fulfills~$r$.
    We will show that the value for request~$r$ is updated with the weight of~$G'$.
    To this end, we first describe the iteration that leads to this update and then show that all checks from our algorithm are met and therefore the value is updated.

    Let~$\peb_{\inte}$ and~$\peb_{\ex}$ be the subsets of~$\peb$ that appear in~$C_{\inte}$ and~$C_{\ex}$, respectively.
    Let~$V_{\inte}$ and~$V_{\ex}$ be the vertices in~$G'$ in~$C_{\inte}$ and~$C_{\ex}$, respectively, and let~$\peb' = V_{\inte} \cap V_{\ex}$ be the vertices in~$G'$ in~$C$.
    Note that~$\peb' \cap C_{\di} = \emptyset$ and~$|\peb' \cap C_{\li}| \leq 20\sqrt{d}\log(d)$ by assumption (since Requirement~R4 only requires us to find solutions that satisfy these two properties).
    Let~$\peb'_{\inte} = \peb_{\inte} \cup \peb'$ and~${\peb'_{\ex} = \peb_{\ex} \cup \peb'}$.
    Let~$E_{\inte} \subseteq E'$ be the set of edges in~$G'$ with both endpoints in~$V_{\inte}$ and let~$E_{\ex} = E' \setminus E_{\inte}$ be the set of edges with at least one endpoint in~$V_{\ex} \setminus V_{\inte}$.
    Let~$G'_{\inte} = (V_{\inte},E_{\inte})$ and~$G'_{\ex} = (V_{\ex},E_{\ex})$.
    Let~$X_{\inte}$ and~$X_{\ex}$ be the set of nodes~$u \in U$ such that~$f(u) \cap V_{\inte} \neq \emptyset$ and~$f(u) \cap V_{\ex} \neq \emptyset$, respectively.
    Let~$P_{\inte}$ be the partitions of~$\peb'_{\inte}$ such that two nodes~$v_1,v_2$ appear in the same block in~$P_{\inte}$ if and only if there exists a~$u\in X_{\inte}$ with~$v_1,v_2 \in f(u)$ and~$v_1$ and~$v_2$ are in the same connected component in~$G'_{\inte}[f(u)]$.
    Let~$\mathcal{E}_{\inte}$ be the set of edges~$\{u_1,u_2\}$ in~$H[X_{\inte}]$ such that there is no edge~$\{v_1,v_2\}$ in~$G'$ with~$v_1 \in f(u_1)$ and~$v_2 \in f(u_2)$.
    Let~$Y_{\inte}$ be the set of nodes~$u \in U$ such that (i)~$u \notin X_{\inte}$, (ii)~$P^u_{\inte}$ contains more than one set, or~(iii)~at least one edge in~$\mathcal{E}_{\inte}$ is incident to~$u$.  
    Let~$c_{\inte} \colon X_{\inte} \rightarrow 2^{V_{\inte}}$ be defined as~$c(u) = f(u) \cap \peb'_{\inte}$.
    Let~$P_{\ex},\mathcal{E}_{\ex},Y_{\ex}$ and~$c_{\ex}$ be defined analogously.
    Finally, let~$\ell' = |V_{\inte}|$.
    Note that since~$|V'| \leq \ell$ and $|V_{\inte} \cap V_{\ex}| = |\peb'|$, it holds that~$|V_{\ex}| \leq \ell - \ell' + |\peb'|$.
    Let~$r_{\inte} = (\peb'_{\inte},(X_{\inte},Y_{\inte}),c_{\inte},P_{\inte},\mathcal{E}_{\inte},\ell')$ and~$r_{\ex} = (\peb'_{\ex},(X_{\ex},Y_{\ex}),c_{\ex},P_{\ex},\mathcal{E}_{\ex},\ell-\ell'+|\peb'|)$.
    
    We next show that the above values satisfy all checks from our algorithm.
    In particular, we will show that~$G'_{\inte}$ and~$G'_{\ex}$ fulfill requirements~$r_{\inte}$ and~$r_{\ex}$.
    Some checks are done for the interior and exterior separately and others check combinations of both.
    Since the arguments for the former will always be identical, we only show them for~$r_{\inte}$.
    The first check is whether~$\peb'$ is a pebble set for~$(C,d,1,1)$.
    Since~$\peb' \cap C_{\di} = \emptyset$ and~$|\peb' \cap C_{\li}| \leq 20\sqrt{d}\log(d)$, the set~$\peb'$ is indeed considered as a pebble set for~$(C,d,1,1)$.
    The next check in our algorithm is whether~$H[X_{\inte} \cup X_{\ex}]$ is isomorphic to~$H[X]$ and whether~$X_{\inte} \cap X_{\ex} \subseteq Y_{\inte} \cap Y_{\ex}$.
    Note that~$f$ assigns each node in~$X$ to a non-empty subset of~$V'$ by definition and does not assign any vertices to nodes~$u \in U \setminus X$.
    Since~$V_{\inte} \cup V_{\ex} = V'$, it holds that~$X_{\inte} \subseteq X$ and any node~$u \in X \setminus X_{\inte}$ is assigned to at least one vertex~$V_{\ex}$ and is therefore contained in~$X_{\ex}$, that is, $X_{\inte} \cup X_{\ex} = X$ and therefore~$H[X_{\inte} \cup X_{\ex}] = H[X]$.
    Now assume towards a contradiction that~$X_{\inte} \cap X_{\ex}\not\subseteq Y_{\inte} \cap Y_{\ex}$.
    Then, there exists a node~$u \in (X_{\inte} \cap X_{\ex}) \setminus (Y_{\inte} \cap Y_{\ex})$.
    By the construction of~$Y_{\inte}$ it holds that~$P^u_{\inte}$ contains at most one set and no edge in~$\mathcal{E}_{\inte}$ is incident to~$u$.
    Thus, $G'_P[f(u)] = G'[f(u)]$ and we can contract all these vertices into a single node, in particular,~$G'[f(u)]$ is connected.
    However, $u$ is not in the domain of~$c_{\inte}$ or~$c_{\ex}$ and~$f(u)$ therefore contains no vertices in~$C$.
    Since~$C$ is a separator and~$u \in X_1 \cap X_2$, it holds that~$f(u)$ contains vertices in both~$V_{\inte} \setminus V_{\ex}$ and in~$V_{\ex} \setminus V_{\inte}$, a contradiction to~$G'[f(u)]$ being connected.
    Hence, $X_{\inte} \cap X_{\ex} \subseteq Y_{\inte} \cap Y_{\ex}$.
    Our algorithm then checks whether~$c_{\inte}$ is an extension of~$\mu$ and whether~$c(u) \cap M_{\inte} \subseteq c_{\inte}(u)$.
    Note that the former holds as~$f$ is an extension of~$\mu$ by definition and~$c_{\inte}$ is an extension of~$f$ (note that~$c$ and~$c_{\inte}$ are both extended by~$f$ and extensions of~$f$ as they assign each node in~$U$ to the same set of vertices just restricted to their respective domains).
    The latter holds as~$f$ is an extension of~$c$.
    We also check that for all~$u \in U$ and all~$v \in \peb'$ that~$v \in c_{\inte}(u)$ if and only if~$v \in c_{\ex}(u)$.
    Note that this trivially holds as both~$c_1$ and~$c_2$ are extensions of~$f$, that is, it holds for all~$u \in U$ and~$v \in \peb'$ that~$v \in c_{\inte}(u)$ if and only if~$v \in f(u)$ if and only if~$v \in c_{\ex}(u)$. 
    
    Then, we check whether~$P_{\inte}$ refines~$c_{\inte}$, whether~$P_{\inte}^u \sqcup P_{\ex}^u = P^u$ for all~$u \in X \cap Y$, and~$P_{\inte}^u \sqcup P_{\ex}^u$ consists of a single set for each~$u \in (X_{\inte} \cap X_{\ex})\setminus Y$.
    The first one is true as by definition we never add two vertices~$v_1,v_2$ to a common block if there are~$u_1 \neq u_2$ such that~$v_1 \in f(u_1)$ and~$v_2 \in f(u_2)$.
    Since~$c_{\inte}$ refines~$f$, the same also holds for~$c_{\inte}$ and therefore~$P$ refines~$c_{\inte}$.
    For the second check, consider any~$u \in X \cap Y$ and the subgraph~$G'[f(u)]$.
    By assumption this graph has connected components for~$c(u)$ as described by~$P^u$.
    By definition of~$P_{\inte}$ two nodes~$v_1,v_2$ appear in the same block in~$P_{\inte}$ if and only if there exists a~$u\in X_{\inte}$ with~$v_1,v_2 \in f(u)$ and~$v_1$ and~$v_2$ are in the same connected component in~$G'_{\inte}[f(u)]$.
    Thus~$G'_{\inte}[f(u)]$ has connected components for~$c_{\inte}(u)$ as described by~$P^u_{\inte}$ and the same for~$c_{\ex}(u)$ and~$P^u_{\ex}$.
    Since~$G'[f(u)]$ is the union of~$G'_{\inte}[f(u)]$ and~$G'_{\ex}[f(u)]$, \cref{lem:connectivity} yields that~$P_{\inte}^u \sqcup P_{\ex}^u = P^u$.
    For the third check, note that since~$f$ witnesses that~$G'$ fulfills~$r$, it holds for all nodes~$u \in X \setminus Y$ that~$G'[f(u)]$ only has one connected component as~$u$ is not in the domain of~$P$ and therefore~$G'_P$ does not add any edges between vertices in~$f(u)$ but they are still contracted to a single vertex.
    As shown above~$G'_{\inte}[f(u)]$ has connected components for~$c_{\inte}(u)$ as described by~$P_{\inte}^u$ and the same for~$c_{\ex}(u)$ and~$P_{\ex}^u$.
    By \cref{lem:connectivity} $P_{\inte}^u \sqcup P_{\ex}^u$ describes a connected graph, that is, a partition containing only a single set.
    The final check of our algorithm is whether~$\mathcal{E}_{\inte} \cap \mathcal{E}_{\ex} = \mathcal{E}$.
    Recall that~$\mathcal{E}_{\inte}$ is the set of edges~$\{u_1,u_2\}$ in~$H[X_{\inte}]$ such that there is no edge~$\{v_1,v_2\}$ in~$G'$ with~$v_1 \in f(u_1)$ and~$v_2 \in f(u_2)$.
    For any edge~$\{u_1,u_2\} \in \mathcal{E}$, it holds that there is no edge~$\{v_1,v_2\}$ in~$G'$ with~$v_1 \in f(u_1)$ and~$v_2 \in f(u_2)$ and thus the same holds for~$G'_{\inte}$, that is, $\{u_1,u_2\} \in \mathcal{E}_{\inte}$.
    As the same argument also applies to~$\mathcal{E}_{\ex}$, it holds that~$ \mathcal{E} \subseteq \mathcal{E}_{\inte} \cap \mathcal{E}_{\ex}$.
    In the other direction, for any edge contained in~$\mathcal{E}_{\inte} \cap \mathcal{E}_{\ex}$ it holds that there is no edge~$\{v_1,v_2\}$ in~$G'_{\inte}$ or in~$G'_{\ex}$ with~$v_1 \in f(u_1)$ and~$v_2 \in f(u_2)$.
    Since the edges in~$G'$ are the union of the edges in~$G'_{\inte}$ and~$G'_{\ex}$, it also holds that there is no edge in~$G'$ satisfying the above.
    Thus, the edge has to be contained in~$\mathcal{E}$ as~$f$ does not yet model the edge but~$f$ witnesses that~$G'$ satisfies~$r$.
    All checks are therefore successful.
        
    Finally, we show that the correct value for the weight of~$G'$ is computed.
    To this end, observe that~$G'_{\inte}$ and~$G'_{\ex}$ do not share any edges and the union of both graphs is by definition~$G'$.
    As~$G'_{\inte}$ and~$G'_{\ex}$ fulfill requirements~$r_{\inte}$ and~$r_{\ex}$ as shown above and they do not contain any vertices from~$C_{\di}$ and at most~$20\sqrt{d}\log(d)$ vertices from~$C_{\li}$ as~$G'$ has the same property, their weight (or something smaller) is returned by the two queries and their sum is stored.
    Thus, we return a result that is at most the weight of~$G$, concluding the first direction of the correctness proof.
    
    Next, we show the other direction of the correctness proof, that is, we show that whenever we update a solution, then this corresponds to a subgraph of the given weight that fulfills request~$r$.
    To this end, let~$r_1=(\peb_1,(X_1,Y_1),c_1,P_1,\mathcal{E}_1,\ell_1)$ and~$r_2=(\peb_2,(X_2,Y_2),c_2,P_2,\mathcal{E}_2,\ell_2)$ be the requests that are used to update a solution for request~$r=(\peb,(X,Y),c,P,\mathcal{E},\ell)$.
    Let~${G_1 = (V_1,E_1)}$ and~${G_2=(V_2,E_2)}$ be the graphs of minimum weight that fulfill these requests and let~$f_1$ and~$f_2$ be the functions that witness that the graphs fulfill the requests.
    We show that~${G^* = (V_1 \cup V_2, E_1 \cup E_2)}$ fulfills requirement~$r$.
    Since the weight of the union is clearly at most the sum of the two individual weights (which is the result we report), this will conclude the proof of correctness.
    To this end, let~$\peb' = \peb_1 \cap (C_{\he} \cup C_{\li} \cup C_{\di})$.
    Note that since we queried~$r_1$ and~$r_2$ for cycle~$C$, it holds that~$\peb' = \peb_2 \cap \peb_1$ and both~$G_1$ and~$G_2$ contain precisely the vertices in~$\peb$ from~$C$, that is, no vertices from~$C_{\di}$ and at most~$20 \sqrt{d} \log(d)$ vertices from~$C_{\li}$.
    Moreover, note that~$\ell_2 = \ell - \ell_1 + |\peb'|$ and hence~$|V_1 \cup V_2| \leq \ell_1 + \ell_2 - |\peb'| = \ell$.
    It remains to show that there is a function~$f$ that witnesses that~$(V_1 \cup V_2, E_1 \cup E_2)$ fulfills request~$r$.
    To this end, we will consider functions~$f'_1$ and~$f'_2$ which both have domain~$X$ and~$f'_1(u) = \emptyset$ if~$u \notin X_1$ and~$f'_1(u) = f_1(u)$ otherwise.
    The function~$f'_2$ is defined analogously and~$f(u) = f'_1(u) \cup f'_2(u)$ for all~$u \in X$.
    Note that each vertex in~$G^*$ appears in exactly one set~$f(u)$ as the vertices in~$X_1 \cap X_2$ are exactly the vertices in~$\peb'$ and~$c_1$ and~$c_2$ agree on those vertices by construction.
    
    The three requirements we need to show are (i) for each~$u \in X$ it holds that~$c(u) \cup \mu(u) \subseteq f(u)$, (ii) when contracting all edges~$e$ in~$G^*_P$ for which there exists a~$u \in X$ with~$e \subseteq f(u)$, then the resulting graph is isomorphic to~$H[X] - \mathcal{E}$ and the isomorphism maps for each node~$u \in X$ the vertex that is the result of contracting all vertices in~$f(u)$ to~$u$, and (iii) for each~$u \in X$, the graph~$G'[f(u)]$ has connected components for~$c(u)$ as described by~$P^u$.
    Let~$u \in X$ be any node.
    For the first requirement, note that each node~$v \in \mu(u)$ is contained in~$f(u)$ as both~$c_1$ and~$c_2$ are extensions of~$\mu$ and therefore~$v \in f_1(u)$ if~$u \in X_1$ and otherwise~$v \in f_2(u)$ as~$X = X_1 \cup X_2$ by construction.
    In either case, it holds that~$v \in f(u)$ and the first requirement is fulfilled.

    We next show the third requirement.
    This follows from \cref{lem:connectivity} as follows.
    We make a case distinction whether~$u \in X_1 \cap X_2$ or not.
    If~$u \notin X_1 \cap X_2$, then assume without loss of generality that~$u \in X_1 \setminus X_2$.
    Then, the following holds as~$G_2$ does not contain any vertices in~$f(u)$ and, in particular, $\peb'$ also does not contain any vertices from~$c_1(u)$: $G^*[f(u)] = G_1[f(u)]$, $c(u) = c_1(u)$, and~$P^u = P_1^u$.
    Note that~$f_1$ guarantees that~$G_1[f(u)] = G^*[f(u)]$ has connected components for~$c_1(u) = c(u)$ as described by~$P^u_1 = P^u$.
    So assume next that~$u \in X_1 \cap X_2$.
    Then, $u$ is part of the domain of both~$c_1$ and~$c_2$ and~$c_1(u) \cap c_2(u) = f(u) \cap \peb'$.
    We have that~$G_1[f(u)]$ has connected components for~$c_1(u)$ as described by~$P_1^u$ for each~$u \in X_1$ and~$G_2[f(u)]$ has connected components for~$c_2(u)$ as described by~$P_2^u$ for each~$u \in X_2$.
    Moreover, $P^u = P_1^u \sqcup P_2^u$ by construction.
    Thus, \cref{lem:connectivity} states that the union of the two graphs (which is equal to~$G^*[f(u)]$) has connected components for~$c(u) \subseteq c_1(u) \cup c_2(u) \setminus \peb'$ as described by~$P_1^u \sqcup P_2^u$.
    This shows that the third requirement is fulfilled.

    Finally, we prove that the second requirement holds.
    To this end, notice that~$f_1$ witnesses that~${G_1}_{P_1}$ is a model for~$H[X_1] - \mathcal{E}_1 = (X_1,E'_1)$ and~$f_2$ is a witness for the fact that~${G_2}_{P_2}$ is a model for~$H[X_2] - \mathcal{E}_2 = (X_2,E'_2)$.
    We show that~$f$ witnesses that~$G^*_P$ is a model for~$H[X] - \mathcal{E}$.
    First, we show that~$G^*_P[f(u)]$ is connected for each~$u \in X$.
    Afterwards, we show that each edge in~$H[X]-\mathcal{E}$ is modeled by~$G^*_P$.
    Towards the former, note that~$G^*_P[f(u)]$ is only potentially disconnected if~$G^*[f(u)]$ contains a connected component that does not contain any vertex in~$\peb$ as all these vertices appear in~$P^u$ and are either contained in the same block or not.
    In the former case, the vertices appear in the same connected component of~$G^*[f(u)]$ as given by the third requirement.
    In the latter case, an edge is added in~$G^*_P$ and thus the graph is trivially connected.
    So assume that~$G^*[f(u)]$ contains a connected component that does not contain any vertex in~$\peb$.
    Note that this implies that~$u \notin Y$ by the definition of~$c$ and~$f$.
    If~$u \notin X_1 \cap X_2$, then assume without loss of generality that~$u \in X_1 \setminus X_2$.
    Then, $P^u = P_1^u$ holds as~$G_2$ does not contain any vertices in~$f(u)$ and therefore~$G^*_P[f(u)] = {G_1}_{P_1}[f(u)]$ and since~$f_1$ witnesses that~${G_1}_{P_1}[f(u)]$ is connected, the same holds for~$G^*_P[f(u)]$.
    If~$u \in X_1 \cap X_2$, then~$u \in (X_1 \cap X_2) \setminus Y$ as~$u \notin Y$.
    Then, we require that~$P_1^u \sqcup P_2^u$ consists only of a single set.
    Since~$G_1[f(u)]$ and~$G_2[f(u)]$ have connected components for~$c_1(u) = c_2(u) \subseteq \peb'$ as described by~$P_1^u$ and~$P_2^u$ and since~$P_1^u \sqcup P_2^u$ only consists of a single set, \cref{lem:connectivity} states that the union~$G^*[f(u)]$ has only one connected component.
    Thus, $G^*_P[f(u)]$ indeed has only one connected component.
    
    Next, we show that each edge in~$H[X] - \mathcal{E}$ is modeled.
    Since~$H[X_1 \cup X_2]$ is isomorphic to~$H[X]$, it holds for all edges in~$H[X_1] - (\mathcal{E}_1 \cup \mathcal{E}_2)$ that they are modeled in~$G^*$.
    It remains to analyze the edges in~$\mathcal{E}_1 \setminus \mathcal{E}_2$ (the argument for edges in~$\mathcal{E}_2 \setminus \mathcal{E}_1$ is symmetric).
    Note that each such edge is modeled in~$G_2$ and therefore also by~$G^*$.
    Hence, $G^*_P$ is a model for~$H[X] - \mathcal{E}$ as witnessed by~$f$ and the second requirement is fulfilled.
    All three requirements are fulfilled and~$G^*$ therefore fulfills request~$r$, concluding the proof of correctness for our algorithm for Requirement R4.
    
    We next analyze the running time of the algorithm for requirement R4.
    The number of possible pebble sets~$\peb'$ for~$(C,1)$ is by \cref{lem:numberofpebbles} in~$\Oh(2^{80\sqrt{d}\log^2(d)})$ and they can be compute in~$\Oh(2^{80\sqrt{d}\log^2(d)}d)$ time.
    Iterating over all~$\ell \leq d$ adds an additional factor of~$d$.
    Computing~$\peb_{\inte}, \peb_{\ex}, \peb'_{\inte},$ and~$\peb'_{\ex}$ can all easily be done in~$\Oh(n^2)$ time.
    We then iterate over all pairs of non-isomorphic separations of size at most~$t$ of~$H$, which can be computed in~$\Oh(|U|^3\sigma_t(H))$ time by \cref{lem:sigmasimple}.
    The number of possible pairs is in~$\Oh((\sigma_t(H))^2)$ and for each pair, we can check whether they are isomorphic in~$\Oh(|U|)$ time~\cite{HW74}.
    Computing~$Z_1,Z_2,H[Z_1]$, and~$H[Z_2]$ take~$\Oh(|U|)$ time.
    All coloring functions~$c_1$ and~$c_2$ can be computed in~$\Oh(2^{66000\sqrt{d}\log^{10}(d)}d)$~time by \cref{lem:coloringfunctions}.
    The number of pairs is in~$\Oh((2^{66000\sqrt{d}\log^{10}(d)})^2) = \Oh(2^{132000\sqrt{d}\log^{10}(d)})$ and the check whether they satisfy our three requirements for them takes~$\Oh(n)$ time.
    Next, we iterate over all pairs~$(P_1,P_2)$ of partitions for~$\peb'_{\inte}$ and~$\peb'_{\ex}$ (which can be computed in~$\Oh(2^{66000\sqrt{d}\log^{10}(d)})$~time by \cref{lem:numberofpartitions}), respectively, and compute~$P^u_1 \sqcup P^u_2$ for all~$u \in Z$ and perform two simple checks.
    The number of pairs is in~$\Oh(2^{132000\sqrt{d}\log^{10}(d)})$.
    Note that~$P^u_1 \sqcup P^u_2$ for all~$u \in Z$ can be computed in one step by computing~$P_1 \sqcup P_2$ and therefore in~$\Oh(n^2)$ time.
    Finally, we iterate over all pairs of subsets~$\mathcal{E}_1,\mathcal{E}_2$ of edges in~$E_1$ and~$E_2$, respectively.
    Since~$\peb'_{\inte}$ and~$\peb'_{\ex}$ both have size at most~$t$, the number of edges in~$E_1$ and~$E_2$ is at most~$\nicefrac{5t}{2}$, each.
    Thus, the number of possible combinations is at most~$2^{5t}$.
    Checking whether~$\mathcal{E}_1 \cap \mathcal{E}_2 = \mathcal{E}$ can clearly be done in~$\Oh(n^2)$ time and computing the sum of two numbers and checking whether it is smaller than the current entry takes~$\Oh(1)$ time.
    Thus, the overall running time for our algorithm for Requirement R4 is in
    \begin{align*}
    &\Oh(2^{80\sqrt{d}\log^2(d)}d (n^2 + (|U|^3 \sigma_t(H) + ((\sigma_t(H))^2 (|U| + 2^{66000\sqrt{d}\log^{10}(d)}n \\
    &\quad + 2^{132000\sqrt{d}\log^{10}(d)} (n + 2^{66000\sqrt{d}\log^{10}(d)} + 2^{132000\sqrt{d}\log^{10}(d)}(n^2+2^{5t}(n^2 + 1))))))))\\
    =\ &\Oh(2^{80\sqrt{d}\log^2(d)}d 2^{132000\sqrt{d}\log^{10}(d)} 2^{132000\sqrt{d}\log^{10}(d)} |U|^3 (\sigma_t(H))^2 2^{5t} n^2)\\
    \subseteq\ &\Oh(2^{264080\sqrt{d}\log^{10}(d)}d |U|^3 (\sigma_t(H))^2 2^{5t} n^2)\\
    =\ &\Oh(\gamma_H(d) n^{\rho}).
    \end{align*}
    
    The final step is to present a reduction from \minor{} to \minorr{} and analyze the running time.
    To this end, we apply \cref{lem:baker} with~$d=k+1$ to compute~$k+1$ $k$-outerplanar graphs.
    We then solve \minorr{} for each of these graphs using \cref{thm:main-fpt} with~$d=k$, portal set~${B=(\emptyset,\emptyset,\emptyset)}$, $x=1$, and the input weight and root functions.
    For each such iteration, we query the optimal solution for request~$(\emptyset,(U,\emptyset),c,\emptyset,\emptyset,k)$, where~$c \colon \emptyset \rightarrow 2^{V}$ is an empty function.
    If any iteration finds a solution with weight at most~$W$, then we output yes and otherwise we output no.
    Note that any subgraph fulfilling the request~$(\emptyset,(U,\emptyset),c,\emptyset,\emptyset,k)$ corresponds to a subgraph~$G'$ with at most~$k$ vertices that contains~$H[U] - \emptyset = H$ as a rooted minor.
    Since the weight of~$G'$ is at most~$W$, this shows that there is a subgraph of~$G$ of weight at most~$W$ that contains~$H$ as a rooted minor and hence we indeed have a yes-instance.
    In the other direction, any subgraph of weight at most~$W$ that contains~$H$ as a rooted minor also fulfills the given request and thus the algorithm is correct.
    Using \cref{thm:main-fpt}, the arguments above, the fact that we may assume\footnote{If~$k < |U|$, then the instance is trivially a no-instance as contracting edges can never increase the number of vertices and if~$k > n$, then setting~$k = n$ yields an equivalent instance.}~$|U| \leq k \leq n$, and~$t = 12000\sqrt{k}\log^8(k)$, solving each instance of \minorr{} takes time
    \[2^{\tilde{\Oh}(\sqrt{k})} 2^{264080\sqrt{k}\log^{10}(k)}k |U|^3 (\sigma_t(H))^2 2^{5t} n^{2.49} \subseteq 2^{\tilde{\Oh}(\sqrt{k})}(\sigma_t(H))^2 n^{2.49}.\]
    The time for computing all inputs is in~$\Oh(kn)$ by \cref{lem:baker}, because induced subgraphs can be computed in~$O(n+m)$ time, and since~$m \in \Oh(n)$ for planar graphs.
    The overall running time is therefore in~${\Oh(nk) + \Oh((k+1) 2^{\tilde{\Oh}(\sqrt{k})} (\sigma_t(H))^2 n^{2.49}) = 2^{\tilde{\Oh}(\sqrt{k})}(\sigma_t(H))^2 n^{2.49}}$.
    This concludes the proof.
\end{proof}

We next show how to get a subexponential parameterized algorithm when only parameterizing by~$k$ and discuss relevant special cases of \minor.
Note that we can assume without loss of generality that~$k \geq |U|$ as otherwise the answer is no as any graph with~$k$ vertices cannot be a model for~$H$.
Combined with \cref{lem:sigmasimple} this immediately implies a subexponential algorithm for the parameter~$k$.

\wsi{} is a special case of \minor{} where~$k = |U|$ and \wsp{} is the special case where~$H$ consists of an independent set~$I=\{u_1,u_2\ldots,u_t\}$ of size~$t$ and~$\mu(u_i)=X_i$ for each~$i \in [t]$.
The number of non-isomorphic separations~$(X,Y)$ of size at most~$t$ for any graph with~$t$ vertices is at most~$3^{t}$ as the separation selects for each vertex whether it is contained in~$X, Y$, or both.
Hence, we get the following corollary from \cref{thm:minor}.

\begin{corollary}
    Let~$t = 12000 \sqrt{k} \log^8(k)$.
    \minor{} can be solved in~$2^{\Oh(\nicefrac{k}{\log(k)})}n^{2.49}$ time for planar graphs,
    \wsi{} for planar graphs can be solved in~$2^{\tilde{\Oh}(\sqrt{|U|})}\sigma_t(H)n^{2.49}$ time and in~$2^{\Oh(\nicefrac{k}{\log(k)})}n^{2.49}$ time, where~$\sigma_t(H)$ denotes the number of non-isomorphic separations of~$H$ of size at most~$t$, and \wsp{} can be solved in~$2^{\tilde{\Oh}(\sqrt{k})}n^{2.49}$ time for planar graphs.
\end{corollary}

\subsection{Parallel Induced Minors}
We next investigate induced minors with parallel contractions.
Studying induced minors rather than minors only makes a small difference.
We simply replace a subgraph by an induced subgraph and whenever we combine two solutions, we compute all induced edges in the separator and subtract them from the total.
Dealing with a maximization problem rather than a minimization problem also poses no noticeable hurdle.
The main obstacle are the parallel contractions for two reasons.
First and foremost, the number of non-isomorphic separations of a planar multigraph is exponential in its number of vertices even for separations of size zero as we will show.
Second, if one tries to keep track of which edges are already part of the model, the number of possibilities is exponential even for just a pair of vertices with many edges between them.
Instead, we will keep track for each underlying edge (including self loops) how many parallel edges are already modeled using \emph{\textbf{counting functions}}.

Let us now discuss why the number of non-isomorphic separations of size zero of a planar multigraph~$H=(U,F)$ is exponential in~$|U|$ and how to resolve this issue.
Let~$H$ consist of~$\nicefrac{|U|}{2}$ pairs of vertices, where the~$i$\textsuperscript{th} pair is connected by~$i$ parallel edges.
Any separation of size zero partitions these pairs into two subsets and since each such partition is a non-isomorphic separation, the number of such separations is~$2^{\nicefrac{|U|}{2}}$.
Instead of using~$|U|$ as an upper bound on the number of non-isomorphic separations and a lower bound for~$k$, we show that the parameter~$\psi(H)=\max(|U|,\frac{2|F|}{5})$ can be used in the case of planar multigraphs.
To this end, we first show that we can always assume that~$k \geq \psi(H)$.

\begin{lemma}
    \label{lem:psi}
    Let~$G$ be a simple planar graph and~$H$ be a planar multigraph. Any model for~$H$ in~$G$ contains at least~$\psi(H)$ vertices.
\end{lemma}
\begin{proof}
    Assume towards a contradiction that a model~$G'=(V',E')$ with~$|V'| < \psi(H)$ vertices exists.
    We make a case distinction whether~$\psi(H)=|U|$ or~$\psi(H) = \nicefrac{2|F|}{5}$.
    In the former case, note that~$G'$ contains~$|V'| \leq |U|$ vertices.
    Since~$G'$ is a model of~$H$, a sequence of (parallel) edge contractions transforms~$G'$ into~$H$.
    However, the number of vertices reduces with each edge contraction.
    Thus, any graph~$H'$ for which~$G'$ is a model contains at most~$|V'| < |U|$ vertices, contradicting the fact that~$G'$ is a model for~$H$.

    Now consider the case where~$\psi(H) = \nicefrac{2|F|}{5}$.
    Since~$G'$ is a subgraph of~$G$, it holds that~$G'$ is a simple planar graph and hence 5-degenerate.
    This implies~${|E'| \leq \nicefrac{5|V'|}{2} < \nicefrac{5\psi(H)}{2} = \nicefrac{5\frac{2|F|}{5}}{2} = |F|}$.
    Similar to the first case, observe that the number of edges decreases (by one) with each (parallel) edge contraction.
    This contradicts the fact that~$G'$ is a model for~$H$.
\end{proof}

We next show that for a planar multigraph~${H = (U,F)}$, the number of non-isomorphic separations of size~$t \in \tilde{\Oh}(\sqrt{|U|})$ is in~$2^{\Oh(\nicefrac{\psi(H)}{\log(\psi(H))})}$.
Afterwards, we show that this also directly implies that the number of such separations is in~$2^{\Oh(\nicefrac{|U|}{\log(|U|)})}$ when~$H$ is a simple graph (which we relied on in the previous subsection).

\begin{lemma}
    \label{lem:sigma}
    Given a planar multigraph~$H=(U,F)$ and a non-negative integer~$t \in \Oh(\psi(H)^{\frac{2}{3}})$, the number~$\sigma_t(H)$ of non-isomorphic separations of size at most~$t$ of~$H$ is in~$2^{\Oh(\nicefrac{\psi(H)}{\log(\psi(H))})}$.
    The separations can be enumerated in~$\Oh(\psi(H)^3 \sigma_t(H))$~time.
\end{lemma}
\begin{proof}
    We follow the proof strategy for simple planar graphs by Bodlaender, Nederlof, and van der Zanden~\cite{BNZ16} to count the number of non-isomorphic separations~$(X,Y)$ of~$H$ of size at most~$t$.
    There are~$\binom{|U|}{t}\leq \psi(H)^t \in \psi(H)^{\Oh(\psi(H)^{\nicefrac{2}{3}})}$ possible choices for the separator~$Z = X \cap Y$ and all subsets of~$U$ of size at most~$t$ can be enumerated in~$\Oh(\binom{|U|}{t}) \subseteq \Oh(\sigma_t(H))$ time as~$\sigma_t(H) \geq \binom{|U|}{t}$.
    Next, we compute the connected components of~$H[U \setminus Z]$ in~$\Oh(|U| + |F|) \subseteq \Oh(\psi(H))$ time.
    Let the \emph{\textbf{weight}} of a connected component~$C$ in~$H[U \setminus S]$ denote the number of edges with at least one endpoint in~$C$, that is, the number of edges between two vertices in~$C$ plus the number of edges between vertices in~$C$ and vertices in~$Z$.
    A connected component is \emph{\textbf{heavy}} if its weight is at least~$\nicefrac{\log(\psi(H))}{180}$ and it is \emph{\textbf{light}} otherwise.

    Next, we show how to construct a (simple, bipartite, and planar) graph~$H'$ with~$Z$ as the vertices on one side of the bipartition and a new vertex set~$B$ on the other side.
    We start with the graph~$H$, contract all connected components in~$H[U \setminus Z]$ into single vertices (the vertices in~$B$), remove all self loops and parallel edges, and remove all edges between two vertices in~$Z$.
    Note that the above can be computed in~$\Oh(\psi(H))$~time.
    It is known that there are at most~$5|Z| \leq 5t$ vertices in~$B$ of degree at least~$6$.
    Moreover, there are at most~$37|Z| \leq 37t$ subsets~$S$ of~$Z$ such that there exists a vertex in~$B$ whose neighborhood is exactly~$S$~\cite{Gaj13}.
    We say a connected component of~$H[U \setminus Z]$ is \emph{\textbf{highly connected}} if the corresponding vertex in~$H'$ has at least six neighbors in~$Z$.
    Connected components that are light and \emph{not} highly connected are \emph{\textbf{small}}.
    
    We will next show that the number of non-isomorphic small connected components is at most~$37t \psi(H)^{\nicefrac{1}{4}}$.
    Note that a small connected component~$C$ is characterized by the following five characteristics:
    (i) the simple graph~$G_C$ that is the result of removing parallel edges and self loops, (ii) the set of neighbors in~$Z$ that any vertex in~$C$ is adjacent to, (iii) for each vertex in~$C$ the information whether it contains at least one self loop, (iv) for each vertex~$v$ in~$C$ the neighbors of~$v$ in~$Z$, and (v) the multiplicity of each edge.
    A result by Amini, Fomin, and Saurabh~\cite{AFS12} states that the number of simple planar graphs with at most~$n$ vertices is at most~$2^{36n}$ and that they can be enumerated in~$\Oh(2^{36n})$ time.
    Since each considered connected component is light, it contains less than~$\nicefrac{\log(\psi(H))}{180}$ edges and since it is connected by definition, it contains at most~$\nicefrac{\log(\psi(H))}{180}$ vertices.
    The number of options for characteristic (i) is therefore at most~$2^{\nicefrac{36 \log(\psi(H))}{180}}$.
    Using the known upper bound on the number of possible neighborhoods~\cite{Gaj13}, the number of options for characteristic (ii) is at most~$37t$.
    The number of choices for (iii) and (iv) combined is then~$2^{\frac{6 \log(\psi(H))}{180}} = 2^{\nicefrac{\log(\psi(H))}{30}}$ as each vertex in~$C$ has~$2^6$ choices to be connected to each of the at most~$5$ neighbors in~$Z$ and itself (self loop).
    The number of ways to distribute the at most~$\nicefrac{\log(\psi(H))}{180}$ edges incident to vertices in~$C$ to the set of edges in~$G_C$ and the edges between vertices in~$C$ and vertices in~$Z$ such that each edge has multiplicity at least one~(the number of options for characteristic (v)) is equivalent to the number integer compositions\footnote{An integer composition of an integer~$n$ is an ordered sequence of positive integers summing to~$n$.} of the number~$\lfloor \nicefrac{\log(\psi(H))}{180} \rfloor$, which is upper-bounded by~$2^{\nicefrac{\log(\psi(H))}{180}}$~\cite{Stanley:1997vol1}.
    Combining all of the above, the number of non-isomorphic small connected components is at most~$2^{\nicefrac{36 \log(\psi(H))}{180}} \cdot 37t \cdot 2^{\nicefrac{\log(\psi(H))}{30}} \cdot 2^{\nicefrac{\log(\psi(H))}{180}} < 37t 2^{45\frac{\log(\psi(H))}{180}} = 37t\psi(H)^{\nicefrac{1}{4}} \in \Oh(\psi(H)^{\nicefrac{11}{12}})$ as~$t \in \Oh(\psi(H)^{\nicefrac{2}{3}})$.

    Given the separator~$Z = X \cap Y$ for unknown~$X$ and~$Y$, the number of highly connected connected components in~$H[U \setminus Z]$ is at most~$5t$~\cite{Gaj13}.
    We next analyze the number of heavy connected components.
    Note that the number of edges in~$H$ is at most~$\nicefrac{5\psi(H)}{2}$.
    Since each heavy connected component is incident to at least~$\nicefrac{\log(\psi(H))}{180}$ edges and there are no edges between two such connected components (because~$Z$ is a separator), there are at most~$\nicefrac{450\psi(H)}{\log(\psi(H)}$ heavy connected components in~$H[U \setminus Z]$.
    Hence, the number of possible partitions of highly connected and heavy connected components to either~$X$ or~$Y$ is at most~$2^{5t} 2^{\nicefrac{450\psi(H)}{\log(\psi(H))}} \in 2^{\Oh(\nicefrac{\psi(H))}{\log(\psi(H))})}$.

    We now combine all previous steps.
    The number of possible separators~$Z$ is in~$\psi(H)^{\Oh(\psi(H)^{\nicefrac{2}{3}})}$.
    The number of possible partitions of highly connected and heavy connected components to either~$X$ or~$Y$ is in~$2^{\Oh(\nicefrac{\psi(H))}{\log(\psi(H))})}$.
    The number of non-isomorphic connected components is in~$\Oh(\psi(H)^{\nicefrac{11}{12}})$ and we decide for each of them how many are added to~$X$.
    There are at most~$\psi(H)$ vertices in~$H$ and this also upper bounds the total number of connected components in~$H[U \setminus Z]$.
    Thus, the number of non-isomorphic separations of~$H$ is in
    \[\psi(H)^{\Oh(\psi(H)^{\nicefrac{2}{3}})} \cdot 2^{\Oh(\nicefrac{\psi(H)}{\log(\psi(H))})} \cdot \psi(H)^{\Oh(\psi(H)^{\nicefrac{11}{12}})} \subseteq 2^{\Oh(\nicefrac{\psi(H)}{\log(\psi(H))})}.\]

    It remains to analyze the running time for enumerating all non-isomorphic separations of~$H$.
    Note that~$|U| + |F| \in \Oh(\psi(H))$ by definition of~$\psi$.
    The time required to iterate over all possible separators~$Z$ of size at most~$t$ is in~$\Oh(\binom{|U|}{t})$.
    Computing the connected components of~$H[U \setminus Z]$ takes~$\Oh(\psi(H))$~time.
    Next, we show how to partition all connected components of~$H[U \setminus Z]$ such that for two connected components with vertex sets~$C_1$ and~$C_2$ it holds that~$C_1$ and~$C_2$ are in the same block of the partition if and only if~$G[Z \cup C_1]$ is isomorphic to~$G[Z \cup C_2]$.
    To this end, we first compute simple planar graphs~$G'_i$ for each connected component with vertex set~$C_i$ such that~$G[Z \cup C_i]$ and~$G[Z \cup C_j]$ are isomorphic if and only if~$G'_i$ and~$G'_j$ are isomorphic.
    For each pair of graphs, we then check whether they are isomorphic in~$\Oh(\psi(H))$~time~\cite{HW74}.
    Since the number of connected components is at most~$|U| \leq \psi(H)$, this takes~$\Oh(\psi(H)^3)$ time for all pairs combined and this immediately yields the desired partition.
    
    Note that~$\sigma_t(H)$ is precisely the sum of the number of combinations of how many connected components from each block of the partition are added to~$X$ over all separations~$Z$ of size at most~$t$.
    Moreover, for each possible combination, we can compute the separation~$(X,Y)$ in~$\Oh(\psi)$ time.
    Thus, the overall time required to compute all non-isomorphic separations of~$H$ is in~$\Oh(\psi(H)^3\sigma_t(H))$.
    This concludes the proof.
\end{proof}

Note that when~$H$ is a simple graph, then it is 5-degenerate and therefore~$\psi(H) = |U|$.
Thus, the lemma also directly implies \cref{lem:sigmasimple}, which we restate here for convenience.

\sigmasimple*

\begin{proof}
    We show that~$\psi(H) = |U|$ for any simple planar graph~$H=(U,F)$.
    The corollary then is directly implied by \cref{lem:sigma}.
    To this end, assume towards a contradiction that~$\psi(H) \neq |U|$.
    This implies that~$\psi(H) = \nicefrac{2|F|}{5}$ and $\psi(H) > |U|$.
    Note that any~$d$-degenerate simple graph with~$n$ vertices contains at most~$\nicefrac{5n}{2}$ edges.
    Hence, $|F| \leq \nicefrac{5|U|}{2}$.
    This yields~$\psi(H) \leq \nicefrac{2|F|}{5} \leq \nicefrac{2 \frac{5|U|}{2}}{5} = |U|$, a contradiction.
\end{proof}

We next define counting functions.
Intuitively, these return for each edge~$e$ in~$H[X \cap Y]$ a number~$b_e$ and we want to find a model in which exactly~$b_e$ parallel edges between the endpoints of~$e$ are modeled.
Formally, given a separations~$(X,Y)$ of a planar multigraph~$H$, let~$Z = X \cap Y$ and~$E_{Z}$ be the set of edges of~$H[Z]$ after removing parallel edges (but not self loops).
Then, a counting function~$\delta \colon E_{Z} \rightarrow \mathds{N} \cup \{0\}$ is any function such that~$\delta(e)$ is at most the number of parallel edges with endpoints~$e$ in~$H[X \cap Y]$.
We next analyze the number of possible counting functions for a planar multigraph~$H$.

\begin{lemma}
    \label{lem:countingfunctions}
    Let~$H=(U,F)$ be a planar multigraph and let~$Z \subseteq U$.
    The number of possible counting functions for~$E_Z$ is in~$2^{\Oh(\log(\psi(H))|U|)}$.
    All such counting functions can be enumerated in~$2^{\Oh(\log(\psi(H))|U|)}|U|$ time.
\end{lemma}

\begin{proof}
    Since~$H$ is planar, the graph resulting from removing parallel edges and self-loops is 5-degenerate and therefore contains at most~$\nicefrac{5|U|}{2}$ edges.
    Combined with the at most~$|U|$ possible self loops, this yields~$\nicefrac{7|U|}{2}$ possible edges.
    For each of these edges, a counting function can assign an integer between~$0$ and~$|F| \in \Oh(\psi(H))$, that is, the number of possible choices is in~${(\psi(H))^{\Oh(\nicefrac{7|U|}{2})} \subseteq 2^{\Oh(\log(\psi(H))|U|)}}$.
    We can iterate over all of them efficiently and for each check whether the multiplicity is at most the multiplicity in~$H$ in~$\Oh(|U|)$ time.
\end{proof}

We can now present our algorithm for \im{}.
Since it is only a slight adaptation of the subexponential parameterized algorithm for \minor, we do not present the full (lengthy) proof.
Instead, we give a high-level description featuring the key differences between the two algorithms and proofs.

\begin{theorem}
    \label{thm:inducedminor}
    Let~$t = 12000\sqrt{k}\log^8(k)$.
    For planar graphs, \im{} can be solved in~$2^{\tilde{\Oh}(\sqrt{k}\log(\psi(H)))}(\sigma_t(H))^2n^{2.49}$ time, where~$\sigma_t(H)$ denotes the number of non-isomorphic separations of~$H$ of size at most~$t$.
\end{theorem}

\begin{proof}[Proof sketch]
    A request is a tuple~$r=(\peb,(X,Y),c,P,\delta,\ell)$, where~$\peb$ is a pebble set, $(X,Y)$ is a separation of~$H$ of size at most~$12000\sqrt{d}\log^8(d)$, $c$ is a coloring function that assigns each vertex in~$\peb$ to a node in~$X \cap Y$ (formally each node in~$X \cap Y$ is assigned a non-empty subset of vertices in~$\peb$), $P$ is a partition of~$\peb$ that refines~$c$, $\delta$ is a counting function for~$H[X \cap Y]$ ($\alpha_{\pi}$ is again 1), and~$\ell \leq d$ is an upper bound on the number of vertices in a solution.
    For a request~$r$, let~$Z = X \cap Y$, let~$H^{\delta}_{Z}$ be the subgraph of~$H[Z]$ that has edge multiplicities as specified by~$\delta$, and let~$H^{\delta}_X$ be the subgraph of~$H[X]$ that has the same multiplicities as~$H$ for all edges incident to at least one node in~$X \setminus Y$ and multiplicities as in~$H^{\delta}_Z$ for all other edges.
    Then, an optimal solution for request~$r$ corresponds to a set of at most~$\ell$ vertices that induce a subgraph of maximum weight that models~$H^{\delta}_X$ exactly after adding the following edges: for each~$u \in Z$, choose one representative from each set in~$P^u$ and add a forest between all representatives for~$u$.
    We say that a graph models~$H_X^{\delta}$ exactly, if after contracting all vertices in the model of~$u$ along a spanning tree, the resulting edge multiplicities exactly coincide with~$H_X^{\delta}$.

    When combining two solutions along a separator~$C=(C_{\he},C_{\li},C_{\di})$, we iterate over all pebble sets~$\peb'$ that contain at most~$20 \sqrt{d} \log(d)$ vertices from~$C_{\li}$ and no vertices from~$C_{\di}$.
    The procedure is similar to the case of \minorr{} with two exceptions.
    First, when combining two solutions, we can no longer assume that the two subgraphs are edge-disjoint.
    Instead, we subtract the weight of all edges induced by~$\peb'$ from the total sum as these are precisely the edges that are counted twice.
    Second, we need to compute which graph the union of two graphs models exactly.
    Here, we make a case distinction between self loops and other edges.
    For self loops, we use a slight adaptation of \cref{lem:acyclic} to compute the number of self loops (the feedback edge number) of the union of two forests represented by two partitions.
    For each other edge~$e = \{u_1,u_2\}$ between two nodes in~$Z$, it is easy to verify that the union of two graphs models exactly~$\delta_1(e) + \delta_2(e) - g(e)$ parallel edges with endpoints~$e$, where~$\delta_1(e)$ is the number of parallel edges modeled by the first graph, $\delta_2(e)$ is the number of edges modeled by the second graph, and~$g(e)$ is the number of edges induced by~$\peb'$ that have one endpoint~$v_1 \in c_1(u_1)$ and one endpoint~$v_2 \in c_1(u_2)$ (note that we require that~$v \in c_1(u)$ if and only if~$v \in c_2(u)$ for all~$u \in U$ and~$v \in \peb'$).
    Finally, we only consider counting functions for separations with at most~$t$ nodes and thus \cref{lem:countingfunctions} allows us to enumerate them all in~$2^{\Oh(\log(\psi(H))\sqrt{d}\log^8(d)}$ time.
    Choosing~$\gamma_H(d) \in 2^{\Oh(\log(\psi(H))\sqrt{d} \log^{10}(d))} (\sigma_t(H))^2$ allows us to also iterate over all pairs of counting functions.
    The rest of the proof is nearly identical to the proof of \cref{thm:minor}.
\end{proof}

Combined with \cref{lem:psi,lem:sigma}, this also shows that \im{} can be solved in~$2^{\Oh(\nicefrac{k}{\log(k)})}n^{2.49}$ time.

\begin{corollary}
    \label{cor:inducedminor}
    For planar input graphs, \im{} can be solved in~$2^{\Oh(\nicefrac{k}{\log k})}n^{2.49}$ time.
\end{corollary}

\begin{proof}
    We simply compute~$\psi(H)$ in~$\Oh(|U| + |F|)$ time.
    If~$\psi(H) > k$, then we return no.
    This is correct by \cref{lem:psi}.
    Otherwise, the running time presented in \cref{thm:inducedminor} is upper bounded by~$2^{\Oh(\nicefrac{k}{\log k})}n^{2.49}$ by \cref{lem:sigma}.
\end{proof}
    
Similar to the previous subsection, \wisi{} is the special case of \im{} where~$k = |U|$.
Note that if~$H$ is not a simple graph, then the answer is always no in this case.
Moreover, \wim{} is a special case of \wisi{} where the number of non-isomorphic separations of size at most~$12000 \sqrt{k} \log^8(k)$ is in~$2^{\tilde{\Oh}(\sqrt{k})}$.
This yields the following.

\begin{corollary}
    For planar input graphs, \wisi{} can be solved in~$2^{\Oh(\nicefrac{|U|}{\log(|U|)})}n^{2.49}$ time and in~$2^{\Oh(\nicefrac{k}{\log(k)})}n^{2.49}$ time.
    \wim{} can be solved in~$2^{\tilde{\Oh}(\sqrt{k})}n^{2.49}$ time for planar input graphs.
\end{corollary}

\section{Conclusion}

We introduced a general framework for developing subexponential parameterized algorithms for problems on planar graphs that ask for a set of at most~$k$ vertices satisfying some property.
The framework is applicable to non-bidimensional problems where the solution can induce a graph with many connected components and it can handle directed graphs and weighted problems as well.
We showed how to apply the framework to a selection of problems that were previously unknown to have subexponential parameterized algorithms for planar input graphs.
In the case of (unweighted) \den, this resolves an open problem from the literature.
Motivated by our positive result for \wif, we conclude with two open problems of our own.
Note that \wif{} can also be defined as finding a maximum-weight induced subgraph of size at most~$k$ that does not contain a triangle as a minor.
We ask the following.
\begin{quote}
Can our algorithm be extended to find an (induced) subgraph of size~$k$ that does not contain an arbitrary but fixed graph~$H$ as a minor?    
\end{quote}
\begin{quote}
Is there a subexponential parameterized algorithm to find an (induced) subgraph of size~$k$ that excludes all graphs from a fixed family~$\mathcal{F}$ as minors?    
\end{quote}
Note that if we replace minor in the above questions by subgraph, then the question was recently settled in the affirmative~\cite{LPSXZ25}.

\medskip
{\noindent\textbf{Acknowledgements.} We thank anonymous reviewers of previous versions of this paper for many helpful comments that helped significantly improve the clarity of the paper. They also pointed us to some additional related work like the existing subexponential parameterized algorithm for \textsc{Weighted Independent Set}~\cite{masterthesis}.} 

\bibliographystyle{alphaurl}
\bibliography{ref}

\end{document}